\documentclass[11pt, a4paper]{article}
\usepackage{amsmath, amssymb, amsthm, amsfonts, mathrsfs}
\usepackage{geometry}
\usepackage{hyperref}
\usepackage{authblk}
\usepackage{graphicx}
\usepackage{array}
\usepackage{titlesec}
\usepackage[utf8]{inputenc}
\usepackage[T1]{fontenc}
\usepackage{newtxtext, newtxmath} 
\usepackage{enumitem}
\usepackage{braket}
\usepackage{tikz}
\usetikzlibrary{cd, positioning, arrows.meta}
\usepackage{mathtools}
\usepackage{xcolor}
\usepackage{comment}
\usepackage{algorithm}
\usepackage{algpseudocode}

\usepackage[backend=biber,style=numeric]{biblatex}
\hypersetup{
    colorlinks=true,
    linkcolor=blue,
    filecolor=magenta,      
    urlcolor=cyan,
    citecolor=red,
}

\theoremstyle{definition}
\newtheorem{theorem}{Theorem}[section]
\newtheorem{lemma}[theorem]{Lemma}
\newtheorem{proposition}[theorem]{Proposition}
\newtheorem{corollary}[theorem]{Corollary}
\newtheorem{definition}[theorem]{Definition}
\newtheorem{example}[theorem]{Example}
\newtheorem{remark}[theorem]{Remark}

\newcommand{\mf}[1]{\mathfrak{#1}}

\DeclareMathOperator{\Tr}{Tr}
\DeclareMathOperator{\ad}{ad}
\DeclareMathOperator{\Ad}{Ad}
\DeclareMathOperator{\Id}{Id}
\DeclareMathOperator{\Span}{span}
\DeclareMathOperator{\Ker}{Ker}
\DeclareMathOperator{\Image}{Im}

\DeclareMathOperator{\End}{End}

\DeclareMathOperator{\cone}{cone}
\DeclareMathOperator{\aff}{\mathfrak{aff}}
\DeclareMathOperator{\Lie}{Lie}

\newcommand{\Lcal}{\mathcal{L}}
\newcommand{\Hcal}{\mathcal{H}}
\newcommand{\Bcal}{\mathcal{B}}

\newcommand{\Ecal}{\mathcal{E}}

\newcommand{\Dcal}{\mathcal{D}}

\newcommand{\g}{\mathfrak{g}}
\newcommand{\s}{\mathfrak{s}}
\newcommand{\rr}{\mathfrak{r}}

\newcommand{\gl}{\mathfrak{gl}}
\newcommand{\sll}{\mathfrak{sl}}

\newcommand{\su}{\mathfrak{su}}
\newcommand{\so}{\mathfrak{so}}
\DeclareMathOperator{\rad}{rad}
\DeclareMathOperator{\nil}{nil}

\newcommand{\C}{\mathbb{C}}
\newcommand{\R}{\mathbb{R}}

\newcommand{\WGKLS}{\mathcal{W}_{\text{GKLS}}}

\newcommand{\Wrad}{W_{\rr}}
\newcommand{\Cquot}{C_{\text{quot}}}

\newcommand{\gLK}{\mf{g}_{LK}}
\newcommand{\Htr}{\Hcal_{0}} 

\begin{document}

\title{The Geometry of Dissipative Complexity: A Levi-Type Decomposition Theorem for Markovian Quantum Dynamics via Lie Wedges and Invariant Cones}

\author[1]{Yibin Wang}
\author[2]{Dhara Thakkar}

\affil[1]{Graduate School of Mathematics, Nagoya University, 464-8601 Nagoya, Aichi, Japan.\\
\texttt{yibinw0210@gmail.com}}
\affil[2]{Graduate School of Mathematics, Nagoya University, 464-8601 Nagoya, Aichi, Japan.\\
\texttt{dharathakkar754@gmail.com}}

\date{}

\maketitle

\begin{abstract}
Lie algebras describe how control Hamiltonians combine in closed quantum systems. In open Markovian systems, however, dissipation introduces irreversible directions that the Lie algebra alone does not retain. A dynamical Lie wedge preserves this information as a convex cone of locally admissible generators. The classical Levi decomposition separates a finite-dimensional Lie algebra into semisimple and solvable parts. In this work, we prove a corresponding decomposition and reconstruction theorem for dynamical Lie wedges.

The theorem decomposes the wedge into semisimple and solvable data, records how these components are coupled, and provides a converse reconstruction of the wedge. The framework yields four structural types of the generated algebra. We further prove that the total dissipation strength of every admissible generator depends only on its coordinate in the solvable radical. In particular, if the generated Lie algebra is semisimple, every admissible generator is Hamiltonian. Together, these results extend Lie-algebraic structural methods to irreversible Markovian control and clarify how dissipation is encoded in the generator geometry.
\end{abstract}

\newpage
\setcounter{tocdepth}{3}
\tableofcontents

\section{Introduction: From Dynamical Lie Algebras to Dynamical Lie Wedges}\label{sec:introduction}

Symmetry and dynamics meet in the Dynamical Lie Algebra (DLA) framework. For a closed quantum system, Lie algebras and their representations provide the language \cite{wigner1939unitary,hall2013quantum}, and the available Hamiltonians generate a real Lie algebra under commutation. This algebra determines controllability \cite{altafini2001controllability} and supplies structural data for quantum and geometric-control analysis \cite{d2021introduction,jurdjevic1997geometric}. Group-manifold methods have underpinned geometric control theory since its inception \cite{brockett1972system,sussmann1972controllability} and now inform classifications such as those for dynamical Lie algebras of $2$-local spin systems \cite{wiersema2024classification}.

The mathematical cornerstone is the Levi Decomposition Theorem for finite-dimensional Lie algebras \cite{knapp1996lie,varadarajan2013lie}. It identifies the solvable radical $\mf{r}$ as the unique maximal solvable ideal and guarantees a semisimple Levi complement $\mf{s}$. After choosing such a complement,
\begin{equation}
\g=\mf{s}\ltimes\mf{r}.
\end{equation}
The radical and quotient $\g/\mf r$ are intrinsic, whereas the embedded complement is choice-dependent; different choices are related by Levi--Malcev conjugacy. For physical generators, this decomposition supplies structural data. Corollary~\ref{cor:hamiltonian_rigidity}(2) shows that a semisimple system algebra has a Hamiltonian-only wedge and no synthesizable dissipation. In mixed systems, dynamical quantities such as mixing, thermalization, steady-state uniqueness, and the dissipative gap require information beyond the radical and semisimple factor.

\paragraph*{Why Irreversibility Calls for a Lie Wedge}

Carrying the DLA framework to open systems runs into a structural property of the generator set. The Markovian generators do not form a Lie algebra. Hamiltonian generators close under commutation to form $\mf{su}(d)$, whereas the physically admissible dissipative generators $\WGKLS$ fail this closure. Dirr et al.\ \cite{dirr2009lie} establish the form of the non-closure. Their Corollary~III.3 shows that $\WGKLS$ is a Lie wedge but not a Lie semialgebra, so the Lie bracket $[\Lcal_1, \Lcal_2]$ of two valid Lindbladians can violate the Conditional Complete Positivity (CCP) condition.

\begin{example}[The Closure Paradox: Qutrit Violation]
\label{ex:qutrit_closure}
Consider a simple qutrit system ($d=3$). The commutator of two physically valid rank-1 dissipative generators, $\Lcal_1$ (generated by $L_1 = \lambda_1 + i\lambda_4$) and $\Lcal_2$ (generated by $L_2 = \lambda_2 + i\lambda_5$), has an indefinite Kossakowski form; equivalently, its conditional Choi compression has a negative eigenvalue.\footnote{This example uses the unnormalized Gell-Mann convention $\Tr(\lambda_i\lambda_j)=2\delta_{ij}$, while the default convention in Eq.~\eqref{eq:gkls_kossakowski} is Hilbert--Schmidt orthonormal. In the raw $\lambda$ basis the commutator has Kossakowski spectrum $\{-2,2,0,\dots\}$; after passing to $F_i=\lambda_i/\sqrt{2}$ the Kossakowski spectrum is $\{-4,4,0,\dots\}$. The conditional Choi compression has minimum eigenvalue $-4$ under the Choi normalization used here.} It therefore violates Conditional Complete Positivity and lies outside $\WGKLS$. This counterexample is established here by direct computation in the Gell-Mann basis.
\end{example}

The non-closure is the structural reason a classification cannot rest on the generated algebra alone. The algebraic operation of commutation pulls against the geometric constraint of complete positivity, so a faithful classification has to carry both. We do this by passing from the Lie algebra to the \textit{Dynamical Lie Wedge Pair}, which holds the geometry of admissibility alongside the algebra in the presence of irreversibility.
\paragraph*{Structural Positioning}

The study of open quantum systems spans several decades, yet a structural classification analogous to the DLA is still missing. Work on Markovian dynamics has largely addressed particular phenomena and has not produced a unified generator-level structure. To the best of our knowledge, a systematic classification framework grounded in the Levi decomposition structure of the dynamically generated Lie algebra, which explicitly distinguishes three distinct classes: the purely solvable type, the semisimple type (corresponding to the Hamiltonian-only scenario), and the mixed Levi type, has not been established in existing literature. Geometric approaches have instead probed response properties and the space of Lindbladians through differential geometry and Fisher information \cite{albert2016geometry}, without supplying this algebraic classification of complexity.

Lie semigroup theory and geometric control supply the mathematical language. Irreversible evolution forms a semigroup, whose infinitesimal counterpart is a Lie wedge \cite{hilgert1989lie}. For quantum control, the GKLS generators form such a wedge. This viewpoint has supported reachability and controllability analyses \cite{dirr2009lie,kurniawan2012controllability,cai2025lindbladians}, explicit studies of coherently controlled unital systems \cite{OMeara2011IllustratingTG,schulteherbruggen2017kossakowski}, and the reduction of a single fixed Lindbladian under fast unitary control to a system on the eigenvalue simplex \cite{malvetti2024reachability}. Those results retain their stated control assumptions.

We classify the algebraic and local wedge structure attached to a specified generating set; Remark~\ref{rem:scope_control_systems} fixes the input. Reachability from a selected state is the corresponding global question. O'Meara identified the Lie wedge, distinct from the algebra generated by the Lindbladians, as the relevant open-system object \cite{o2025lie}. The same work shows why Lie semialgebras are too rigid for a general treatment: the stated semialgebra condition holds only when coherent controls affect neither the drift Hamiltonian nor the incoherent component. A systematic classification based on the internal wedge geometry and a Levi splitting has not been developed.

\paragraph*{The Dynamical Lie Wedge Pair}

This paper develops structural tools that combine geometric control, Lie-algebra structure, and invariant cones. The ambient algebra of physical generators is non-reductive, unlike the reductive algebra $\mf{u}(d)$ underlying unitary dynamics, as Theorem~\ref{thm:structure_glk} and Corollary~\ref{cor:glk_nonreductive} establish. Neeb's theory provides a qualified comparison with cones invariant under the full algebra. The reconstruction theorem follows directly from fiber data and Lie-wedge conditions.

The quantum-control literature separates the Lie algebra generated by brackets from the geometry of physical admissibility. The system Lie algebra and system Lie wedge of \cite{dirr2009lie,OMeara2011IllustratingTG,o2025lie} are those two objects. We package them as the Dynamical Lie Wedge Pair (DLWP)
\[
S(G)=(\g(G),W(G))
\]
for a specified set $G$ of available physical generators. Here $\g(G)$ is the smallest real Lie algebra containing $G$; it includes formal bracket directions that may fail CCP and need not be reachable by positive-time sequences. Operational reachability is encoded separately by $\Sigma(G)$ and its tangent wedge $L(\Sigma(G))$, with inverse evolutions available only along reversible directions. The cone $W(G)$ is the smallest closed Lie wedge containing $G$. It lies in $\WGKLS$ and is the minimal closed-wedge inner approximation to $L(\Sigma(G))$; equality holds precisely when $W(G)$ is global.

\paragraph*{Structural Results}

The Levi-Type Decomposition Theorem of Section~\ref{sec:levi_type_theorem} gives a fibered normal form relative to a chosen Levi splitting. It re-encodes the wedge through its radical slice, quotient cone, and fibers, and reconstructs the wedge from data satisfying the Lie-wedge conditions. The normal form builds on semigroup-theoretic reconstruction methods for invariant cones \cite{hilgert1989lie}; the exact minimal-wedge construction is a separate result used by the structural program.

After choosing $\g=\mf{s}\ltimes\mf{r}$, the radical and quotient data remain intrinsic, while the embedded Levi coordinates and fibers are relative to that choice. The components have the following interpretations:
\begin{enumerate}
    \item The radical wedge $W_{\rr}=W\cap\mf r$ is a Lie wedge in the solvable radical. Pure dephasing and cascaded decay give examples with solvable system algebras. Decay channels, conserved quantities, and mixing behavior depend on additional represented generator data.
    \item The quotient cone $C_{\mathrm{quot}}=\pi_{\s}(W)$ records quotient coordinates that admit full lifts. Physical generators occur as complete lifts, while mixing, thermalization, steady-state uniqueness, and scrambling require dynamical information beyond the presence or dimension of this cone.
    \item The action of $\mf s$ on $\mf r$ constrains the fibers, and the full wedge is encoded by $W_{\rr}$, $C_{\mathrm{quot}}$, and the complete fiber map. The FCC is edge invariance in these coordinates; complete positivity and spectral stability are tested independently.
\end{enumerate}

Lie-semigroup theory supplies Lie wedges, tangent calculus, and reconstruction methods for invariant cones \cite{hilgert1989lie}; geometric control supplies the pair $(\g(G),W(G))$ \cite{dirr2009lie,OMeara2011IllustratingTG,o2025lie}. This paper proves the exact iterative construction of the minimal closed wedge and the chosen-splitting fiber re-encoding with a matching reconstruction theorem. Corollary~\ref{cor:hamiltonian_rigidity} places the dissipation strength of every element of $W$ in its radical coordinate. The comparison with Neeb's theory applies to bracket-induced symplectic nilradicals, while convex type is vacuous for $\gLK$ itself.

The decomposition and reconstruction are proved directly from closed convex-process data and edge invariance. Neeb's full-algebra invariant-cone theory \cite{neeb1994classification,Neeb2000}, following Vinberg \cite{vinberg1963theory} and Olshanskii \cite{olshanskii1981invariant}, supplies the later comparison. Its convex-type condition combines a moment-map cone requirement with a condition on the center action for a bracket-induced symplectic nilradical module. The comparison therefore applies when the relevant module exists and the invariance domains agree. For $\gLK$, the nilradical is abelian and the condition is vacuous.

This work provides an algebraic framework for classifying Markovian dissipative control systems by solvable type, semisimple type (Hamiltonian-only), and mixed Levi type. Section~\ref{sec:preliminaries} establishes the ambient affine algebra and GKLS wedge. Section~\ref{sec:dlwp_framework} defines the DLWP and constructs its minimal closed wedge. Section~\ref{sec:levi_type_theorem} proves the fibered normal form and reconstruction and gives counterexamples and examples, including the detailed non-reductive system of Section~\ref{sec:example_non_reductive}. Section~\ref{sec:neebs_theory} gives the FCC--Neeb comparison. The appendices contain longer proofs and the invariant-cone background.

\section{Mathematical Preliminaries and the Algebraic Structure of GKLS Generators}
\label{sec:preliminaries}

This section collects the background the rest of the paper draws on. Alongside standard material on Lie algebra structure theory, convex geometry, and Lie wedges, it contains the results specific to open quantum dynamics that later sections lean on, namely the structure theorem and non-reductivity of the Lindblad--Kossakowski algebra $\gLK$ together with the Lie wedge structure of the GKLS cone $\WGKLS$. Readers familiar with Lie theory may skim Sections~\ref{sec:lie_structure_prelim} and~\ref{sec:convex_lie_wedges_prelim}.

\subsection{Lie Algebra Structure Theory}
\label{sec:lie_structure_prelim}

We review the concepts from Lie theory that the classification program rests on, with emphasis on non-reductive structure and the Levi decomposition.

\paragraph*{Basic Definitions and the Adjoint Representation}

Throughout, $\g$ is a finite dimensional real Lie algebra, a vector space with a bilinear alternating bracket $[\cdot, \cdot]$ obeying the Jacobi identity $[X, [Y, Z]] + [Y, [Z, X]] + [Z, [X, Y]] = 0$. An ideal $\mf{i}$ satisfies $[\g, \mf{i}] \subseteq \mf{i}$ and a subalgebra is closed under the bracket; the ideals are the invariant subspaces of the adjoint action and carry the structural decomposition. That action is the adjoint representation $\ad: \g \to \text{Der}(\g) \subset \End(\g)$, $\ad(X)(Y) = [X, Y]$, each $\ad(X)$ a derivation by the Jacobi identity. Its kernel is the center $Z(\g) = \{X \in \g \mid [X, Y] = 0 \text{ for all } Y \in \g\}$, always an ideal. Exponentiating gives the bracket-preserving automorphisms $e^{\ad(X)}$, which transport cones along the flow of the edge.

\paragraph*{Solvability, Nilpotency, and the Structure of Radicals}

Solvability and nilpotency organize algebras through their commutator series. The derived series $\g^{(0)} = \g$, $\g^{(k+1)} = [\g^{(k)}, \g^{(k)}]$ makes $\g$ solvable when $\g^{(k)} = \{0\}$ for some $k$, and the lower central series $\g^{0} = \g$, $\g^{k+1} = [\g, \g^{k}]$ makes it nilpotent when $\g^{k} = \{0\}$; nilpotency implies solvability but not the reverse. Under its representation-theoretic hypotheses, Lie's theorem turns solvability into a complete invariant flag.

\begin{theorem}[Lie's Theorem]
Let $\g$ be a finite-dimensional solvable complex Lie algebra and let $V$ be a nonzero finite-dimensional complex representation space. There is a basis of $V$ in which all represented elements of $\g$ are upper triangular simultaneously; equivalently, $V$ has a complete $\g$-invariant flag.
\end{theorem}

Over the reals one complexifies first, $\g_{\C} = \g \otimes_{\R} \C$, preserving solvability; the resulting invariant flag is the algebraic content of the Introduction's solvable type. The radical $\mf{r} = \rad(\g)$ is the maximal solvable ideal and the nilradical $\mf{n} = \nil(\g)$ the maximal nilpotent one, with $\mf{n} \subseteq \mf{r} \subseteq \g$ and $[\mf{r}, \mf{r}] \subseteq \mf{n}$, so $\mf{r}/\mf{n}$ is abelian \cite{dixmier1996enveloping} (Appendix~\ref{app:nonreductive_structure}). The nilradical enters Neeb's theory of invariant cones (Section~\ref{sec:invariant_cones_prelim}); its symplectic modules constrain how the adjoint action coexists with convexity, and its nilpotency class is an algebraic invariant of the radical.

\subsubsection{Semisimplicity and the Levi Decomposition}

Semisimplicity is the complementary algebraic case. A Lie algebra is semisimple when $\rad(\g) = \{0\}$, equivalently when the Killing form $B(X, Y) = \Tr(\ad(X) \circ \ad(Y))$ is nondegenerate (Cartan's criterion), and simple when it is nonabelian with no non-trivial ideals. A semisimple algebra is an essentially unique direct sum of simple ideals $\g = \bigoplus_i \mf{s}_i$. This decomposition is algebraic and does not by itself imply mixing, thermalization, or scrambling in a selected dynamics.

The structure theorem organizing these concepts is the Levi decomposition \cite{knapp1996lie}.

\begin{theorem}[Levi Decomposition]
Any finite dimensional real Lie algebra $\g$ can be written as a semidirect sum (Levi decomposition):
\begin{equation}
\g = \mf{s} \ltimes \mf{r},
\end{equation}
where $\mf{r} = \rad(\g)$ is the radical, and $\mf{s}$ is a semisimple subalgebra (the Levi factor). The Levi factor is unique up to conjugation by an inner automorphism of the form $e^{\ad(z)}$ where $z$ belongs to the nilradical $\mf{n}$ (Malcev Harish-Chandra theorem \cite{jacobson2013lie}).
\end{theorem}

The semidirect structure is set by the action of $\mf{s}$ on $\mf{r}$, a homomorphism $\rho: \mf{s} \to \text{Der}(\mf{r})$ with $\rho(S)(R) = [S, R]$. This $\rho$ is the algebraic coupling between the selected Levi factor and the radical.

With $\g$ identified with $\mf{s} \oplus \mf{r}$, the bracket of $X_i = S_i+R_i$ in $\g = \mf{s} \ltimes_{\rho} \mf{r}$ is
\begin{equation}
\label{eq:semidirect_bracket}
[S_1+R_1, S_2+R_2] = \underbrace{([S_1, S_2])}_{\in \mf{s}} + \underbrace{([R_1, R_2] + \rho(S_1)(R_2) - \rho(S_2)(R_1))}_{\in \mf{r}}.
\end{equation}
The $\rho$ terms show how the semisimple component modulates the solvable one, the coupling that makes non-reductive algebras hard.

\subsubsection{Reductive and Non-Reductive Algebras}

The coupling $\rho$ fixes the global structure, and the split between reductive and non-reductive algebras is central to this work.

\begin{definition}[Reductivity]
A Lie algebra $\g$ is reductive if its radical equals its center, $\rad(\g) = Z(\g)$, equivalently if the adjoint representation is completely reducible (Weyl's theorem \cite{Humphreys2012}).
\end{definition}

When $\g$ is reductive the radical is abelian, since $[\mf{r}, \mf{r}] \subseteq \mf{n}$ and $\mf{r}=Z(\g)$ force $[\mf{r}, \mf{r}]=0$. The decomposition becomes a direct sum $\g = \mf{s} \oplus Z(\g)$ with trivial $\rho$, so the semisimple and central parts decouple.

Unitary dynamics live in subalgebras of $\mf{u}(d) = \mf{su}(d) \oplus \R \cdot iI$, which is reductive, so the closed-system DLA framework reduces to the semisimple part $\mf{su}(d)$. This is why that framework is tractable.

A Lie algebra is non-reductive if $\rad(\g) \neq Z(\g)$. This occurs when the adjoint representation is not completely reducible. This structure arises in two ways:
\begin{enumerate}
    \item The radical $\mf{r}$ is non abelian, which happens when the nilradical $\mf{n}$ is non-trivial and $\mf{r}$ acts non-trivially on it.
    \item The radical $\mf{r}$ is abelian, but $\rho$ acts non-trivially, so the radical is an ideal lying outside the center.
\end{enumerate}

\paragraph{The Prevalence of Non-Reductivity.}
Such algebras are common in physics, the Euclidean and Poincar\'e algebras among them (Appendix~\ref{app:nonreductive_structure}). The ambient Markovian algebra $\gLK$ is itself non-reductive (Section~\ref{sec:structure_glk}), and System Lie Algebras can inherit the feature when the induced Levi action on the radical is non-trivial or when the radical is non-abelian. The coupling $\rho$ enters the FCC and constrains the fiber map, but it does not determine the fibers.

\subsection{Markovian Dynamics, the Algebra $\g_{LK}$, and its Non-Reductivity}
\label{sec:structure_glk}

We review the formalism of Markovian open quantum systems and identify the ambient Lie algebra in which the classification takes place.

\paragraph*{Operator Spaces and Superoperators}

We consider a system with Hilbert space $\Hcal$ ($\dim(\Hcal)=d$). The space of bounded linear operators is $\Bcal(\Hcal)$, isomorphic to the matrix algebra $M_d(\C)$. The state is given by a density operator $\rho \in \Bcal(\Hcal)$, satisfying $\rho \ge 0$ (positive semidefinite) and $\Tr(\rho)=1$. The space of Hermitian operators is denoted $\Bcal_H(\Hcal)$, a real vector space of dimension $d^2$.

Dynamics are described by superoperators $\Lcal \in \End(\Bcal(\Hcal))$. We work in the Liouville space representation, where operators are viewed as vectors (vectorization) and superoperators as matrices acting on these vectors. The space is often equipped with the Hilbert Schmidt inner product $\langle A, B \rangle = \Tr(A^\dagger B)$. The dimension of the space of superoperators is $d^4$.

\subsubsection{The Ambient Lie Algebra of Physical Dynamics $\g_{LK}$}
\label{sec:glk_definition}

We are interested in dynamics that preserve the essential properties of the density operator. The fundamental physical constraints defining the space of infinitesimal generators are:
\begin{enumerate}
    \item \textbf{Hermiticity Preservation (HP):} $\Lcal(\rho)^\dagger = \Lcal(\rho)$ if $\rho^\dagger = \rho$. This ensures that the evolution preserves the reality of observables.
    \item \textbf{Trace Annihilation (TA):} $\Tr(\Lcal(\rho)) = 0$. This infinitesimal condition ensures the conservation of probability for the integrated semigroup. (The evolution $\Ecal_t = e^{t\Lcal}$ satisfies $\Tr(\Ecal_t(\rho)) = \Tr(\rho)$).
\end{enumerate}

These two conditions are linear, and they are conditions on a \emph{generator}: a trace-annihilating map sends all of $\Bcal_H(\Hcal)$ into the traceless subspace, so it is never invertible. The HPTA maps therefore form a Lie algebra, and the group whose Lie algebra it is consists of the invertible Hermiticity-preserving, trace-\emph{preserving} maps, a closed subgroup of $GL(\Bcal_H(\Hcal))$. The structure of these maps and the associated constraints have been studied extensively since the foundational work on positive maps.

\begin{definition}[Lindblad Kossakowski Lie Algebra $\g_{LK}$]
The Lie algebra $\g_{LK}$ is the real Lie algebra consisting of all Hermiticity-preserving and trace-annihilating superoperators on $\Bcal(\Hcal)$.
\end{definition}

The space of Hermitian operators has dimension $d^2$, and the space of endomorphisms of $\Bcal_H(\Hcal)$ has dimension $d^4$; the constraints above cut this down. The structure theorem below builds an explicit isomorphism onto an affine algebra, and the exact count $\dim \g_{LK} = (d^2-1)^2 + (d^2-1) = d^4 - d^2$ is read off from it. The GKLS generators $\WGKLS$, introduced in Section~\ref{sec:structure_gkls}, form a cone inside $\g_{LK}$.

\subsubsection{The Structure Theorem: $\gLK$ is an Affine Algebra}

How hard the cone geometry of open dynamics can get is decided by the structure of the ambient algebra. The following theorem identifies that structure exactly.

\begin{theorem}[Structure of $\g_{LK}$ (affine isomorphism due to O'Meara \cite{o2025lie})]
\label{thm:structure_glk}
The Lindblad Kossakowski Lie algebra $\g_{LK}$ is isomorphic to the Lie algebra of the affine group acting on the space of traceless Hermitian operators,
\begin{equation}
\g_{LK} \cong \aff(\Htr) = \gl(\Htr) \ltimes \Htr,
\end{equation}
where the translational part $\Htr$ is an abelian ideal \cite{o2025lie}.
\end{theorem}

\begin{proof}[Proof sketch]
Trace annihilation forces $\Image(\Lcal) \subseteq \Htr$ for every $\Lcal \in \g_{LK}$, so the assignments $A_\Lcal = \Lcal|_{\Htr}$ and $b_\Lcal = \Lcal(I)$ are well defined, and $\Phi(\Lcal) = (A_\Lcal, b_\Lcal)$ maps $\g_{LK}$ into $\aff(\Htr)$. The verification that $\Phi$ is a bijective Lie algebra homomorphism is a direct computation, given in Appendix~\ref{app:proof_structure_glk}.
\end{proof}

The affine picture is the generalized Bloch representation. Writing a trace-one state as $\rho=I/d+x$, with $x\in\Htr$, gives
\[
\dot x=A_\Lcal x+\frac{1}{d}b_\Lcal,
\qquad b_\Lcal=\Lcal(I).
\]
Thus $b_\Lcal$ records non-unital drift and is nonzero exactly when the maximally mixed state is not stationary under the generator. Cooling or purification requires additional assumptions about the model, the state, and the relevant energy ordering.

The isomorphism is O'Meara's \cite{o2025lie}. The following consequence, which anchors the use of non-reductive cone theory throughout this paper, is established here.

\begin{corollary}[Non-reductivity of $\g_{LK}$]
\label{cor:glk_nonreductive}
For $d \ge 2$,
\begin{equation}
\rad(\g_{LK}) \cong \R \cdot \Id_{\Htr} \ltimes \Htr, \qquad Z(\g_{LK}) = \{0\},
\end{equation}
so $\rad(\g_{LK}) \neq Z(\g_{LK})$ and $\g_{LK}$ is non-reductive.
\end{corollary}

\begin{proof}
We argue in the affine model of Theorem~\ref{thm:structure_glk}. The linear part splits as $\gl(\Htr) = \sll(\Htr) \oplus \R \cdot \Id_{\Htr}$, with $\sll(\Htr)$ semisimple for $d \ge 2$. Since the translational ideal $\Htr$ is abelian and the action of the semisimple part of $\gl(\Htr)$ on it is faithful, the radical of the semidirect product is $\rad(\gl(\Htr)) \ltimes \Htr = \R \cdot \Id_{\Htr} \ltimes \Htr$, of dimension $d^2$; the product is semidirect and not direct, since $[\Id_{\Htr}, b] = b \neq 0$ for $b \neq 0$, so the radical is non-abelian with nilradical $\Htr$ (see \textit{Lie Groups and Lie Algebras}, Ch.~1, \S\S 5--6 \cite{bourbaki1989lie}). For the center, suppose $(A, b)$ commutes with every $(A', b')$. Vanishing of $[A, A']$ for all $A' \in \gl(\Htr)$ forces $A = \lambda \Id_{\Htr}$ by Schur's lemma, because the standard representation of $\gl(\Htr)$ on $\Htr$ is irreducible. The choice $b' = 0$ then gives $A' b = 0$ for every $A'$, hence $b = 0$ by faithfulness of the action, and finally $\lambda b' = 0$ for all $b' \in \Htr$ gives $\lambda = 0$. The center is trivial, and the corollary follows.
\end{proof}

The affine coupling makes $\gLK$ non-reductive and separates open dynamics from the closed case, where the ambient algebra $\mf{u}(d)$ is reductive (Section~\ref{sec:lie_structure_prelim}). In $\gLK$ the linear part shears the translational part, so convex-cone analysis must account for that coupling. Section~\ref{sec:invariant_cones_prelim} gives the qualified comparison with Neeb's full-algebra invariant-cone theory, while the reconstruction proof uses only the Lie-wedge data.

\subsection{Convex Geometry, Lie Semigroups, and Lie Wedges}
\label{sec:convex_lie_wedges_prelim}

When evolution is irreversible it generates a semigroup. Lie semigroup theory is the right framework, and its infinitesimal counterpart is the Lie wedge \cite{hilgert1989lie}, which links the Lie algebra to the convex geometry that irreversibility imposes.

\paragraph*{Cones, Wedges, and Duality}

We fix the convex-geometry notation the Lie wedge needs. A cone $W$ in a real vector space $V$ satisfies $\lambda X \in W$ for all $X \in W$ and $\lambda \ge 0$; it is a convex cone when also $X+Y \in W$ for all $X, Y \in W$, and a wedge when it is closed. Its edge $E(W) = W \cap (-W)$ is the largest subspace it contains, and $W$ is pointed when $E(W) = \{0\}$, the purely irreversible case, and generating (nonempty interior) when $W - W = V$. The dual wedge $W^* = \{f \in V^* \mid f(X) \ge 0 \text{ for all } X \in W\}$ is a cone of linear functionals on $V$ and determines the supporting hyperplanes of $W$. When $V$ is a generator algebra, an observable gives such a functional only after an input state $\rho$ has been fixed, through the expectation pairing $f_{\rho,A}(\Lcal)=\Tr[A\Lcal(\rho)]$. This pairing maps the $d^2$-dimensional real space of Hermitian observables into $V^*$ and need not span the dual; for example, $\dim\gLK=d^4-d^2$. Thus there is no canonical identification of $W^*$ with a cone of observables.

\paragraph*{Face Structure and Geometric Invariants}

\begin{definition}[Face]
A face $F$ of $W$ is a subwedge with $X+Y \in F$ and $X, Y \in W$ implying $X, Y \in F$; it is exposed when $F = \{X \in W \mid f(X)=0\}$ for some $f \in W^*$.
\end{definition}

The edge is the smallest face. A wedge with non-trivial edge has no extreme rays at all: pick $e \in E(W) \setminus \{0\}$ and $x \in W \setminus \{0\}$, and the decomposition $x = (x+e) + (-e)$ has both summands in $W$, so $\R_{\ge 0}x$ can be a face only if $-e \in \R_{\ge 0}x$, which forces $x \in E(W)$, and there $x = 2x + (-x)$ with $-x \in W \setminus \R_{\ge 0}x$ rules the ray out as well. The relevant indecomposable conic directions are therefore the extreme rays of the pointed quotient cone $W/E(W)$, and every statement about extreme rays below is to be read in that quotient.

\subsubsection{Lie Wedges: The Infinitesimal Objects of Irreversibility}

A Lie wedge is the tangent object of a Lie semigroup, exactly so locally, with the global question more delicate (Section~\ref{sec:dlwp_framework}); the notion comes from the Lie semigroup literature \cite{hilgert1989lie}.

\begin{definition}[Lie Wedge]
A closed convex cone $W$ in a Lie algebra $\g$ is a Lie wedge if it satisfies the invariance condition:
\begin{equation}
\label{eq:lie_wedge_condition}
e^{\ad(X)}(W) \subseteq W \quad \text{for all } X \in E(W).
\end{equation}
\end{definition}

The edge collects the reversible (group like) generators, and the invariance condition asks that their action leave the cone fixed. For quantum dynamics the edge is Hamiltonian, so the condition encodes a physical fact, that a unitary change of basis preserves complete positivity. Applying \eqref{eq:lie_wedge_condition} to $\pm X$ for $X \in E(W)$ gives $e^{\ad(X)}(W) = W$, so the inclusion form here matches the equality form of \cite{hilgert1989lie}.

\subsubsection{Infinitesimal Condition: Subtangency and the Tangent Cone}

The Lie wedge condition has an infinitesimal form built on the tangent cone, the Bouligand contingent cone at a point.

\begin{definition}[Tangent Cone]
\label{def:tangent_cone_wedge}
The tangent cone $T_Y(W)$ at $Y \in W$ collects the limits $V = \lim_{k\to\infty} \lambda_k (Y_k - Y)$ over sequences $\{Y_k\} \subset W$ converging to $Y$ and positive scalars $\{\lambda_k\}$.
\end{definition}

If $W$ is a convex cone, the definition simplifies significantly. The tangent cone at $Y \in W$ is the closure of the cone generated by the differences $X-Y$ for $X \in W$:
\begin{equation}
T_Y(W) = \overline{\cone(W - \{Y\})} = \overline{\{ \lambda(X-Y) \mid X \in W, \lambda \ge 0 \}}.
\end{equation}
Geometrically $T_Y(W)$ holds the directions pointing from $Y$ into $W$.

The infinitesimal Lie wedge condition is subtangency, which relates the bracket to the tangent cone.

\begin{proposition}[Subtangency Condition]
\label{prop:subtangency}
Let $W$ be a closed wedge in a finite-dimensional Lie algebra, with edge $E(W) = W \cap (-W)$. Then $W$ is a Lie wedge if and only if the subtangency condition holds:
\begin{equation}
\label{eq:subtangency}
[X, Y] \in T_Y(W) \quad \text{for all } X \in E(W), Y \in W.
\end{equation}
\end{proposition}

This is the subtangentiality characterization of Hilgert, Hofmann, and Lawson \cite[Prop.~1.3]{hofmann1991shortcourse}, \cite[Ch.~II, Ch.~IV]{hilgert1989lie}. The forward implication differentiates $e^{\ad(X)}W = W$ at the identity. The reverse is the substantive direction, where closedness does the work, since subtangency of the field $Y \mapsto [X, Y]$ forces the flow $e^{t\,\ad(X)}$ to leave the closed set invariant, a viability argument of Nagumo and Bony--Brezis type. The edge keeps pointing into the wedge, and this is the form we use when deriving the Fiber Compatibility Condition in semidirect products.

\subsection{The GKLS Wedge $\WGKLS$}
\label{sec:structure_gkls}

With the Lie wedge language in place, we can describe the set of physical generators precisely. This subsection records the GKLS theorem, explains why $\WGKLS$ fails to close under the Lie bracket, and establishes the structure it does have, that of a Lie wedge.

\subsubsection{Quantum Dynamical Semigroups and the GKLS Theorem}

Markovian evolution is described by a QDS, a one parameter semigroup of CPTP maps $\{\Ecal_t\}_{t \ge 0}$ generated by $\Lcal$, such that $\Ecal_t = e^{t\Lcal}$. The defining physical constraint is Complete Positivity (CP). A map $\Ecal$ is positive if it maps positive operators to positive operators. It is completely positive if the extended map $\Ecal \otimes \Id_k$ acting on $\Bcal(\Hcal \otimes \C^k)$ is positive for all $k \ge 1$; the stronger condition reflects the tensor product structure of quantum mechanics and is what keeps the dynamics physical in the presence of entanglement with an auxiliary system. Which channels arise from such semigroups, and how the channel set stratifies under composition, is a question of its own, developed in the divisibility theory of Wolf and Cirac \cite{wolf2008dividing}; this paper stays at the level of generators.

The structure of the generator $\Lcal$ of a QDS is given by the celebrated GKLS theorem \cite{gorini1976completely, lindblad1976generators}. We work with the Kossakowski matrix form, as it explicitly separates the coherent and incoherent components and parameterizes the positivity constraint. We denote the space of traceless Hermitian operators by $\Htr$.

\begin{theorem}[GKLS Theorem (Kossakowski Matrix Form)]
A linear map $\Lcal \in \g_{LK}$ is the generator of a QDS if and only if it can be written in the form using a Hilbert--Schmidt orthonormal basis $\{F_j\}_{j=1}^{d^2-1}$ for $\Htr$, with $\Tr(F_i^\dagger F_j)=\delta_{ij}$:
\begin{equation}
\label{eq:gkls_kossakowski}
\Lcal(\rho) = -i[H, \rho] + \sum_{j,k=1}^{d^2-1} K_{jk} \left( F_j \rho F_k^\dagger - \frac{1}{2} \{F_k^\dagger F_j, \rho\} \right),
\end{equation}
where $H$ is a Hermitian operator (Hamiltonian), and the Kossakowski matrix $K$ is positive semidefinite ($K \ge 0$).
\end{theorem}

The condition $K \ge 0$ is equivalent to the Conditional Complete Positivity (CCP) of the generator. The GKLS theorem thereby maps the geometry of the cone of positive semidefinite matrices onto the geometry of the space of physical dynamics, modulo the choice of basis and the Hamiltonian component.

The set $\WGKLS$ of GKLS generators forms a closed convex cone in $\g_{LK}$. Convexity is inherited directly from the positive semidefinite matrices: if $\Lcal_1, \Lcal_2$ have Kossakowski matrices $K_1, K_2 \ge 0$, the combination $\alpha \Lcal_1 + \beta \Lcal_2$ ($\alpha, \beta \ge 0$) has Kossakowski matrix $\alpha K_1 + \beta K_2 \ge 0$. The result is a cone of dimension $d^4-d^2$ whose boundary is stratified by the rank of $K$ and whose extreme rays modulo the edge, that is the extreme rays of the pointed quotient $\WGKLS/E(\WGKLS)$, correspond to rank-one generators. The cone itself carries no extreme ray, its edge being the $(d^2-1)$-dimensional space of Hamiltonian generators (Lemma~\ref{lem:gkls_edge}).

\paragraph*{Failure of Closure under the Lie Bracket}

The set $\WGKLS$ is not a subalgebra of $\g_{LK}$. Lie brackets belong to the formal algebraic closure, but isolating a commutator generally uses inverse evolutions and therefore need not be available for one-sided dissipative directions. Positive-time Trotter alternation supplies conic first-order combinations, while iterated brackets generate only the Lie algebra of the chosen directions. Neither operation guarantees access to every element of $\g_{LK}$, and bracket closure does not preserve the spectral constraint $K \ge 0$. Example~\ref{ex:qutrit_closure} in the Introduction exhibits the failure concretely for a qutrit, where the commutator of two valid rank-1 dissipative generators has an indefinite Kossakowski form, equivalently a negative eigenvalue in its conditional Choi compression; the supporting Gell-Mann computation appears in Appendix~\ref{app:qutrit_counterexample}. Dimension is not what drives the failure. Corollary~\ref{cor:hamiltonian_rigidity}(4) supplies the exact criterion, namely that a bracket of two generators lands in $\WGKLS$ only when it is purely Hamiltonian, and Remark~\ref{rem:rigidity_consequences} records the qubit witness that already violates it.

\paragraph*{Lie Semialgebras and Their Restricted Scope}

A Lie semialgebra is a wedge compatible with the local Campbell--Hausdorff multiplication: $(B \cap W) * (B \cap W) \subseteq W$ for every sufficiently small Campbell--Hausdorff neighborhood $B$. Its bracket form reads $[x, \mathcal{T}_x(W)] \subseteq \mathcal{T}_x(W)$ for every $x \in W$, where
\begin{equation}
\label{eq:tangent_space_at_x}
\mathcal{T}_x(W) \;:=\; E\bigl(T_x(W)\bigr)
\end{equation}
is the \emph{tangent space} at $x$, the largest linear subspace of the tangent cone $T_x(W)$ of Definition~\ref{def:tangent_cone_wedge}. The equivalence of the multiplicative and bracket forms is the semialgebra characterization of Hilgert, Hofmann, and Lawson \cite[Prop.~1.5]{hofmann1991shortcourse}. The bracket form uses this tangent space, rather than the full tangent cone. In $\sll(2,\R)$ with the standard basis $\{H,E,F\}$ the wedge $W = \{aH+bE+cF \mid c \ge 0\}$ is the half-space bounded by the Borel subalgebra $\operatorname{span}\{H,E\}$, and a half-space whose bounding hyperplane is a subalgebra is a Lie semialgebra. Yet $T_H(W) = W$ and $[H,F] = -2F \notin W$, so the criterion stated on the tangent cone would reject it. The semialgebra condition is strictly stronger than the Lie wedge condition, which constrains the geometry only along the edge. Within $\WGKLS$, pointwise bracket stability is therefore confined to the restricted situations in the theorem below.

\begin{theorem}[Lie semialgebras in $\WGKLS$ \cite{o2025lie}]
A Lie wedge $W \subseteq \WGKLS$ generated by a set of controls specialises to a Lie semialgebra only when the coherent controls have no effect on both the inherent drift Hamiltonian and the incoherent part of the dynamics.
\end{theorem}

Thus the Hamiltonian and dissipative parts decouple, as in pure dephasing in a fixed basis, where the relevant generators commute. General open systems instead require a structure that retains the bracket action of the reversible part together with one-sided dissipative directions.

\subsubsection{$\WGKLS$ is a Lie Wedge}

The correct structure, identified in geometric control theory \cite{dirr2009lie}, is that of a Lie wedge. We first record the edge of the GKLS cone, a characterization the rest of the paper uses repeatedly.

\begin{lemma}[Edge of the GKLS wedge (cf. \cite{dirr2009lie})]
\label{lem:gkls_edge}
The edge of $\WGKLS$ consists precisely of the Hamiltonian generators,
\begin{equation}
E(\WGKLS) = \{\Lcal_H \mid \Lcal_H(\rho) = -i[H, \rho],\ H=H^\dagger,\ H \in \Htr\} \cong \su(d).
\end{equation}
In particular, reversible dissipation is impossible: a generator with $\pm\Lcal \in \WGKLS$ has vanishing Kossakowski matrix.
\end{lemma}

\begin{proof}
The GKLS representation \eqref{eq:gkls_kossakowski} is linear in the pair $(H, K)$, so $-\Lcal$ corresponds to the data $(-H, -K)$. Membership of both $\pm\Lcal$ in $\WGKLS$ requires $K \ge 0$ and $-K \ge 0$, hence $K = 0$, leaving the purely Hamiltonian generators. We adopt the convention that $H$ is traceless; under this convention the edge is isomorphic to $\su(d)$.
\end{proof}

\begin{theorem}[$\WGKLS$ is a Lie wedge \cite{dirr2009lie}]
\label{thm:gkls_lie_wedge}
The set $\WGKLS$ is a Lie wedge in the Lindblad Kossakowski Lie algebra $\g_{LK}$.
\end{theorem}

\begin{proof}[Proof sketch]
The set $\WGKLS$ is a closed convex cone, and by Lemma~\ref{lem:gkls_edge} its edge acts by unitary conjugation, which transforms the Kossakowski matrix by congruence and therefore preserves its positivity. The detailed verification is given in Appendix~\ref{app:proof_gkls_lie_wedge}.
\end{proof}

Physically, the invariance condition says that a unitary change of frame cannot turn a physical generator into an unphysical one. The invariance does not extend beyond the edge.

\begin{proposition}[$\WGKLS$ is not an invariant cone]
For $d \ge 2$, $\WGKLS$ is not an invariant cone in $\g_{LK}$.
\end{proposition}

\begin{proof}
Invariance would require $e^{\ad(X)}(\WGKLS) \subseteq \WGKLS$ for every $X \in \g_{LK}$, that is, invariance of complete positivity under conjugation by every invertible Hermiticity- and trace-preserving map $e^X$. Conjugation preserves complete positivity only for special maps (completely positive or completely copositive ones). Concretely, a qubit already refutes invariance. Take the isotropic dilation $X = (\log 2)\,\Id_{\Htr}$, extended to $\Bcal(\Hcal)$ by $X(\mathbb{I}) = 0$. It preserves Hermiticity and annihilates traces, so $X \in \g_{LK}$, while its Kossakowski matrix is $-\tfrac{1}{2}(\log 2)$ times the identity and is therefore negative definite, placing $X$ outside $\WGKLS$. In the affine model of Theorem~\ref{thm:structure_glk} the flow $e^{\ad X}$ leaves the linear part of a generator untouched and doubles the inhomogeneous part. Applied to the amplitude damping generator $\Dcal_{\sigma_-}$, whose Kossakowski spectrum is $\{0,0,1\}$, it returns a generator of spectrum $\{-\tfrac{1}{2},\,0,\,\tfrac{3}{2}\}$, so conditional complete positivity is destroyed while Hermiticity and trace preservation survive. Invariance holds under the unitary subgroup generated by the edge and fails beyond it, which is the Lie wedge property and no more.
\end{proof}

\subsubsection{The Dissipation Functional}
\label{sec:dissipation_functional}

A single linear functional on $\g_{LK}$ separates the coherent directions from the
dissipative ones, and it does so in a way that interacts with the Levi decomposition.
Everything in this subsection is elementary; what it buys is recorded in
Corollary~\ref{cor:hamiltonian_rigidity}, which is what singles out the Levi
splitting, among the decompositions of $\g$ into an ideal and a complementary
subalgebra, as the one along which Section~\ref{sec:levi_type_theorem} decomposes the wedge.

\begin{lemma}[Dissipation functional]
\label{lem:dissipation_functional}
Let $\{F_j\}_{j=1}^{d^2-1}$ be a basis of traceless Hermitian operators on $\Hcal$ with Gram matrix
$G_{jk} = \Tr(F_kF_j)$, and let $\Lcal \in \g_{LK}$ have GKLS data $(H, K)$ in that
basis. Then, as a superoperator on $\Bcal(\Hcal)$,
\begin{equation}
\label{eq:dissipation_functional}
\Tr \Lcal \;=\; -\,d\,\Tr(K G),
\end{equation}
independently of $H$. In the orthonormal normalization $G = \mathbb{I}$ this reads
$\Tr\Lcal = -d\,\Tr K$; in the normalization $\Tr(F_iF_j) = 2\delta_{ij}$ used in
Example~\ref{ex:qutrit_closure} and Appendix~\ref{app:qutrit_counterexample} one has
$G = 2\cdot\mathbb{I}$ and $\Tr\Lcal = -2d\,\Tr K$.
\end{lemma}

\begin{proof}
For fixed $A, B \in \Bcal(\Hcal)$ the superoperator $\rho \mapsto A\rho B$ has trace
$\Tr(A)\Tr(B)$. The Hamiltonian part contributes
$\Tr(-i[H,\cdot]) = -i(d\Tr H - d\Tr H) = 0$. In the dissipative part the sandwich term
contributes $\Tr(F_j)\Tr(F_k^{\dagger}) = 0$ because the basis is traceless, while each of
the two anticommutator terms contributes $\tfrac12 d\,\Tr(F_k^{\dagger}F_j)$. Summing,
$\Tr\Lcal = -d\sum_{jk} K_{jk}G_{jk} = -d\Tr(KG)$, where the last equality uses the
real symmetry of the Gram matrix of a Hermitian basis. For a general complex operator basis
the coordinate contraction is instead $-d\Tr(KG^{\mathsf T})$ with this Gram convention; the
Hermitian hypothesis is therefore material to the displayed formula.
\end{proof}

Write $\varphi(\Lcal) := \Tr K$ for the orthonormal normalization, so that
$\varphi = -\tfrac{1}{d}\Tr(\cdot)$ is linear on $\g_{LK}$. On the GKLS wedge $K \succeq 0$,
so $\varphi \ge 0$ on $\WGKLS$, and $\varphi(\Lcal) = 0$ forces $K = 0$: the functional
vanishes on $\WGKLS$ exactly at the Hamiltonian generators, which by
Lemma~\ref{lem:gkls_edge} are the edge $E(\WGKLS) \cong \su(d)$. It measures the total
dissipation strength.

\begin{corollary}[Hamiltonian rigidity of the commutator ideal]
\label{cor:hamiltonian_rigidity}
Let $\g \subseteq \g_{LK}$ be a subalgebra and $W \subseteq \WGKLS \cap \g$ a wedge. Then
\begin{equation}
\label{eq:hamiltonian_rigidity}
W \cap [\g,\g] \;\subseteq\; E(\WGKLS) \cap \g ,
\end{equation}
that is, every element of $W$ lying in the commutator ideal is purely Hamiltonian.
Consequently, with $\g = \s \ltimes \rr$ the Levi decomposition:
\begin{enumerate}
  \item $W \cap \s \subseteq E(\WGKLS)$: the wedge meets the Levi factor only in
        Hamiltonian directions;
  \item if $\g$ is semisimple then $W$ consists entirely of Hamiltonian generators, so a
        semisimple system algebra synthesizes no dissipation at all;
  \item $\varphi(X) = \varphi(\pi_{\rr}(X))$ for every $X \in \g$, where $\pi_{\rr}$ is the
        projection onto $\rr$ along the Levi section: all dissipation strength is carried by
        the radical component;
  \item for any $\Lcal_1, \Lcal_2 \in \g_{LK}$, the bracket $[\Lcal_1,\Lcal_2]$ lies in
        $\WGKLS$ if and only if it is purely Hamiltonian.
\end{enumerate}
\end{corollary}

\begin{proof}
Traces of commutators vanish, so $\varphi$ vanishes on $[\g,\g]$ by
Lemma~\ref{lem:dissipation_functional}. If $X \in W \cap [\g,\g]$ then $K(X) \succeq 0$ and
$\Tr K(X) = 0$, whence $K(X) = 0$ and $X \in E(\WGKLS)$. For (1), $\s = [\s,\s] \subseteq
[\g,\g]$ by semisimplicity. Item (2) is (1) with $\g = \s$, using $W \subseteq \g$. Item (3)
follows since $\varphi$ vanishes on $\s \subseteq [\g,\g]$. For (4), $\Tr[\Lcal_1,\Lcal_2] = 0$
gives $\Tr K = 0$; if the bracket lies in $\WGKLS$ then $K \succeq 0$ forces $K = 0$, and the
converse is immediate.
\end{proof}

\begin{remark}[Bracket-closure consequences]
\label{rem:rigidity_consequences}
Item (4) gives the exact bracket criterion. The qutrit case $d=3$ appears in
Example~\ref{ex:qutrit_closure}. The failure already occurs for a qubit, since
$[\Dcal_{\sigma_-},\Dcal_{\sigma_x}]$ is a nonzero pure Bloch translation with Kossakowski
spectrum $\{-1,0,1\}$. In dimension $d = 2$, a pair of \emph{unital}
dissipators cannot witness it: unitality makes the Bloch parts symmetric, and a commutator of
symmetric matrices is antisymmetric, hence a rotation. From $d = 3$ on, unital pairs fail too.
Item (2) identifies the semisimple-type class with the closed-system, no-synthesizable-dissipation class.
\end{remark}

\subsection{Invariant Cones, Spectral Terminology, and the Convex Type Condition}
\label{sec:invariant_cones_prelim}

A condition stronger than the Lie wedge property is invariance under the adjoint action of the full algebra. We record the definition, fix the spectral terminology used below, and then state the criterion from Neeb's theory; Appendix~\ref{app:neeb_theory_details} gives more detail.

\begin{definition}[Invariant Cone]
A Lie wedge $W$ is an invariant cone if $e^{\ad(X)}(W) \subseteq W$ for all $X \in \g$.
\end{definition}

If $W$ is an invariant cone, the edge $E(W)$ must be an ideal in $\g$. In reductive algebras the classification of invariant cones is well understood and related to the geometry of the associated symmetric spaces. In non-reductive algebras existence is a delicate question, because the coupling $\rho$ between Levi factor and radical can generate flows that no pointed cone survives.

\paragraph*{Stability of the Adjoint Action}

The following terms describe the spectrum of a specified operator $\ad(X)$. They do not, by themselves, decide whether the algebra or a module is of convex type.

\begin{itemize}
    \item \textbf{Elliptic:} A semisimple operator with purely imaginary spectrum generates a bounded one-parameter group, represented geometrically by rotations or oscillations.
    \item \textbf{Hyperbolic:} An operator with an eigenvalue of nonzero real part has an exponentially expanding or contracting mode.
    \item \textbf{Parabolic:} A nonzero nilpotent Jordan part with zero real spectral part produces polynomial shear. Mixed cases can contain more than one of these components.
\end{itemize}

Whether a pointed generating invariant cone exists is a separate algebraic question. In particular, failure of convex type need not produce an expanding mode.

\paragraph*{Neeb's Theory and the Convex Type Condition}

Neeb's theory of invariant cones \cite{neeb1994classification, Neeb2000}, building on work of Vinberg \cite{vinberg1963theory} and Olshanskii \cite{olshanskii1981invariant}, supplies the relevant module condition; ordered structures of this kind also occur in algebraic quantum field theory and the analysis of causality \cite{Brunetti2002}. In Neeb's Section~II the acting algebra $\g$ is real and reductive. For a symplectic $\g$-module $(V,\omega)$ with action $A_X$, define $\phi_X(v)=\tfrac12\omega(A_Xv,v)$, let $\Phi$ be the corresponding moment map, and set $C_V=\overline{\cone(\Phi(V))}$. The module is of convex type when $C_V$ is pointed, equivalently when its dual $W_V=C_V^*$ is generating, and when the center $\mf{z}(\g)$ acts semisimply with purely imaginary eigenvalues \cite[Definition~II.1]{neeb1994classification}. If, in addition, $\g$ has the compactly embedded Cartan used in Proposition~II.23, $V_{\mathrm{fix}}=\{v\mid\g.v=0\}$ is zero, and the center condition is already assumed, that proposition makes convex type equivalent to the existence of an element in $\mf{z}(\mf{k})$, where $\mf{k}$ is the maximal compactly embedded subalgebra, for which $\phi_X$ is positive definite. Thus a positive-definite quadratic Hamiltonian is a witness under those hypotheses; it is not the unqualified definition. No converse spectral claim follows from its absence.

The Dynamical Lie Wedge $W(G)$ is generally invariant only under its edge, so Neeb's invariant-cone conclusions do not transfer directly. Section~\ref{sec:neebs_theory} records a structured correspondence between the FCC and convex type when the nilradical carries a bracket-induced symplectic module; for $\gLK$ itself that module is absent and the convex type condition is vacuous.

\section{The Dynamical Lie Wedge Pair (DLWP) Framework}
\label{sec:dlwp_framework}

Carrying the algebraic classification to open systems calls for a different approach. We have to keep the algebraic structure generated by commutation apart from the geometric constraints that complete positivity imposes, and then study how the two interact. The Dynamical Lie Wedge Pair (DLWP) is the object we introduce for this purpose. It carries the algebraic classification program over to irreversible evolution.

\paragraph*{The Conceptual Shift: Algebra vs. Geometry}

In closed systems the physical generators, the skew-Hermitian operators, form a real Lie algebra. Here the algebraic structure and the geometric space of admissible dynamics are one and the same, and the DLA paradigm rests on that coincidence. Classifying the dynamics then amounts to classifying the Lie algebra.

In open systems this coincidence breaks down. The Lie bracket still defines the algebraic structure, but it generates a space $\g_{LK}$ strictly larger than the cone of physical generators $\WGKLS$. We are therefore forced to study a relationship rather than a single object, namely how the algebra generated by the bracket sits against the cone that complete positivity carves out.

\subsection{Formal Definition and Construction of the DLWP}
\label{sec:dlwp_definition}

We now define the central object of the classification, attached to a quantum system through its set of available generators $G$. The construction follows geometric control theory and the theory of Lie semigroups \cite{dirr2009lie, OMeara2011IllustratingTG}.

\begin{definition}[Dynamical System Specification]
A generator-labelled Markovian control specification is a finite family of independently actuable positive rays $\R_{\geq0}\Lcal_i$ with $\Lcal_i\in\WGKLS$. We write $G=\{\Lcal_1,\ldots,\Lcal_m\}$ for chosen representatives of those rays; replacing a representative by a positive multiple does not change the specification. Bidirectional coherent control is stated by including both $\Lcal_H$ and $-\Lcal_H$ in $G$. A decomposition of a fixed autonomous Lindbladian supplies such a family only when its components can be actuated separately.
\end{definition}

\begin{definition}[Dynamical Lie Wedge Pair (DLWP)]
\label{def:DLWP}
Given a specification $G$, the Dynamical Lie Wedge Pair (DLWP) is the pair $S(G) = (\g(G), W(G))$, where:

\begin{enumerate}
    \item \textbf{System Lie Algebra} $\g(G)$. This is the algebra generated by $G$ under the Lie bracket; equivalently, the smallest real Lie subalgebra of $\g_{LK}$ that contains $G$. We write
    \begin{equation}
    \g(G) = \Lie(G).
    \end{equation}
    It is the open-system counterpart of the \emph{dynamical Lie algebra} of a control family, a construction studied at length in geometric control theory and quantum control \cite{jurdjevic1997geometric, d2021introduction, Schirmer2000Complete}.
    \item \textbf{Dynamical Lie Wedge} $W(G)$. This is the smallest closed Lie wedge in $\g(G)$ that contains $G$. Proposition~\ref{prop:construction_W(G)_saturation} proves its existence and gives an iterative construction. It is an inner approximation of the tangent wedge of the reachability semigroup in Definition~\ref{def:W(G)_semigroup}; the two coincide exactly when $W(G)$ is global.
\end{enumerate}
\end{definition}

\begin{remark}[What is being classified]
\label{rem:scope_control_systems}
The classification input is a generator-labelled control specification, not a decomposition of a bare Liouvillian. If the only available dynamics is a fixed autonomous generator $\Lcal$, then $G=\{\Lcal\}$, the algebra $\g(G)=\R\Lcal$ is one-dimensional and abelian, the minimal local wedge is $\R_{\geq0}\Lcal$, and the global object is its semigroup. A Hamiltonian--jump decomposition may be used as $G$ only when the resulting rays are independently actuable; physical distinguishability alone is insufficient. This restriction also matters because the GKLS representation admits gauge transformations, including inhomogeneous shifts $L_k\mapsto L_k+c_k\mathbb I$ with a compensating change of $H$ and unitary recombinations of jump operators. Such changes do not create additional control rays. Throughout the paper, the declared actuation data are part of the specification, and all invariants belong to the resulting pair $(\g(G),W(G))$. The multi-generator case describes piecewise-controlled Markovian dynamics, in the same sense in which a dynamical Lie algebra describes a coherent control family.
\end{remark}

\subsubsection{Interpretation of the Components}

\paragraph{The System Lie Algebra $\g(G)$.}
The algebra $\g(G)$ is the formal Lie-algebraic envelope of $G$. It collects the real linear combinations and iterated brackets generated algebraically from $G$, but these directions need not be reachable through positive-time sequences because isolating a commutator can require inverse evolutions. Operational reachability belongs to $\Sigma(G)$ and $L(\Sigma(G))$; reversible edge directions supply the inverses needed for their adjoint action. As Section~\ref{sec:structure_gkls} shows, $\g(G)$ generally contains elements that are not valid GKLS generators. It records the algebraic complexity and symmetry of the interactions with positivity set aside, and its Levi decomposition supplies the algebraic backbone of the classification. The dimension and structure of $\g(G)$ measure the range of algebraic relations generated by the available interactions.

\paragraph{The Dynamical Lie Wedge $W(G)$.}
The wedge $W(G)$ is the minimal closed Lie-wedge local model determined by $G$. It carries the constraints of irreversibility and complete positivity, and by construction $W(G) \subseteq \g(G)$. The intersection $\WGKLS \cap \g(G)$ is a closed convex cone containing $G$, and because $\g(G)$ is a linear subspace its edge is $E(\WGKLS) \cap \g(G)$. For $X$ in that edge the map $e^{\ad X}$ preserves $\WGKLS$, by the Lie wedge property of Theorem~\ref{thm:gkls_lie_wedge}, and preserves $\g(G)$, because $\g(G)$ is a subalgebra. Thus $\WGKLS \cap \g(G)$ is a closed Lie wedge in $\g(G)$ containing $G$, and minimality gives
\begin{equation}
W(G) \subseteq \g(G) \cap \WGKLS.
\end{equation}
The containing algebra $\g(G)$ is itself a closed Lie wedge, but it belongs to $\WGKLS$ only when all of its elements are Hamiltonian. Taking $W(G)$ to be the smallest Lie wedge containing $G$ therefore gives the minimal closed-wedge approximation determined by the specification and records how positivity restricts the algebraic possibilities left open by $\g(G)$.

\subsubsection{Construction of the Dynamical Lie Wedge $W(G)$}

The construction of $W(G)$ is non-trivial and requires clarification. The wedge $W(G)$ must be distinguished from the simple convex cone generated by $G$, $\cone(G) = \{\sum \alpha_i \Lcal_i \mid \alpha_i \ge 0\}$. In general, $\cone(G)$ is not necessarily a Lie wedge, as it may not satisfy the invariance condition (Equation \eqref{eq:lie_wedge_condition}). The Lie wedge condition imposes constraints related to the action of the reversible component of the dynamics.

Two distinct objects are in play here, and we keep them apart. The \emph{local} object is the smallest closed Lie wedge in $\g(G)$ containing $G$, which we denote by $W(G)$. The \emph{global} object is the tangent wedge of the closed reachability semigroup. In control-theoretic terminology, this tangent wedge is the smallest global Lie wedge containing the control directions and is called their global Lie saturate \cite[Thm.~II.5(b)]{dirr2009lie}. The two objects need not coincide, and their gap is the local--global problem of Lie semigroup theory \cite{hilgert1989lie}.

\begin{definition}[Lie Semigroup and Tangent Wedge]
\label{def:W(G)_semigroup}
The Lie semigroup $\Sigma(G)$ generated by $G$ is the closure of the set of all finite products of the form $e^{t_1 \Lcal_1} \cdots e^{t_k \Lcal_k}$ with $\Lcal_i \in G, t_i \ge 0$. Each factor evolves under a chosen representative of one declared ray for a nonnegative time, so $\Sigma(G)$ is the closed reachable set for piecewise-constant switching among those rays; a bidirectional coherent direction is available only when both signs occur in $G$. Its tangent wedge at the identity, written $L(\Sigma(G))$, consists of the generators subtangent to $\Sigma(G)$ and is a closed Lie wedge \cite{hilgert1989lie}.
\end{definition}

The tangent wedge contains $G$ but need not lie inside $\g(G)$, so the relevant comparison inside $\g(G)$ is $L(\Sigma(G)) \cap \g(G)$. It is a closed Lie wedge in $\g(G)$ containing $G$, by the argument already used for $\WGKLS \cap \g(G)$ in Section~\ref{sec:dlwp_definition}, with $L(\Sigma(G))$ in place of $\WGKLS$. Minimality of $W(G)$ therefore gives the inclusion
\begin{equation}\label{eq:saturate_vs_tangent}
W(G) \;\subseteq\; L(\Sigma(G)).
\end{equation}
The inclusion can be strict. If $G = \{\Lcal_H\}$ consists of a single Hamiltonian generator $\Lcal_H = -i[H,\cdot\,]$ with quasi-periodic flow, then $W(G) = \R_{\geq 0}\,\Lcal_H$ is already a closed Lie wedge with trivial edge, while the closure of $\{e^{t\Lcal_H} : t \geq 0\}$ is a compact group, so $L(\Sigma(G))$ contains $-\Lcal_H$ as well. Equality in \eqref{eq:saturate_vs_tangent} is the statement that the wedge $W(G)$ is \emph{global} in $\g(G)$, a separate and generally difficult question in Lie semigroup theory \cite{hilgert1989lie}. The structural results use the inclusion, and the decomposition theorem of Section~\ref{sec:levi_type_theorem} takes an arbitrary Lie wedge as input. The classification framework works with the minimal closed Lie wedge, an inner approximation of reachability of the kind used in quantum control \cite[Sec.~II-D]{OMeara2011IllustratingTG}.

The practical construction of $W(G)$ from $G$ enforces the Lie wedge condition iteratively. In certain situations the construction simplifies. If the closed conical hull $\overline{\cone}(G)$ has trivial edge, the invariance condition is vacuous, so $\overline{\cone}(G)$ is itself a closed Lie wedge, and minimality gives $W(G) = \overline{\cone}(G)$.

However, if the edge $E_0 = E(\cone(G))$ is non-trivial, we must account for the action of the group generated by $E_0$ on the generators $G$. The adjoint-orbit closure can generate new reversible elements, so the edge of the resulting Lie wedge $E(W(G))$ may be strictly larger than $E_0$.

\paragraph{The Mechanism of Edge Growth.}
Edge growth occurs when the adjoint orbit of the Hamiltonian subset creates new lineality in the closed conical hull. Specifically, let $G_{E_0}$ be the connected Lie group generated by $E_0$ and set $W_{\mathrm{orb}}=\overline{\cone(\{\Ad(g)(X)\mid X\in G,\ g\in G_{E_0}\})}$. If a Hamiltonian direction $H\notin E_0$ satisfies $H,-H\in W_{\mathrm{orb}}$, then $H$ joins the enlarged edge. A compact edge orbit can supply both signs of a one-sided Hamiltonian direction. Genuinely dissipative orbit directions cannot cancel into the edge: the dissipation functional of Lemma~\ref{lem:dissipation_functional} is nonnegative on $\WGKLS$ and vanishes exactly on Hamiltonian generators. For a concrete illustration of maximal edge growth, consider the analysis of the ``magic cone'' in $\mf{sl}(2, \R)$, where adjoint-orbit closure can generate the entire algebra (see, e.g., Chapter V in \cite{hilgert1989lie}).

The construction of $W(G)$ must account for this Hamiltonian lineality growth. Consequently, the Levi decomposition must use the edge of the completed wedge.

\begin{proposition}[The iterative closed-wedge construction yields $W(G)$]
\label{prop:construction_W(G)_saturation}
The smallest closed Lie wedge in $\g(G)$ containing $G$ exists: the family of closed Lie wedges containing $G$ is nonempty ($\g(G)$ itself belongs to it) and is stable under arbitrary intersections. Initialized at $\cone(G)$, the iteration below stabilizes after finitely many steps, and its output is this smallest closed Lie wedge, namely $W(G)$. It adapts the five-step inner-approximation template of \cite[Sec.~II-D]{OMeara2011IllustratingTG}; the exactness and finite termination asserted here follow from the proof below. In particular, the fixed point lies inside the global tangent wedge $L(\Sigma(G))$ by \eqref{eq:saturate_vs_tangent}.
\end{proposition}

Following the orbit-and-edge steps of \cite[Sec.~II-D]{OMeara2011IllustratingTG}, start from $W_0 = \cone(G)$. Each abstract stage reads off the current edge $E_k = W_k \cap (-W_k)$, replaces $W_k$ by the closed conical hull of its adjoint orbit under the subgroup $G_{E_k} = \langle \exp(E_k) \rangle$, and compares the new edge with $E_k$. This description presupposes that the closed conical hull can be formed and its lineality determined. The sequence reaches an index $k$ with $E_{k+1}=E_k$, because strict inclusions among the finite-dimensional subspaces $E_k$ cannot continue indefinitely. The cone may still enlarge on the step that first yields this equality. At that index $W_{k+1}$ is a Lie wedge, and one further application gives $W_{k+2}=W_{k+1}=W(G)$.

\begin{proof}[Proof of Proposition~\ref{prop:construction_W(G)_saturation}]
Stability under intersections: an intersection of closed convex cones is a closed convex cone, and an element of the edge of the intersection lies in the edge of every member, so every member, and hence the intersection, is invariant under the corresponding $e^{\ad X}$. For the iteration, let $V$ be any closed Lie wedge containing $G$; we claim $W_k \subseteq V$ for every $k$. Indeed $W_0 = \cone(G) \subseteq V$, and if $W_k \subseteq V$ then $E_k = W_k \cap (-W_k) \subseteq E(V)$, so for every $g \in G_{E_k}$, written as a finite product of exponentials $e^{X_1} \cdots e^{X_m}$ with $X_i \in E_k \subseteq E(V)$, the wedge condition for $V$ absorbs each factor and gives $\Ad_g(W_k) \subseteq V$; and $V$, being a closed convex cone, absorbs the closed conical hull, giving $W_{k+1} \subseteq V$. The subspaces $E_k$ are nondecreasing, so finite dimensionality gives an index $k$ with $E_{k+1}=E_k$. The cone may still enlarge on that step. By construction, however, $W_{k+1}$ is invariant under $\Ad(G_{E_k})$: each element of this group permutes the adjoint-orbit union and hence preserves its closed conical hull. Since $E(W_{k+1})=E_{k+1}=E_k$, this is precisely invariance under the group generated by its edge, and $W_{k+1}$ is a closed Lie wedge containing $G$. The preceding claim puts it inside every such wedge, proving $W_{k+1}=W(G)$. Applying the construction once more uses the same edge and the already established invariance, so $W_{k+2}=W_{k+1}$.
\end{proof}

\begin{remark}[Finite abstract edge-stage construction]
Finite dimensionality bounds the number of strict enlargements of the subspaces $E_k$, so an index with $E_{k+1}=E_k$ occurs after finitely many stages. The cone can still enlarge on the step that first produces this equality; the proof shows that $W_{k+1}$ is then a Lie wedge and that the next stage is fixed. This abstract argument does not by itself give an effective algorithm for a nonpolyhedral stage, because one must still represent the closed conical hull and determine its lineality.
\end{remark}

\begin{remark}[Determining the edge $E_k$]
When $W_k$ is supplied by a finite polyhedral representation, its lineality space $E_k=W_k\cap(-W_k)$ is determined by finite-dimensional linear programming. For a nonpolyhedral $W_k$, the same formula defines the edge abstractly, but effective computation requires both a representation of $W_k$ and an edge oracle. Finite dimensionality alone does not supply either one.
\end{remark}

\begin{remark}[Local wedge and global tangent wedge]
Definition~\ref{def:DLWP} and the fixed point of Proposition~\ref{prop:construction_W(G)_saturation} give the same local object, the minimal closed Lie wedge $W(G)$. The tangent wedge $L(\Sigma(G))$ of Definition~\ref{def:W(G)_semigroup} is the global Lie saturate \cite[Thm.~II.5(b)]{dirr2009lie}; it matches $W(G)$ when the latter is global and otherwise is an upper bound by \eqref{eq:saturate_vs_tangent}. All subsequent arguments use this inclusion.
\end{remark}

In general, building $W(G)$ calls for the iterative closed-wedge construction and edge computation above. The gap between $\cone(G)$ and the global tangent wedge is central to control theory, where it governs reachability \cite{dirr2009lie}; a fast-unitary-control reachability wedge for a Lindbladian is computed in \cite[Prop.~2.6]{malvetti2024reachability}. For classifying the intrinsic dynamics fixed by $G$, we work with the minimal closed Lie wedge throughout.

\subsection{The Structure of the System Lie Algebra $\g(G)$}

The first step in the classification is the analysis of the algebraic structure of $\g(G)$. We apply the Levi decomposition:
\begin{equation}
\g(G) = \mf{s}(G) \ltimes \mf{r}(G).
\end{equation}

The structure of $\mf{s}(G)$ (the semisimple Levi factor) and $\mf{r}(G)$ (the solvable radical), along with the action $\rho$ of $\mf{s}(G)$ on $\mf{r}(G)$, determines the algebraic complexity of the system. This decomposition provides the basis for the classification of complexity:

\begin{itemize}
    \item \textbf{Solvable type:} If $\mf{s}(G)=\{0\}$, then $\g(G)$ is solvable. In every finite-dimensional complex representation, or after complexifying a real representation, Lie's theorem gives simultaneous upper triangularization and a complete invariant flag. Abelian, nilpotent, and general solvable radicals remain distinct algebraic cases. Pure dephasing and cascaded decay illustrate two particular realizations.
    \item \textbf{Semisimple type (Hamiltonian-only):} If $\mf{r}(G)=\{0\}$, then $\g(G)$ is semisimple. Corollary~\ref{cor:hamiltonian_rigidity}(2) makes $W(G)$ Hamiltonian-only, so this is the closed-system, no-synthesizable-dissipation class. Its algebraic complexity is recorded by the number and type of its simple factors.
    \item \textbf{Mixed reductive type:} If both $\mf{s}(G)$ and $\mf{r}(G)$ are non-trivial and $\rad(\g(G))=Z(\g(G))$, then $\g(G)$ is reductive. Equivalently, its adjoint representation is completely reducible \cite{knapp1996lie}, and the Levi product is the direct sum $\g(G)=\mf{s}(G)\oplus\mf{r}(G)$.
    \item \textbf{Mixed non-reductive type:} If both factors are non-trivial and $\rad(\g(G))\neq Z(\g(G))$, then $\g(G)$ is non-reductive. A non-trivial Levi action, $\rho\neq 0$, specifies one branch of this class. The corresponding DLWP is action-coupled exactly when both $W_{\rr}$ and $C_{\text{quot}}$ are non-trivial. The class also includes $\rho=0$ with a non-abelian radical. The algebraic data alone imply no conclusion about mixing or relaxation.
\end{itemize}

The analysis of the System Lie Algebra provides the algebraic invariants that form the basis of the classification. The subsequent analysis focuses on how these algebraic structures are constrained by the geometry of the Dynamical Lie Wedge.

\subsection{The Geometry of the Embedding $W(G) \hookrightarrow \g(G)$}

The structure of the DLWP is encoded in the geometry of how the Dynamical Lie Wedge $W(G)$ is embedded within the System Lie Algebra $\g(G)$. This embedding captures the tension between the algebraic structure generated by the Lie bracket and the physical constraints imposed by complete positivity.

\paragraph*{Key Geometric Features}

The geometric features of the wedge furnish the key local invariants of the specified system. They record its reversible edge, lineality, and convex structure.

\begin{itemize}
    \item \textbf{The Edge $E(W(G))$:} The subalgebra $E(W(G)) = W(G) \cap (-W(G))$ consists of the reversible (Hamiltonian) generators contained in the physical local wedge. Its structure determines the available reversible symmetries; a larger edge means more reversible directions.
    \item \textbf{The Interior $\text{int}(W(G))$:} If the interior of the wedge is non-empty relative to $\g(G)$, then the wedge is full-dimensional in its system algebra. Interior membership is a convex-geometric property; the relaxation or mixing produced by an interior generator remains a separate dynamical question.
    \item \textbf{The Generating Property:} We say $W(G)$ is generating in $\g(G)$ if $\Span(W(G)) = \g(G)$ (equivalently, $\g(G) = W(G) - W(G)$). This is a linear-span property of the local wedge. It does not by itself establish controllability or equality with the global tangent wedge.
    \item \textbf{Pointedness:} If $E(W(G)) = \{0\}$, the wedge is pointed. It then contains no nonzero reversible direction. This is a statement about the local wedge, not a global reachability criterion.
\end{itemize}

\paragraph*{The Classification Goal}

The classification problem is reformulated as the problem of classifying the isomorphism classes of DLWPs. Two DLWPs $(\g_1, W_1)$ and $(\g_2, W_2)$ are isomorphic if there exists a Lie algebra isomorphism $\phi: \g_1 \to \g_2$ such that $\phi(W_1) = W_2$. This definition captures both the algebraic structure and the geometric embedding.

The central challenge addressed in this paper is to determine how the algebraic classification of $\g(G)$ (the Levi decomposition) manifests in the geometry of the wedge $W(G)$. This requires understanding how the geometric constraints of the wedge interact with the algebraic structure of the semidirect product, particularly in the non-reductive case. The classification of DLWPs provides a rigorous mathematical framework for classifying the complexity of open quantum dynamics.

\subsection{Framework for DLWP Classification}
\label{sec:algorithm_dlwp}

To put the framework to work, we follow a systematic procedure for classifying a given open quantum system specified by its generators $G$. The steps are not uniform in cost. Standard constructive algorithms are available for the algebraic tasks in Steps~1 and~2, whereas the wedge computation in Step~3, the FCC test in Step~4, and the spectral test in Step~5 are analytic and can be expensive.

\begin{enumerate}
    \item \textbf{Algebraic Closure:} Compute the System Lie Algebra $\g(G) = \Lie(G)$ by iteratively computing commutators until the dimension stabilizes.
    \item \textbf{Levi Decomposition:} Split $\g(G) = \mf{s} \ltimes \mf{r}$ into a semisimple factor $\mf{s}$ and the solvable radical $\mf{r}$. Over characteristic zero, constructive algorithms for this decomposition are standard \cite{Graaf2000LieAT}.
        \begin{itemize}
            \item If $\mf{r}=\{0\}$, the system has semisimple type (Hamiltonian-only).
            \item If $\mf{s}=\{0\}$, the system has solvable type.
            \item If both factors are non-zero and $\rad(\g)=Z(\g)$, the system has mixed reductive type.
            \item If both factors are non-zero and $\rad(\g)\neq Z(\g)$, the system has mixed non-reductive type. This includes $\rho=0$ with a non-abelian radical.
        \end{itemize}
        After assigning one of these four structural types, proceed to Step~3.
    \item \textbf{Edge Identification and Closed-Wedge Construction:}
    Apply the abstract closed-wedge iteration to obtain $W(G)=W_{k+1}$ and its stabilized edge $E(W)$. With a finite cone representation and an edge oracle, this step becomes an effective computation. The orbit-and-edge template is recorded in \cite[Sec.~II-D]{OMeara2011IllustratingTG}; exactness here is Proposition~\ref{prop:construction_W(G)_saturation}.
    \begin{remark}[Constructive Algebraic Stages]
   The abstract orbit-and-edge construction has no general complexity bound without a representation of each cone and an edge oracle. Standard constructive algorithms remain available for determining $\g(G)$ and its Levi decomposition $\mf{s} \ltimes \mf{r}$, even when the wedge geometry requires separate analytic work.
    \end{remark}
    \item \textbf{FCC Test (See Prop. \ref{prop:infinitesimal_fcc}):} Verify the Fiber Compatibility Condition either globally, by edge invariance, or equivalently infinitesimally, by subtangency of the commutator against the local tangent cone:
        \begin{equation}
        \text{Test: } [X_0, Y] \in T_Y(W) \quad \forall X_0 \in E(W), Y \in W.
        \end{equation}
        Operationally, writing $X_0=S_0+R_0$ and $Y=S+R$, this requires verifying that the algebraic shear $[R_0, R] + \rho(S_0)R-\rho(S)R_0$ lies within the contingent derivative $DW(S|R)([S_0, S])$, ensuring the geometric capacity of the fiber accommodates the control-induced drift.
    \item \textbf{Spectral Analysis:} If the action of $\mf{s}$ on $\mf{r}$ is non-trivial, compute the spectrum of the relevant operators $\rho(S)$. Record any expanding or shearing component from that calculation. The FCC test remains the separate edge-invariance check of Step~4 and does not supply spectral compensation.
\end{enumerate}

\paragraph*{Computational Aspects and Foundational Value}

Applying the framework to a given system is demanding. Computing the System Lie Algebra $\g(G)$ means iterating commutators of superoperators on a $d^2$ dimensional space, so the ambient algebra $\g_{LK}$ has dimension $O(d^4)$; standard basis algorithms apply \cite{Graaf2000LieAT}, but the growth of the Hilbert space confines exact closure to few body systems, with the dimension of $\g(G)$ already a useful complexity indicator \cite{Schirmer2000Complete}. Determining $W(G)$ is harder still, bringing in the face structure, the dual cone, and the fiber variation of the non-reductive case, with tools from convex optimization and geometric control.

The framework identifies the structural Levi type and the FCC edge-invariance constraint on coupled wedge geometry. Complete positivity separately fixes the ambient physical cone. These data do not assign spectral behavior, a perturbative regime, a mean-field approximation, or a random-matrix model to a system.

\section{The Levi-Type Decomposition Theorem: Proof and Structure}
\label{sec:levi_type_theorem}

This section contains the central result of this work, a generalization of the Levi Decomposition Theorem to Dynamical Lie Wedge Pairs. The theorem splits the Dynamical Lie Wedge $W(G)$ along the Levi decomposition of the System Lie Algebra $\g(G)$, and this splitting supports the four structural Levi types used in the remainder of the paper. We build it in stages. Section~\ref{sec:proof_levi_type} attaches three components to a given wedge and works out their individual properties; Section~\ref{sec:reconstruction_levi} reverses direction and reconstructs a Lie wedge from abstract fiber data; Section~\ref{sec:decomposition_theorem_statement} assembles the two into the decomposition theorem. The wedge-theoretic vocabulary (wedges, edges, tangent cones, invariance) is that of the mathematical theory of Lie semigroups \cite{hilgert1989lie}. The key new element is the Fiber Compatibility Condition, which governs the coupling between the components in the semidirect product structure.

\subsection{Components of the Decomposition and Their Properties}
\label{sec:proof_levi_type}

We consider a DLWP $(\g, W)$. Let $\g = \mf{s} \ltimes \mf{r}$ be the Levi decomposition of $\g$, where $\mf{r}$ is the solvable radical and $\mf{s}$ is the semisimple Levi factor. Let $\pi_{\s}: \g \to \mf{s} \cong \g/\mf{r}$ be the canonical projection homomorphism, with $\Ker(\pi_{\s}) = \mf{r}$; we identify $\g$ with the vector space $\mf{s} \oplus \mf{r}$ via the fixed Levi section. We adapt the toolkit of the mathematical theory of semigroups \cite{hilgert1989lie} to the physical constraints of GKLS generators.

\begin{definition}[Components of the decomposition]
\label{def:fiber_cone}
Let $(\g, W)$ be a DLWP as above.
\begin{enumerate}
    \item The Radical Wedge is the intersection of $W$ with the solvable radical,
    \begin{equation}
    W_{\rr} = W \cap \mf{r}.
    \end{equation}
    \item The Quotient Cone is the canonical projection of $W$ onto the Levi factor,
    \begin{equation}
    C_{\text{quot}} = \pi_{\s}(W).
    \end{equation}
    \item For $S \in C_{\text{quot}}$, the fiber $W(S)$ is the nonempty closed convex set of radical components that complete $S$ to an element of the wedge,
    \begin{equation}
    W(S) = \{R \in \mf{r} \mid S+R \in W\}.
    \end{equation}
    The set-valued assignment $S \mapsto W(S)$ is the Fiber Map.
\end{enumerate}
\end{definition}

Geometrically, $S + W(S)$ is the slice of $W$ over the base point $S$, the intersection of $W$ with the affine subspace $S+\mf{r}$. The wedge is recovered from its slices by the tautological identity
\begin{equation}
\label{eq:fiber_union}
W = \bigcup_{S \in C_{\text{quot}}} (S + W(S)).
\end{equation}
All the structure therefore resides in the properties these components inherit from the Lie wedge $W$, which we now establish one component at a time.

The analysis of a given Lie wedge here, and the reconstruction from abstract data in the next subsection, use the standard tools for wedges and semidirect products \cite{hilgert1989lie}.

\paragraph{The radical wedge.}

We first characterize the intersection of $W$ with the radical $\mf{r}$. This component captures the part of the dynamics confined to the solvable ideal.

\begin{lemma}[Structure of the Radical Wedge]
\label{lem:radical_wedge}
The radical wedge $W_{\rr} = W \cap \mf{r}$ is a Lie wedge in the Lie algebra $\mf{r}$. Its edge is $E(W_{\rr}) = E(W) \cap \mf{r}$.
\end{lemma}

\begin{proof}
\textbf{1. Closed Convex Cone Structure:} $W_{\rr}$ is the intersection of a closed convex cone $W$ and a subspace $\mf{r}$. The intersection of closed convex sets is closed and convex. Thus, $W_{\rr}$ is a closed convex cone in $\mf{r}$.

\textbf{2. Identification of the Edge $E(W_{\rr})$:}
We analyze the definition of the edge:
\begin{align}
E(W_{\rr}) &= W_{\rr} \cap (-W_{\rr}) = (W \cap \mf{r}) \cap ((-W) \cap \mf{r}) \\
&= (W \cap (-W)) \cap (\mf{r} \cap \mf{r}) = E(W) \cap \mf{r}.
\end{align}
This identifies the reversible generators within the radical component.

\textbf{3. Verification of the Lie Wedge Condition:} We must show $e^{\ad(X)}(W_{\rr}) \subseteq W_{\rr}$ for all $X \in E(W_{\rr})$.

Let $X \in E(W_{\rr})$. By the identification above, $X \in E(W)$. Since $W$ is assumed to be a Lie wedge in $\g$, it satisfies the invariance condition: $e^{\ad(X)}(W) \subseteq W$.

We also use the fact that the radical $\mf{r}$ is an ideal in $\g$. Ideals are invariant under the adjoint action of the entire algebra: $e^{\ad(Y)}(\mf{r}) = \mf{r}$ for all $Y \in \g$. This follows directly from the ideal property: $\ad(Y)(\mf{r}) \subseteq \mf{r}$ gives $e^{\pm\ad(Y)}(\mf{r}) \subseteq \mf{r}$, and the two inclusions force equality. (The radical is in fact a characteristic ideal, preserved by all automorphisms, but the ideal property already suffices here.) In particular, this holds for $X \in E(W_{\rr})$.

The adjoint action $e^{\ad(X)}$ is an automorphism of $\g$. Automorphisms distribute over intersections:
\begin{equation}
e^{\ad(X)}(W_{\rr}) = e^{\ad(X)}(W \cap \mf{r}) = e^{\ad(X)}(W) \cap e^{\ad(X)}(\mf{r}).
\end{equation}

Substituting the invariance properties:
\begin{equation}
e^{\ad(X)}(W_{\rr}) = e^{\ad(X)}(W) \cap \mf{r}.
\end{equation}

Since $e^{\ad(X)}(W) \subseteq W$, the intersection with $\mf{r}$ satisfies the inclusion:
\begin{equation}
e^{\ad(X)}(W) \cap \mf{r} \subseteq W \cap \mf{r} = W_{\rr}.
\end{equation}

Therefore, $e^{\ad(X)}(W_{\rr}) \subseteq W_{\rr}$. $W_{\rr}$ satisfies the definition of a Lie wedge within the subalgebra $\mf{r}$.
\end{proof}

\paragraph{The quotient cone.}

We next analyze the projection of the wedge $W$ onto the selected semisimple Levi factor $\mf{s}$. This projection records quotient coordinates of admissible full lifts; it has no standalone spectral or mixing interpretation.

\begin{lemma}[Structure of the Quotient Cone]
\label{lem:quotient_cone}
The quotient cone $C_{\text{quot}} = \pi_{\s}(W)$ is a convex cone in $\mf{s}$ (it need not be closed; see Remark~\ref{rem:quotient_closedness}). The projected edge $E_{\text{quot}} = \pi_{\s}(E(W))$ is a subalgebra of $\mf{s}$ and is automatically closed. The quotient cone is invariant under the adjoint action of $E_{\text{quot}}$:
\begin{equation}
e^{\ad(S_0)}(C_{\text{quot}}) \subseteq C_{\text{quot}} \qquad \text{for all } S_0 \in E_{\text{quot}}.
\end{equation}
\end{lemma}

\begin{proof}
\textbf{1. Convex Cone Structure:} The projection $\pi_{\s}: \g \to \mf{s}$ is a linear map. The image of a convex cone under a linear map is always a convex cone.

\textbf{2. The Projected Edge:} The edge $E(W)$ is a subalgebra of $\g$, and $\pi_{\s}$ is a Lie algebra homomorphism, so $E_{\text{quot}} = \pi_{\s}(E(W))$ is a subalgebra of $\mf{s}$. As the linear image of a subspace it is itself a subspace, and every linear subspace of a finite-dimensional vector space is closed. No such argument is available for $C_{\text{quot}}$, whose closedness is a strictly stronger property (Remark~\ref{rem:quotient_closedness}).

\textbf{3. Invariance under $E_{\text{quot}}$:} We must show that $e^{\ad(S_0)}(C_{\text{quot}}) \subseteq C_{\text{quot}}$ for all $S_0 \in E_{\text{quot}}$.

Let $S_0 \in E_{\text{quot}}$. By definition of the projection, there exists $R_0 \in \mf{r}$ such that $X_0 = S_0+R_0 \in E(W)$.
Let $S \in C_{\text{quot}}$. Similarly, there exists $R \in \mf{r}$ such that $X = S+R \in W$.

Since $W$ is a Lie wedge and $X_0 \in E(W)$, the invariance condition holds: $e^{\ad(X_0)}(X) \in W$.

We analyze the projection of this transformed element onto $\mf{s}$. The projection $\pi_{\s}$ is not just a linear map; it is a Lie algebra homomorphism (since $\mf{r}$ is the kernel). A fundamental property of Lie algebra homomorphisms is that they intertwine the adjoint representations of the respective algebras. Infinitesimally:
\begin{equation}
\pi_{\s} \circ \ad(Y) = \ad(\pi_{\s}(Y)) \circ \pi_{\s} \quad \forall Y \in \g.
\end{equation}
This property extends to the exponentiated (group) action:
\begin{equation}
\pi_{\mathfrak{s}}\circ e^{\operatorname{ad}(Y)} = e^{\operatorname{ad}(\pi_{\mathfrak{s}}(Y))}\circ \pi_{\mathfrak{s}}\quad \forall\,Y\in\mathfrak{g}.
\end{equation}

Applying this to the elements $X_0$ and $X$:
\begin{equation}
\pi_{\s}(e^{\ad(X_0)}(X)) = e^{\ad(\pi_{\s}(X_0))}(\pi_{\s}(X)),
\end{equation}
and substituting the components:
\begin{equation}
\pi_{\s}(e^{\ad(X_0)}(X)) = e^{\ad(S_0)}(S).
\end{equation}

Since $e^{\ad(X_0)}(X) \in W$, its projection must belong to the image of $W$ under $\pi_{\s}$, which is $C_{\text{quot}}$. Thus, $e^{\ad(S_0)}(S) \in C_{\text{quot}}$. This establishes the invariance of $C_{\text{quot}}$ under the action of $E_{\text{quot}}$.
\end{proof}

\begin{remark}
The quotient cone is generally not a Lie wedge in $\mf{s}$. The Lie wedge condition would demand invariance under the edge of the wedge itself, $E(C_{\text{quot}})$, whereas $E_{\text{quot}}$ is generally only a subalgebra of $E(C_{\text{quot}})$. Thus $C_{\text{quot}}$ is merely a cone invariant under the action of the specific subalgebra $E_{\text{quot}}$. The distinction matters for the classification.
\end{remark}

\begin{remark}[Closedness of the quotient cone]
\label{rem:quotient_closedness}
Linear images of closed convex cones can fail to be closed, and $C_{\text{quot}}$ is no exception. A standard example is the closed convex cone $C = \{(x,y,z) \in \R^3 \mid x \ge 0,\ z \ge 0,\ xz \ge y^2\}$, whose image under the projection $(x,y,z) \mapsto (y,z)$ is the non-closed set $\{(y,z) \mid z > 0\} \cup \{(0,0)\}$. Closedness of $C_{\text{quot}}$ is accordingly a genuine additional hypothesis whenever an argument needs it; the decomposition theorem below does not. Three sufficient conditions cover the cases of interest.
\begin{enumerate}
    \item \textbf{Recession condition.} If $W \cap \mf{r} \subseteq E(W)$, then $\pi_{\s}(W)$ is closed. To see this, fix an inner product on $\g$ and let $S_n = \pi_{\s}(X_n) \to S$ with $X_n \in W$, where each $X_n$ is chosen of minimal norm in the closed set $W \cap (S_n + \mf{r})$. If some subsequence of $(X_n)$ is bounded, a limit point $X \in W$ has $\pi_{\s}(X) = S$, and we are done. Otherwise $\|X_n\| \to \infty$, and after passing to a subsequence $X_n/\|X_n\| \to Y \in W$ with $\|Y\| = 1$ and $\pi_{\s}(Y) = \lim_n S_n/\|X_n\| = 0$, so that $Y \in W_{\rr} \subseteq E(W)$ and hence $-Y \in W$. For $t \ge 0$ the element $X_n - tY$ then lies in $W \cap (S_n + \mf{r})$, while $\langle X_n, Y \rangle = \|X_n\|\,(1+o(1)) > 0$ for large $n$, so the choice $t = \langle X_n, Y\rangle$ yields $\|X_n - tY\|^2 = \|X_n\|^2 - \langle X_n, Y\rangle^2 < \|X_n\|^2$, contradicting minimality. In the GKLS setting this condition typically fails: the radical wedge contains genuinely dissipative rays, and these never lie in the edge.
    \item \textbf{Polyhedrality.} If $W$ is polyhedral, its linear image $\pi_{\s}(W)$ is again polyhedral, hence closed.
    \item \textbf{Full coherent control on the Levi factor.} If $E_{\text{quot}} = \mf{s}$, the chain $\mf{s} = E_{\text{quot}} \subseteq C_{\text{quot}} \subseteq \mf{s}$ forces $C_{\text{quot}} = \mf{s}$, a subspace and hence closed. This condition covers the depolarizing-noise example of Section~\ref{sec:example_reductive}, where $C_{\text{quot}} = \mf{su}(2)$.
\end{enumerate}
For edges, by contrast, no hypothesis is ever needed: $E_{\text{quot}} = \pi_{\s}(E(W))$ is the linear image of a subspace, and $E(C_{\text{quot}}) = C_{\text{quot}} \cap (-C_{\text{quot}})$ is the lineality space of a convex cone; both are subspaces and therefore closed in finite dimension. Corollary~\ref{cor:structure_edge} is unaffected by the possible non-closedness of $C_{\text{quot}}$.
\end{remark}

\paragraph{The fiber map.}

We turn to the structure transverse to the radical, encoded in the fiber map of Definition~\ref{def:fiber_cone}. Its basic geometric and topological properties are the following.

\begin{lemma}[Properties of the Fiber Map]
\label{lem:fiber_map_properties}
The Fiber Map $S \mapsto W(S)$ satisfies the following properties:
\begin{enumerate}
    \item $W(S)$ is a nonempty closed convex set in $\mf{r}$ for all $S \in C_{\text{quot}}$.
    \item $W(0) = W_{\rr}$. (The fiber above the origin is the Radical Wedge).
    \item Subadditivity: $W(S_1) + W(S_2) \subseteq W(S_1+S_2)$.
    \item Homogeneity (Scaling): $\lambda W(S) \subseteq W(\lambda S)$ for $\lambda \ge 0$. Since $W$ is a cone, equality holds for $\lambda > 0$.
    \item Closed graph: the graph $\Gamma = \{(S, R) \in \mf{s} \times \mf{r} \mid S \in C_{\text{quot}},\ R \in W(S)\}$ is closed in $\mf{s} \times \mf{r}$.
\end{enumerate}
\end{lemma}
\begin{proof}
These properties follow directly from the definition of the fiber and the fact that $W$ is a closed convex cone.
\begin{enumerate}
    \item The translate $S + W(S)$ is the intersection of $W$ with the affine subspace $S+\mf{r}$. Since $W$ is closed and convex, and the affine subspace is closed and convex, $W(S)$ is closed and convex. It is nonempty precisely because $S \in \pi_{\s}(W)$.
    \item $W(0) = \{R \in \mf{r} \mid 0+R \in W\} = W \cap \mf{r} = W_{\rr}$.
    \item Subadditivity follows because $W$ is a convex cone. Let $R_1 \in W(S_1)$ and $R_2 \in W(S_2)$. Then $X_1 = S_1+R_1 \in W$ and $X_2 = S_2+R_2 \in W$. Convexity gives $\tfrac{1}{2}(X_1+X_2) \in W$, and the cone property absorbs the factor $2$, so $X_1+X_2 \in W$.
    $X_1+X_2 = (S_1+S_2) + (R_1+R_2)$. By definition of the fiber above $S_1+S_2$, $R_1+R_2 \in W(S_1+S_2)$. Thus $W(S_1) + W(S_2) \subseteq W(S_1+S_2)$.
    \item Homogeneity follows from the cone property of $W$. If $R \in W(S)$, then $S+R \in W$. For $\lambda > 0$, $\lambda(S+R) = \lambda S + \lambda R \in W$. Thus $\lambda R \in W(\lambda S)$, so $\lambda W(S) \subseteq W(\lambda S)$. Conversely, if $R' \in W(\lambda S)$, then $\lambda S + R' \in W$. Since $W$ is a cone, $S + R'/\lambda \in W$. Thus $R'/\lambda \in W(S)$, so $R' \in \lambda W(S)$. Equality holds for $\lambda > 0$. For $\lambda=0$, $0 \cdot W(S) = \{0\} \subseteq W(0)=W_{\rr}$.
    \item Under the linear isomorphism $\mf{s} \times \mf{r} \to \g$, $(S,R) \mapsto S+R$, the graph $\Gamma$ is carried onto $W$ itself: if $S+R \in W$, then $S = \pi_{\s}(S+R) \in C_{\text{quot}}$ and $R \in W(S)$ by definition, and the converse inclusion is immediate. Linear isomorphisms of finite-dimensional spaces are homeomorphisms, so the closedness of $\Gamma$ in $\mf{s} \times \mf{r}$ is the same statement as the closedness of $W$. The topology of $C_{\text{quot}}$ plays no role here, which matters because $C_{\text{quot}}$ itself may fail to be closed (Remark~\ref{rem:quotient_closedness}).
\end{enumerate}
\end{proof}

In the language of convex analysis, Lemma~\ref{lem:fiber_map_properties} says that the fiber map is a closed convex process from $\mf{s}$ to $\mf{r}$, a set-valued map whose graph is a closed convex cone \cite[\S 39]{rockafellar1970convex}. Under the identification $(S,R) \mapsto S+R$ that graph is the wedge $W$ itself, and items (1)--(5) are the standard properties of such processes. The lemma is therefore convex-analytic bookkeeping. The substance of the decomposition lies elsewhere, in the interaction of this convex structure with the Lie bracket, which the Fiber Compatibility Condition isolates.

\paragraph{The Fiber Compatibility Condition.}

We come to the constraint that couples the components. It originates in the invariance condition defining a Lie wedge, and in the reductive case it degenerates: trivial coupling leaves the radical pointwise fixed under the edge, and the condition then asks only that the fiber map be constant along the edge orbits in the quotient cone.

\begin{proposition}[Necessity of the Fiber Compatibility Condition]
\label{lem:fiber_compatibility_main}
Let $W$ be a Lie wedge in $\g = \mf{s} \ltimes \mf{r}$. Then for every $X_0 = S_0+R_0 \in E(W)$ and every $S \in C_{\text{quot}}$,
\begin{equation}
\label{eq:FCC}
e^{\ad(X_0)}(S+W(S)) \subseteq e^{\ad(S_0)}(S) + W(e^{\ad(S_0)}(S)).
\end{equation}
We call \eqref{eq:FCC} the Fiber Compatibility Condition (FCC). The fiber on the right is well defined because $e^{\ad(S_0)}(S) \in C_{\text{quot}}$ by Lemma~\ref{lem:quotient_cone}.
\end{proposition}
\begin{proof}
We analyze the action of the edge $E(W)$ on an arbitrary element of the wedge $W$, decomposed into its fiber structure.

Let $X = S+R \in W$. By definition, $R \in W(S)$.
Let $X_0 = S_0+R_0 \in E(W)$. Since $W$ is a Lie wedge, the invariance condition holds: $X' = e^{\ad(X_0)}(X) \in W$.

We decompose $X'$ into its semisimple and radical components, $X' = S'+R'$.

\textbf{Semisimple Component $S'$:}
We use the intertwining property of the projection homomorphism $\pi_{\s}$ (as derived in Lemma \ref{lem:quotient_cone}):
\begin{equation}
S' = \pi_{\s}(X') = \pi_{\s}(e^{\ad(X_0)}(X)) = e^{\ad(\pi_{\s}(X_0))}(\pi_{\s}(X)) = e^{\ad(S_0)}(S).
\end{equation}
We know $S' \in C_{\text{quot}}$, as established in Lemma \ref{lem:quotient_cone}.

\textbf{Radical Component $R'$:}
The radical component is the remainder:
\begin{equation}
R' = X' - S' = e^{\ad(X_0)}(S+R) - e^{\ad(S_0)}(S).
\end{equation}
The explicit form of $R'$ involves the coupling $\rho$ through the semidirect-product adjoint action (infinitesimally, Eq. \eqref{eq:semidirect_bracket}) and the Baker--Campbell--Hausdorff series; in a non-reductive algebra the action of $e^{\ad(X_0)}$ does not split into independent actions on $\mf{s}$ and $\mf{r}$. For the present proof, the decomposition into components is all we need.

\textbf{Constraint Imposed by the Fiber Structure:}
Since $X' = S'+R' \in W$, by the definition of the fiber, the radical component $R'$ must belong to the fiber above the semisimple component $S'$:
\begin{equation}
R' \in W(S') = W(e^{\ad(S_0)}(S)).
\end{equation}
This is the constraint imposed by the geometry of the wedge: the radical component of the transformed element must lie within the fiber corresponding to the transformed semisimple component.

\textbf{Derivation of the FCC:}
We now consider the action of $e^{\ad(X_0)}$ on the entire fiber $S+W(S)$. The image is the set $e^{\ad(X_0)}(S+W(S))$.
We have shown that every element $X'$ in this image can be written as $X' = S'+R'$, where $S' = e^{\ad(S_0)}(S)$ and $R' \in W(S')$.

Therefore, the image of the fiber under the adjoint action must be contained within the fiber above the transformed base point $S'$:
\begin{equation}
e^{\ad(X_0)}(S+W(S)) \subseteq S' + W(S') = e^{\ad(S_0)}(S) + W(e^{\ad(S_0)}(S)).
\end{equation}
This is precisely \eqref{eq:FCC}.
\end{proof}

\subsection{Reconstruction from Fiber Data}
\label{sec:reconstruction_levi}

We now reverse direction and prove the converse, which is what makes the decomposition a genuine classification statement. The data are abstract; no ambient wedge is presupposed. In particular, the edge over which the invariance conditions are quantified must itself be defined from the data, and we do so explicitly. This removes the circularity that would arise if the data were required to be invariant under the edge of a wedge not yet constructed.

\begin{proposition}[Reconstruction of the Wedge]
\label{lem:sufficiency}
Let $\g = \mf{s} \ltimes \mf{r}$. Suppose we are given a convex cone $C \subseteq \mf{s}$ containing $0$, not assumed closed, together with an assignment of a set $W(S) \subseteq \mf{r}$ to each $S \in C$, subject to the following conditions:
\begin{enumerate}
    \item[(D1)] each $W(S)$ is a nonempty closed convex subset of $\mf{r}$;
    \item[(D2)] (subadditivity) $W(S_1)+W(S_2) \subseteq W(S_1+S_2)$ for all $S_1, S_2 \in C$;
    \item[(D3)] (homogeneity) $\lambda W(S) \subseteq W(\lambda S)$ for all $\lambda \ge 0$ and $S \in C$;
    \item[(D4)] (closed graph) the set $\Gamma = \{(S,R) \mid S \in C,\ R \in W(S)\}$ is closed in $\mf{s} \times \mf{r}$.
\end{enumerate}
Define from these data the candidate edge
\begin{equation}
\label{eq:reconstructed_edge}
\Ecal \;=\; \bigcup_{S \in C \cap (-C)} \bigl(S + (W(S) \cap -W(-S))\bigr),
\end{equation}
and assume in addition:
\begin{enumerate}
    \item[(D5)] (base invariance) $e^{\ad(\pi_{\s}(X_0))}(C) \subseteq C$ for every $X_0 \in \Ecal$;
    \item[(D6)] (compatibility) every $X_0 \in \Ecal$ and every $S \in C$ satisfy the FCC \eqref{eq:FCC}, with $W(\cdot)$ the given assignment.
\end{enumerate}
Then
\begin{equation}
W \;=\; \bigcup_{S \in C} (S + W(S))
\end{equation}
is a Lie wedge in $\g$ with
\begin{equation}
W \cap \mf{r} = W(0), \qquad \pi_{\s}(W) = C, \qquad E(W) = \Ecal,
\end{equation}
and the fiber map of $W$ in the sense of Definition~\ref{def:fiber_cone} is the given assignment $S \mapsto W(S)$. No other subset of $\g$ has these components.
\end{proposition}

\begin{proof}
\textbf{1. The set $W$ is a closed convex cone with the stated components.}

\textit{Scaling:} Let $X = S+R \in W$, so $S \in C$ and $R \in W(S)$, and let $\lambda > 0$. Then $\lambda X = \lambda S + \lambda R$ with $\lambda S \in C$ (a cone) and $\lambda R \in \lambda W(S) \subseteq W(\lambda S)$ by (D3), so $\lambda X \in W$. For $\lambda = 0$, conditions (D1) and (D3) with $S = 0$ give $\{0\} = 0 \cdot W(0) \subseteq W(0)$, hence $0 \in W$.

\textit{Addition:} Let $X_1 = S_1+R_1$ and $X_2 = S_2+R_2$ lie in $W$. Then $S_1+S_2 \in C$ by convexity of the cone $C$, and $R_1+R_2 \in W(S_1)+W(S_2) \subseteq W(S_1+S_2)$ by (D2), so $X_1+X_2 \in W$. A cone closed under addition is convex.

\textit{Closedness:} Under the linear isomorphism $\mf{s} \times \mf{r} \to \g$, $(S,R) \mapsto S+R$, the set $W$ is the image of the graph $\Gamma$. Linear isomorphisms of finite-dimensional spaces are homeomorphisms, so (D4) makes $W$ closed. The hypothesis is closedness of $\Gamma$ in all of $\mf{s} \times \mf{r}$; closedness in $C \times \mf{r}$ would be weaker and would not suffice when $C$ fails to be closed, and no closedness of $C$ is assumed.

\textit{Components:} If $X \in W \cap \mf{r}$, then writing $X = S+R$ with $S \in C$, $R \in W(S)$ and projecting gives $S = \pi_{\s}(X) = 0$, so $X = R \in W(0)$; conversely $0 + W(0) \subseteq W$. Thus $W \cap \mf{r} = W(0)$. Each fiber is nonempty by (D1), so $\pi_{\s}(W) = C$. Distinct base points contribute slices over distinct points of $\mf{s}$, so the slice of $W$ over each $S \in C$ is exactly $S + W(S)$, and the fiber map of $W$ is the given assignment.

\textbf{2. Identification of the edge.}
An element $X = S+R$ lies in $E(W) = W \cap (-W)$ precisely when $S \in C$, $R \in W(S)$, $-S \in C$, and $-R \in W(-S)$. Comparing with \eqref{eq:reconstructed_edge} gives $E(W) = \Ecal$. The invariance hypotheses (D5)--(D6), formulated over the data-defined set $\Ecal$, are therefore exactly invariance hypotheses over $E(W)$.

\textbf{3. The Lie wedge condition.}
Let $X_0 = S_0+R_0 \in E(W)$ and $X = S+R \in W$, and set $X' = e^{\ad(X_0)}(X)$. The intertwining relation $\pi_{\s} \circ e^{\ad(X_0)} = e^{\ad(S_0)} \circ \pi_{\s}$ (proof of Lemma~\ref{lem:quotient_cone}) shows that the semisimple component of $X'$ is $S' = e^{\ad(S_0)}(S)$, and $S' \in C$ by (D5). The FCC (D6) gives
\begin{equation}
e^{\ad(X_0)}(S+W(S)) \subseteq S' + W(S') \subseteq W.
\end{equation}
Since $X \in S + W(S)$, we conclude $X' \in W$. Hence $e^{\ad(X_0)}(W) \subseteq W$ for every $X_0 \in E(W)$, and $W$ satisfies the Lie wedge condition \eqref{eq:lie_wedge_condition}.

\textbf{4. Uniqueness.}
Any subset of $\g$ with quotient projection $C$ and slice $S + W(S)$ over each $S \in C$ equals $\bigcup_{S \in C}(S+W(S))$, which is $W$.
\end{proof}

\subsection{The Decomposition Theorem}
\label{sec:decomposition_theorem_statement}

The pieces now assemble into the main theorem of this work.

\begin{theorem}[Levi-Type Decomposition Theorem]
\label{thm:levi_type_main}
Let $(\g, W)$ be a Dynamical Lie Wedge Pair and let $\g = \mf{s} \ltimes \mf{r}$ be the Levi decomposition.
\begin{enumerate}
    \item \textbf{(Decomposition)} The components of Definition~\ref{def:fiber_cone} satisfy: $W_{\rr}$ is a Lie wedge in $\mf{r}$ with edge $E(W_{\rr}) = E(W) \cap \mf{r}$ (Lemma~\ref{lem:radical_wedge}); $C_{\text{quot}}$ is a convex cone in $\mf{s}$, in general not closed, invariant under the adjoint action of the closed subalgebra $E_{\text{quot}} = \pi_{\s}(E(W))$ (Lemma~\ref{lem:quotient_cone}); the fiber map has nonempty closed convex fibers, satisfies $W(0) = W_{\rr}$, is subadditive and positively homogeneous, and has closed graph in $\mf{s} \times \mf{r}$ (Lemma~\ref{lem:fiber_map_properties}); and the Fiber Compatibility Condition \eqref{eq:FCC} holds (Proposition~\ref{lem:fiber_compatibility_main}).
    \item \textbf{(Reconstruction)} Conversely, every datum $(C,\, S \mapsto W(S))$ satisfying (D1)--(D6) of Proposition~\ref{lem:sufficiency} defines the Lie wedge $W = \bigcup_{S \in C}(S+W(S))$, the unique subset of $\g$ with the given components.
    \item \textbf{(Correspondence)} The two constructions are mutually inverse: decomposing a Lie wedge and reconstructing from the resulting data returns the same wedge, and reconstructing from admissible data and decomposing returns the same data.
\end{enumerate}
\end{theorem}

\begin{proof}
Part 1 collects Lemmas~\ref{lem:radical_wedge}--\ref{lem:fiber_map_properties} and Proposition~\ref{lem:fiber_compatibility_main}. For Parts 2 and 3, first note that the data extracted from a Lie wedge $W$ by Definition~\ref{def:fiber_cone} satisfy the hypotheses of Proposition~\ref{lem:sufficiency} with $C = C_{\text{quot}}$: conditions (D1)--(D4) are Lemma~\ref{lem:fiber_map_properties}, the set $\Ecal$ of \eqref{eq:reconstructed_edge} coincides with $E(W)$ by the slice description of $W$, and then (D5) is Lemma~\ref{lem:quotient_cone} and (D6) is Proposition~\ref{lem:fiber_compatibility_main}. Part 2 is Proposition~\ref{lem:sufficiency}. The composite decompose-then-reconstruct returns $W$ by the identity \eqref{eq:fiber_union}; the composite reconstruct-then-decompose returns the data by the component identities $W \cap \mf{r} = W(0)$, $\pi_{\s}(W) = C$, and the fiber identification, all established in Proposition~\ref{lem:sufficiency}.
\end{proof}

\begin{remark}[Typing and dependence on the Levi section]
Parts 2 and 3 operate at the level of Lie wedges. The reconstructed object is a Lie wedge in $\g$, so the decomposition applies to an arbitrary Lie wedge in an algebra with a fixed splitting $\g = \mf{s} \ltimes \mf{r}$; for the DLWP of a physical system, that wedge is $W(G)$. The data $(C_{\text{quot}}, S \mapsto W(S))$ are defined relative to the chosen Levi section. The section is unique up to conjugation by $e^{\ad z}$ with $z$ in the nilradical (Malcev), and such a conjugation transports decomposition data to decomposition data. The decomposition is canonical relative to a Levi splitting.
\end{remark}

\begin{corollary}[Structure of the Edge]
\label{cor:structure_edge}
The edge $E(W)$ is characterized by the fiber decomposition over the edge of the quotient cone $E(C_{\text{quot}})$. An element $X=S+R$ belongs to $E(W)$ if and only if $S \in E(C_{\text{quot}})$ and the fibers satisfy the antisymmetry condition: $R \in W(S)$ and $-R \in W(-S)$. The edge is reconstructed as:
\begin{equation}
E(W) = \bigcup_{S \in E(C_{\text{quot}})} (S + (W(S) \cap -W(-S))).
\end{equation}
This highlights that the reversible dynamics are constrained by the fiber map even over the semisimple component.
\end{corollary}

\subsection{When the Fibered Form Reduces}
\label{sec:collapse_dichotomy}

A constant fiber map replaces the set-valued datum by a single cone and gives a genuine
dimension reduction. This subsection characterizes that collapse and its obstruction. In the
general case, Theorem~\ref{thm:levi_type_main} is an exact change of coordinates: under the
linear isomorphism $(S,R) \mapsto S+R$, the graph of the fiber map is the wedge itself, so the
data $(\Cquot, S \mapsto W(S))$ and the wedge $W$ determine one another without reducing
complexity.

\begin{proposition}[Collapse criterion]
\label{prop:fiber_collapse}
Let $(\g, W)$ be a DLWP with Levi decomposition $\g = \s \ltimes \rr$. If
$\Cquot \subseteq E(W)$, then the fiber map is constant,
\begin{equation}
W(S) \;=\; \Wrad \qquad \text{for every } S \in \Cquot,
\end{equation}
and consequently
\begin{equation}
\label{eq:collapsed_wedge}
W \;=\; \Cquot + \Wrad .
\end{equation}
In this case $W$ is determined by the pair $(\Cquot, \Wrad)$, two cones in complementary
factors. When both Levi factors are nonzero, each factor has strictly smaller ambient dimension
than $\g$, and the reconstruction of Proposition~\ref{lem:sufficiency} is a genuine reduction.
\end{proposition}

\begin{proof}
Let $S \in \Cquot \subseteq E(W)$, so that $S \in W$ and $-S \in W$; in particular
$0 \in W(S) \cap W(-S)$. For $R \in W(S)$ the element $S+R$ lies in $W$, and adding
$-S \in W$ gives $R = (S+R) + (-S) \in W$ because $W$ is closed under addition. Since
$R \in \rr$, this places $R \in W \cap \rr = \Wrad$. Conversely $S \in W$ and
$\Wrad \subseteq W$ give $S + \Wrad \subseteq W$, so $\Wrad \subseteq W(S)$. Hence
$W(S) = \Wrad$, and \eqref{eq:collapsed_wedge} is \eqref{eq:fiber_union} with the constant
fiber substituted.
\end{proof}

\begin{proposition}[Obstruction to collapse]
\label{prop:fiber_obstruction}
Let $(\g,W)$ be a DLWP with $W \subseteq \WGKLS$. If the fiber map is constant, then
\begin{equation}
\Cquot \;\subseteq\; W \cap \s \;\subseteq\; E(\WGKLS),
\end{equation}
so every base point is a purely Hamiltonian generator. Equivalently, in contrapositive form:
\emph{if some $S \in \Cquot$ has $K(S) \neq 0$, the fiber map is not constant}.
\end{proposition}

\begin{proof}
Constancy gives $W(S) = \Wrad \ni 0$, so $S = S + 0 \in W$ for every $S \in \Cquot$, whence
$\Cquot \subseteq W \cap \s$. Corollary~\ref{cor:hamiltonian_rigidity}(1) supplies the second
inclusion.
\end{proof}

\begin{remark}[Neither condition is an equivalence]
\label{rem:collapse_not_iff}
The hypothesis of Proposition~\ref{prop:fiber_collapse} is sufficient but not necessary, and
it cannot be weakened to $\Cquot \subseteq W$. Two examples mark the gap.
\begin{enumerate}
  \item \emph{Constant fibers without $\Cquot \subseteq E(W)$.} In the reductive algebra
  $\g = \su(2) \oplus \R\,\Lcal_D$ of Section~\ref{sec:example_reductive}, take
  $G$ to consist of a \emph{pointed} cone $C$ of Hamiltonian directions together with
  $\Lcal_D$. Then $\cone(G)$ is pointed, so $E(W) = \{0\}$, the Lie wedge condition is vacuous
  and $W = C + \R_{\ge 0}\Lcal_D$ has constant fibers, while $\Cquot = C \not\subseteq \{0\}$.
  \item \emph{$\Cquot \subseteq W$ without constant fibers.} In the reductive algebra
  $\g=\su(2)\oplus\R z$, with $z$ central, choose a pointed Lie-generating cone
  $C=\cone\{S_1,S_2,S_3\}\subset\su(2)$ and set
  \[
    W=\cone\{(S_1,0),(S_2,0),(S_3,0),(S_1,2z)\}.
  \]
  The zero lifts give $\Cquot=C\subseteq W$. Pointedness of $C$ forces
  $W\cap\R z=\{0\}$ and makes $E(W)=\{0\}$, so this closed polyhedral cone is a Lie wedge
  with vacuous edge invariance. Nevertheless $W(S_1)$ contains the segment $[0,2]z$, whereas
  $W(0)=\Wrad=\{0\}$, and the fibers are not constant.
\end{enumerate}
What survives is the one-way collapse criterion of Proposition~\ref{prop:fiber_collapse} and
the separate nonconstancy test of Proposition~\ref{prop:fiber_obstruction}. The pointwise test
$K(S) \neq 0$ detects fiber motion; it does not by itself show that the FCC is nontrivial,
because a pointed wedge may have zero edge and hence a vacuous edge-invariance condition.
\end{remark}

\begin{example}[The obstruction is realized: $\WGKLS$ in $\gLK(2)$]
\label{ex:fiber_mobility_glk2}
Take $\g = \gLK(2) \cong \gl(3,\R) \ltimes \R^3$ in the affine model of
Theorem~\ref{thm:structure_glk}, with $\s = \sll(3,\R)$, $\rr = \R\,\Id_{\Htr} \ltimes \R^3$,
and $W = \WGKLS(2)$. Base points are traceless Bloch matrices, and by
Lemma~\ref{lem:gkls_edge} such a base point is Hamiltonian precisely when it is
antisymmetric. Since $\Cquot = \sll(3,\R) \not\subseteq \so(3)$,
Proposition~\ref{prop:fiber_obstruction} already tells us the fiber map varies. It does so
computably. Along the radial direction of the fiber, write $S_t = t\cdot\mathrm{diag}(1,1,-2)$
and ask which multiples of $\Id_{\Htr}$ complete $S_t$ to an admissible generator:
\begin{equation}
\label{eq:fiber_mobility_threshold}
\sup\{\lambda \in \R \mid S_t + \lambda\,\Id_{\Htr} \in \WGKLS\} \;=\; -4t
\qquad (t \ge 0),
\end{equation}
In the Hilbert--Schmidt orthonormal Pauli basis $F_j=\sigma_j/\sqrt{2}$, a real diagonal
Kossakowski matrix $K=\operatorname{diag}(k_x,k_y,k_z)$ has Bloch linear part
\[
Q=\operatorname{diag}(-k_y-k_z,-k_x-k_z,-k_x-k_y)
  =K-(\Tr K)\mathbb I_3.
\]
The target $Q=S_t+\lambda\mathbb I_3$ is symmetric and has $b=0$, so its canonical
Kossakowski matrix is real diagonal. Solving for it gives
\[
K=Q-\frac{1}{2}(\Tr Q)\mathbb I_3
  =\operatorname{diag}\left(
      t-\frac{\lambda}{2},
      t-\frac{\lambda}{2},
      -2t-\frac{\lambda}{2}
    \right).
\]
Thus $K\succeq0$ if and only if $\lambda\le 2t$ and $\lambda\le-4t$. Since $t\ge0$,
this reduces to $\lambda\le-4t$. At the endpoint,
$K=\operatorname{diag}(3t,3t,0)\succeq0$, so the supremum is attained.
By contrast, $0 \in W(0) = \Wrad$. The fibers therefore recede linearly as the base point moves
away from the Hamiltonian directions: an anisotropic base point must be paid for with a
minimum amount of isotropic damping. The mechanism is
Lemma~\ref{lem:dissipation_functional} in Bloch coordinates, $\Tr Q = -2\Tr K$, which lets a
traceless base point lift with zero radical component only when $K = 0$.
\end{example}

\begin{remark}[Coordinate slices and fiber variation]
\label{rem:decomposition_reading}
Corollary~\ref{cor:hamiltonian_rigidity} is what singles the Levi splitting out from the other
ideals one might split along. For an admissible lift $X=S+R\in W$, it gives
$\varphi(X)=\varphi(R)$, so all dissipation strength is carried by the radical coordinate. This
statement does not make either coordinate projection a physical generator on its own. The cone
$\Cquot=\pi_{\s}(W)$ records the Levi coordinates of admissible lifts; only the slice
$W\cap\mf{s}$ consists of standalone semisimple generators, and
Corollary~\ref{cor:hamiltonian_rigidity}(1) makes those generators Hamiltonian. Likewise,
$\Wrad=W\cap\mf{r}$ is the separately admissible radical slice, whereas the radical coordinate
$R$ of a general lift need not belong to $\Wrad$.
Proposition~\ref{prop:fiber_collapse} gives a sufficient condition for independence: the fibers
are constant when the quotient cone lies in the edge. Proposition~\ref{prop:fiber_obstruction}
gives a different, one-way conclusion: a non-Hamiltonian quotient coordinate forces the fiber
map to vary. Neither statement makes a nontrivial action $\rho$ necessary for fiber motion, and
fiber motion alone does not make the FCC nonvacuous.
\end{remark}

\subsection{The Significance and Interpretation of the Fiber Compatibility Condition}

For data satisfying D1--D5 of Theorem~\ref{thm:levi_type_main}, the Fiber Compatibility Condition D6 is the remaining necessary and sufficient condition for the reconstructed cone to have the Lie-wedge edge invariance. It constrains the coupling between the semisimple and solvable components of the dynamics.

\subsubsection{Geometric Interpretation: Adaptation of the Fibers}

The FCC dictates how the geometry of the closed convex fibers $W(S)$ must adapt to the adjoint action generated by the edge $E(W)$.

In a non-reductive algebra the adjoint action couples the components through $\rho$. The coupling can shear or rotate the space, and a fixed fiber assignment can fail the Lie wedge condition under that action. Remark~\ref{rem:collapse_not_iff} gives a pointed reductive $\su(2)\oplus\R z$ example showing that inclusion of the quotient cone does not force constant fibers.

When the edge action is nontrivial, the FCC makes the fibers transform compatibly as $S$ varies. Their coordinated change of shape and position counteracts the adjoint action, so global convexity and the Lie wedge property survive. Non-reductive coupling can cause this motion, but it is not necessary: the pointed reductive example of Remark~\ref{rem:collapse_not_iff} has varying fibers with $\rho=0$.


\begin{figure}[h]
    \centering
    \begin{tikzpicture}[scale=1.8,
        axis/.style={->, thick},
        fiber_S/.style={blue, thick},
        fiber_S_fill/.style={blue!20, opacity=0.6},
        fiber_Sp/.style={red, thick},
        fiber_Sp_fill/.style={red!20, opacity=0.6},
        image/.style={black!60, thick, densely dashed},
        image_fill/.style={black!15, opacity=0.8},
        action/.style={->, >=stealth, thick, black!70}
    ]
        \draw[axis] (0,0) -- (4.5,0) node[right] {$\mf{s}$ (Base)};
        \draw[axis] (0,0) -- (0,3.2) node[above] {$\mf{r}$ (Fiber)};

        \coordinate (S) at (1,0);
        \coordinate (Sp) at (3.5,0);
        \fill (S) circle (1.5pt) node[below] {$S$};
        \fill (Sp) circle (1.5pt) node[below] {$S'$};

        \draw[action, dashed] (S) to[bend left=15] node[above, midway, yshift=2pt] {$e^{\ad(S_0)}$} (Sp);

        \coordinate (S_L) at (0.6, 1.8);
        \coordinate (S_R) at (1.4, 1.8);
        
        \coordinate (Img_L) at (3.0, 2.2);
        \coordinate (Img_R) at (4.0, 1.4);

        \coordinate (Sp_L) at (2.9, 2.3);
        \coordinate (Sp_R) at (4.1, 1.3);

        \draw[fiber_S] (S) -- (S_L);
        \draw[fiber_S] (S) -- (S_R);
        \fill[fiber_S_fill] (S) -- (S_L) -- (S_R) -- cycle;
        \node[blue] at (1, 2.1) {$W(S)$};

        \draw[action] (S_L) to[bend right=8] (Img_L);
        \draw[action] (S_R) to[bend left=8] (Img_R);
        \node[action] at (2.25, 2.5) {Action $e^{\ad(X_0)}$};
        
        \draw[fiber_Sp] (Sp) -- (Sp_L);
        \draw[fiber_Sp] (Sp) -- (Sp_R);
        \fill[fiber_Sp_fill] (Sp) -- (Sp_L) -- (Sp_R) -- cycle;
        \node[red] at (4.2, 2.1) {$W(S')$};

        \draw[image] (Sp) -- (Img_L);
        \draw[image] (Sp) -- (Img_R);
        \fill[image_fill] (Sp) -- (Img_L) -- (Img_R) -- cycle;
        \node[black!60] at (3.5, 0.8) {Image};


    \end{tikzpicture}
    \caption{A schematic of the Fiber Compatibility Condition and the Shear-and-Compensate mechanism. The adjoint action $e^{\ad(X_0)}$ moves the base point $S$ to $S'$ and shears the fiber $W(S)$ (blue) into the dashed image. The FCC compensates by requiring the fiber $W(S')$ (red) to contain that sheared image, so the wedge stays admissible under the action.}
    \label{fig:fcc_visualization}
\end{figure}
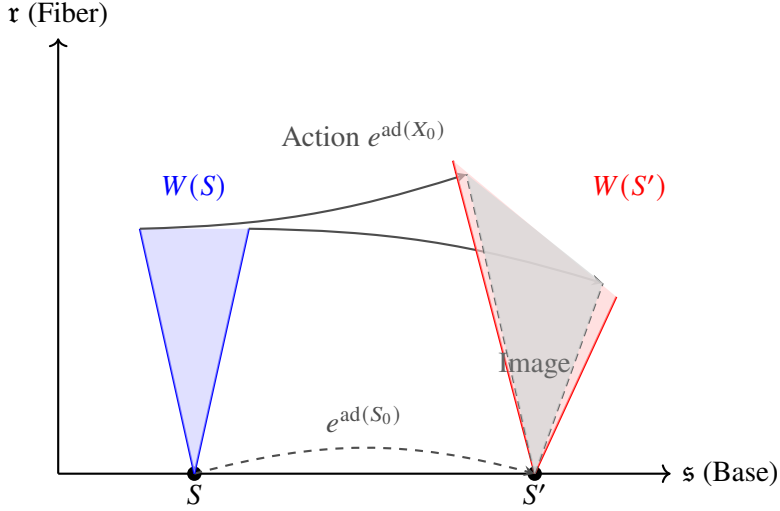

\subsubsection{Infinitesimal FCC and Subtangency}

The FCC (Equation \eqref{eq:FCC}) is expressed in terms of the global (exponentiated) adjoint action. It is instructive to analyze its infinitesimal counterpart, which provides a more direct link to the algebraic coupling $\rho$ and the Lie bracket structure. This analysis relies on the subtangency condition (Equation \eqref{eq:subtangency}), which states that $W$ is a Lie wedge if and only if $[X, Y] \in T_Y(W)$ for all $X \in E(W), Y \in W$.

Let $X_0 = S_0+R_0 \in E(W)$ and $Y = S+R \in W$. We analyze the constraint imposed by the subtangency condition on the commutator $[X_0, Y]$.

The commutator is given by the semidirect product structure (Equation \eqref{eq:semidirect_bracket}):
\begin{equation}
\label{eq:commutator_infinitesimal_fcc}
[X_0, Y] = \underbrace{([S_0, S])}_{\in \mf{s}} + \underbrace{([R_0, R] + \rho(S_0)(R) - \rho(S)(R_0))}_{\in \mf{r}}.
\end{equation}

We must determine the tangent cone $T_Y(W)$ at $Y=S+R$. Since $W$ is reconstructed via the fiber map, the tangent cone structure is determined by the local geometry of the fibers. The tangent cone $T_Y(W)$ involves both the tangent cone to the base $T_S(C_{\text{quot}})$ and the tangent cone to the fiber $T_R(W(S))$, combined in a way that depends on the local behavior of the fiber map $S \mapsto W(S)$; Proposition~\ref{prop:infinitesimal_fcc} below makes this precise. Section~\ref{sec:connection_neeb_fcc_glk} compares this edge-invariance condition with Neeb's full-algebra invariant-cone setting.

The subtangency condition requires $[X_0,Y]\in T_Y(W)$, so the components of the commutator must be compatible with the fiber structure at $Y$. The radical contribution in Eq.~\eqref{eq:commutator_infinitesimal_fcc} consists of the internal bracket $[R_0,R]$ and the semidirect-action terms $\rho(S_0)(R)-\rho(S)(R_0)$. Either contribution can generate a direction incompatible with the convex geometry of $W$. The tangent cone $T_Y(W)$ encodes the local geometry of the fiber map $S \mapsto W(S)$. The infinitesimal FCC requires the fibers to accommodate all such algebraic flows. If the radical contribution points outside the local tangent cone defined by the fiber structure, the subtangency condition, and hence the global FCC, fails.

The following proposition makes this infinitesimal requirement exact, recasting it through the contingent derivative of the fiber map.

\begin{proposition}[Infinitesimal FCC]\label{prop:infinitesimal_fcc}
Let $W$ be a closed Lie wedge with fiber data $(C_{\text{quot}},\, S \mapsto W(S))$, and let $\Gamma = \{(S, R) \mid S \in C_{\text{quot}},\, R \in W(S)\} \subseteq \mf{s} \times \mf{r}$ be the graph of the fiber map. For $(S, R) \in \Gamma$, define the contingent derivative of the fiber map, in the sense of set-valued analysis \cite[Ch.~5]{aubin1990setvalued}, through the tangent cone of the graph:
\begin{equation}
\label{eq:contingent_derivative}
DW(S|R)(V_S) \;=\; \{\, V_R \in \mf{r} \;\mid\; (V_S, V_R) \in T_{(S,R)}(\Gamma) \,\}.
\end{equation}
Then, by Proposition~\ref{prop:subtangency}, the invariance of $W$ under the edge action (the global FCC) is equivalent to the subtangency condition $[X_0, Y] \in T_Y(W)$ for all $X_0 \in E(W)$, $Y \in W$, and at $Y = S+R$, with $X_0 = S_0 + R_0$, subtangency takes the exact componentwise form
\begin{equation}
\label{eq:infinitesimal_fcc_display}
[R_{0}, R] + \rho(S_{0})R - \rho(S)R_{0} \;\in\; DW(S|R)([S_{0}, S]).
\end{equation}
The membership \eqref{eq:infinitesimal_fcc_display} contains the base constraint as well: it forces $[S_0, S] \in T_S(C_{\text{quot}})$.
\end{proposition}

\begin{proof}
By Proposition~\ref{prop:subtangency}, invariance of the closed wedge $W$ under $e^{\ad(X_0)}$ for all $X_0 \in E(W)$ is equivalent to the subtangency condition \eqref{eq:subtangency}. The map $(S, R) \mapsto S + R$ is a linear isomorphism $\mf{s} \times \mf{r} \to \g$ carrying $\Gamma$ bijectively onto $W$ by the slice description \eqref{eq:fiber_union}, and a linear isomorphism carries tangent cones to tangent cones, so $(V_S, V_R) \in T_{(S,R)}(\Gamma)$ if and only if $V_S + V_R \in T_{S+R}(W)$. Decomposing $[X_0, Y]$ into components by \eqref{eq:commutator_infinitesimal_fcc} and applying this identification yields \eqref{eq:infinitesimal_fcc_display}. For the last assertion, if $(V_S, V_R) \in T_{(S,R)}(\Gamma)$, pick sequences $(S_n, R_n) \in \Gamma$ and $h_n \downarrow 0$ with $(S_n - S)/h_n \to V_S$; since $S_n \in C_{\text{quot}}$, this exhibits $V_S \in T_S(C_{\text{quot}})$.

In a fixed direction $u$ one has the inner estimate $\limsup_{h \downarrow 0}\,\bigl(W(S + h u) - R\bigr)/h \subseteq DW(S|R)(u)$. The graph form \eqref{eq:contingent_derivative} is the correct general notion, since it also captures tangent directions attained only along perturbed base directions, as happens at curved parts of the boundary of $C_{\text{quot}}$ where $S + hu$ may leave $C_{\text{quot}}$ for every $h > 0$ while subtangency still holds. The graph contingent cone is no larger than what subtangency controls, because $T_{(S,R)}(\Gamma)$ is built from sequences inside $\Gamma$ and the isomorphism $(S,R) \mapsto S+R$ identifies it with $T_{S+R}(W)$, the exact object the subtangency condition constrains.
\end{proof}

The infinitesimal perspective shows that the FCC constrains compatibility between the full Lie bracket and the convex geometry of the wedge; its radical contribution includes both $[R_0,R]$ and the semidirect coupling terms.

\subsubsection{Physical Interpretation: Positivity Constraints on Coupled Dynamics}

In an open quantum system the wedge $W$ consists of admissible GKLS generators. Relative to the selected Levi splitting, $C_{\text{quot}}$ is the cone of Levi coordinates of admissible lifts, while $W_{\rr}$ is the radical slice. For a general lift $X=S+R\in W$, neither $S$ nor $R$ need belong to $W$ separately; the standalone slices are $W\cap\mf{s}$ and $W\cap\mf{r}$, with the former Hamiltonian by Corollary~\ref{cor:hamiltonian_rigidity}(1). The geometric adaptation described above may involve the internal radical bracket, the coupling $\rho$, or both. In the action-coupled subclass it takes the form of a ``Shear-and-Compensate'' mechanism, in which a reversible edge action shears the radical coordinate through $\rho$. If the fibers do not accommodate that action, an element of the declared wedge is carried outside it. The FCC requires the fiber geometry to compensate for the available reversible action; fixed-basis dephasing under basis-rotating control gives the canonical violation in Section~\ref{sec:counterexamples}.

The adaptation of the closed convex fibers $W(S)$ translates directly to constraints on the parameters of the physical generators, which we now express through the Kossakowski matrix.

\paragraph{The FCC and the Kossakowski Matrix.}
Recall that the wedge $W(G)$ is composed of valid GKLS generators, $W(G) \subseteq \WGKLS$. Each element $X=S+R \in W(G)$ is parameterized by a Hamiltonian $H_X$ and a Kossakowski matrix $K_X$. The fundamental physical constraint is $K_X \ge 0$.

The Kossakowski data belong to the full lift $X$, not to the coordinates $S$ and $R$ separately. As the base coordinate $S$ varies, the allowed full lifts $S+R$ with $R\in W(S)$ must remain physically admissible. The FCC separately dictates how these fibers move under the dynamics generated by the edge $E(W)$.

By Lemma \ref{lem:gkls_edge}, the edge of any physical Lie wedge $E(W(G))$ consists exclusively of Hamiltonian generators: reversibility ($\Lcal, -\Lcal \in \WGKLS$) forces $K=0$, so reversible dissipation is impossible within the GKLS framework.

Therefore, for any $X_0 \in E(W(G))$, the action $e^{\ad(X_0)}$ corresponds precisely to unitary conjugation. Let $X_0 \in E(W)$. The algebraic action $e^{\ad(X_0)}$ transforms the generator $X$ to $X'$.
\begin{equation}
X \mapsto X' = e^{\ad(X_0)}(X) \implies (H_X, K_X) \mapsto (H_{X'}, K_{X'}).
\end{equation}

As noted previously, unitary conjugation transforms the Kossakowski matrix by $\ast$-congruence with a unitary matrix, $K_{X'} = U K_X U^\dagger$. For unitary $U$ this is a similarity, since $U^\dagger = U^{-1}$, so the spectrum is preserved, and $K_X \ge 0$ implies $K_{X'} \ge 0$. The positivity of individual generators is automatically preserved under the edge action.

Why then is the FCC necessary?

The edge action preserves each generator's positivity but need not preserve the collective geometry that the fiber map $S \mapsto W(S)$ defines (Figure~\ref{fig:fcc_visualization}). Were the FCC violated, the edge would carry elements of $W(G)$ outside the wedge, the failure that Section~\ref{sec:counterexamples} exhibits.

In GKLS applications the condition $K\geq0$ restricts the available fibers. Separately, the FCC links the coupling $\rho$ to their motion under the edge action. Section~\ref{sec:neebs_theory} compares that edge-invariance statement with Neeb's invariant-cone condition and keeps spectral stability separate.

\subsection{Illustrative Examples of the Decomposition}

To illustrate the application of the Levi-Type Decomposition Theorem, we analyze fundamental examples of open quantum systems. These examples demonstrate how the framework distinguishes solvable and mixed reductive types.

\subsubsection{Example 1: Solvable-Type Dissipation (Cascaded Dynamics)}
\label{sec:example_integrable}

We analyze a cascaded qutrit under a control specification with two independently actuable decay channels. Its system algebra is solvable, while activating both channels simultaneously at positive rates generates hierarchical decay. The calculation establishes these two properties separately.

\paragraph{System Specification: Cascaded Decay in a Qutrit.}

Consider a 3 level system (qutrit) in a ladder configuration with states $|1\rangle, |2\rangle, |3\rangle$. The dynamics describe a cascade decay process: $|3\rangle \to |2\rangle \to |1\rangle$.
We assume purely dissipative dynamics ($H=0$).
Lindblad operators: $L_1 = \sqrt{\gamma_1} |1\rangle\langle 2|$ (decay from 2 to 1), $L_2 = \sqrt{\gamma_2} |2\rangle\langle 3|$ (decay from 3 to 2), with $\gamma_1,\gamma_2>0$.
The generator set is $G=\{\Lcal_1, \Lcal_2\}$. The dissipative rays $\R_{\geq0}\Lcal_1$ and $\R_{\geq0}\Lcal_2$ can be switched and scaled independently. We choose the representatives $\gamma_1=\gamma_2=1$ for the algebraic calculation; positive rescaling changes neither the cone nor the Lie closure. An inseparable autonomous drift $\gamma_1\Lcal_1+\gamma_2\Lcal_2$ is a singleton specification and does not carry the three-dimensional algebra computed below.

\paragraph{Analysis of the DLWP.}

\textbf{1. System Lie Algebra $\g(G)$:}
We compute the Lie algebra generated by $\Lcal_1$ and $\Lcal_2$. This involves calculating the commutators of the superoperators in the Liouville space representation.

The explicit calculation, detailed in Appendix \ref{app:cascaded_qutrit}, shows that $\g(G)$ is three-dimensional and solvable but not nilpotent: $\g(G) = \operatorname{span}\{\Lcal_1, \Lcal_2, \Lcal_3\}$ with $\Lcal_3 = [\Lcal_1, \Lcal_2]$, and the residual relations $[\Lcal_1, \Lcal_3] = -\Lcal_3$, $[\Lcal_2, \Lcal_3] = \Lcal_3$ make the lower central series stabilize at $\operatorname{span}\{\Lcal_3\}$. The derived series, in contrast, terminates at the second step. The algebra is triangular in a suitable Liouville-space basis, reflecting the unidirectional flow of population, but the damping terms contribute diagonal entries, and these diagonal entries are what obstruct nilpotency.

\textbf{2. Levi Decomposition:}
Since $\g(G)$ is solvable, the Levi decomposition is trivial: $\mf{s}(G)=\{0\}$, $\mf{r}(G)=\g(G)$. The entire algebra is the radical.

\textbf{3. Dynamical Lie Wedge $W(G)$:}
We analyze the structure of the wedge. Since $H=0$ and the generators are purely dissipative, the wedge is pointed. The edge $E(W(G))=\{0\}$.
Since the edge is trivial, the Lie wedge condition (invariance under the edge) is trivially satisfied for any convex cone. Therefore, the Dynamical Lie Wedge is simply the convex cone generated by $G$:
$W(G) = \cone(G) = \{\alpha \Lcal_1 + \beta \Lcal_2 \mid \alpha, \beta \ge 0 \}$. This forms a simple polyhedral cone (a sector in the plane spanned by $\Lcal_1, \Lcal_2$).

\textbf{4. Levi-Type Decomposition of $W(G)$:}
The decomposition is straightforward:
$C_{\text{quot}} = \{0\}$.
$W_{\rr} = W(G)$.
The Fiber Compatibility Condition is trivially satisfied as $\mf{s}=\{0\}$ and $E(W)=\{0\}$.

\textbf{Physical Interpretation:}
For the simultaneous benchmark generator $\gamma_1\Lcal_1+\gamma_2\Lcal_2$ with $\gamma_1,\gamma_2>0$, the jump operators transfer population sequentially along $|3\rangle\to|2\rangle\to|1\rangle$ and select $|1\rangle\langle1|$ as the asymptotic steady state. The lower-triangular representation mirrors that hierarchy. These are properties of this selected benchmark; solvability alone would not imply population transfer, a particular steady state, or the absence of mixing.

\subsubsection{Example 2: Mixed Reductive Dynamics}
\label{sec:example_reductive}

We consider a mixed reductive system with full unitary control and isotropic depolarizing noise. The example separates the semisimple Hamiltonian control algebra from a central dissipative ray.

\paragraph{System Specification: Qubit with Full Control and Depolarizing Noise.}

Consider a qubit system ($d=2$) subject to full unitary control and isotropic depolarizing noise. This is a standard model in quantum information processing.

The generator set is $G = \{\pm\Lcal_{H_x}, \pm\Lcal_{H_y}, \pm\Lcal_{H_z}, \Lcal_D\}$. The Hamiltonian directions are bidirectional controls. The depolarizing generator represents an independently switchable one-sided ray; changing $\gamma>0$ only rescales that ray, and we use $\gamma=1$.
Hamiltonians: $H_i = \sigma_i$ (Pauli matrices). $\Lcal_{H_i}(\rho) = -i[\sigma_i, \rho]$.
Depolarizing channel generator: $\Lcal_D(\rho) = \gamma \sum_{i=x,y,z} (\sigma_i \rho \sigma_i - \rho)$.

\paragraph{Analysis of the DLWP.}

\textbf{1. System Lie Algebra $\g(G)$:}
We compute the Lie closure of $G$.
The Hamiltonian generators $\{\Lcal_{H_i}\}$ generate the Lie algebra isomorphic to $\mf{su}(2)$. The commutation relations are $[\Lcal_{H_i}, \Lcal_{H_j}] = 2\epsilon_{ijk} \Lcal_{H_k}$.

We analyze the commutators involving the depolarizing generator $\Lcal_D$. The depolarizing channel is unitarily covariant: $U \Lcal_D(\rho) U^\dagger = \Lcal_D(U \rho U^\dagger)$ for any unitary $U$. In the superoperator representation, this means the adjoint action of the Hamiltonian generators on $\Lcal_D$ is trivial:
\begin{equation}
\ad(\Lcal_{H_i})(\Lcal_D) = [\Lcal_{H_i}, \Lcal_D] = 0.
\end{equation}
This property places $\Lcal_D$ in the center of $\g(G)$: the noise is independent of the control operations.

The System Lie Algebra is the direct sum:
$\g(G) = \text{Lie}(\{\Lcal_{H_i}\}) \oplus \text{span}(\Lcal_D) \cong \mf{su}(2) \oplus \R$.

\textbf{2. Levi Decomposition:}
The algebra $\g(G)$ is reductive.
Semisimple Levi factor: $\mf{s} = \mf{su}(2)$.
Solvable radical: $\mf{r} = \R$ (the center, spanned by $\Lcal_D$).
The coupling $\rho$ is trivial ($[\mf{s}, \mf{r}]=0$).

\textbf{3. Dynamical Lie Wedge $W(G)$:}
The signed Hamiltonian controls put all of $\mf{su}(2)$ in the edge. The depolarizing generator is central, so the following cone is invariant under that edge and is a closed Lie wedge containing $G$:
$W(G) = \{ \Lcal_H + \alpha \Lcal_D \mid \Lcal_H \in \mf{su}(2), \alpha \ge 0 \}$.
Minimality gives inclusion in this cone, while the signed controls and the ray $\R_{\geq0}\Lcal_D$ give the reverse inclusion. Hence the equality is exact and $E(W(G))=\mf{su}(2)$.
Geometrically, this is a cone whose base is the space $\mf{su}(2)$ and whose axis is the ray generated by $\Lcal_D$.

\textbf{4. Levi-Type Decomposition of $W(G)$:}
We apply Theorem \ref{thm:levi_type_main}.
\begin{itemize}
    \item $W_{\rr} = W(G) \cap \mf{r} = \{ \alpha \Lcal_D \mid \alpha \ge 0 \}$. A ray in $\mf{r}$.
    \item $C_{\text{quot}} = \pi_{\s}(W(G)) = \mf{su}(2)$. The quotient cone is the entire semisimple algebra.
    \item Adjoint Coupling Constraints (FCC): We analyze the fibers. $W(S) = W_{\rr}$ for all $S \in C_{\text{quot}}$. The inclusion $C_{\text{quot}} \subseteq E(W(G))$ implies this constancy by Proposition~\ref{prop:fiber_collapse}; centrality explains the trivial action on the radical.
    The Fiber Compatibility Condition (Eq. \eqref{eq:FCC}) requires:
    $e^{\ad(X_0)}(S+W_{\rr}) \subseteq e^{\ad(S_0)}(S) + W_{\rr}$.
    Since the coupling is trivial ($\rho=0$), the action of $X_0=S_0$ (as $R_0=0$ for the edge elements in this decomposition) on $\mf{r}$ is trivial. $e^{\ad(S_0)}(W_{\rr}) = W_{\rr}$.
    The FCC simplifies to $e^{\ad(S_0)}(S) + W_{\rr} \subseteq e^{\ad(S_0)}(S) + W_{\rr}$, which is trivially satisfied. The absence of coupling means the FCC imposes no constraints.
\end{itemize}

\textbf{Physical Interpretation:}
The semisimple factor is the Hamiltonian control algebra. Convergence to the maximally mixed state and the dissipative gap come from the selected depolarizing generator $\Lcal_D$ with $\gamma>0$. As $\gamma\downarrow0$ through positive values, the factor $\mf{su}(2)$ remains unchanged while the contraction rate and gap tend to zero. Since $\Lcal_D$ is central, the controls rotate independently of the isotropic contraction.

\section{Concrete Illustration: Mixed Dynamics in a Non-Reductive System}
\label{sec:example_non_reductive}

With the decomposition theorem in hand, we put it to work on a concrete system whose System Lie Algebra is non-reductive, one in which control operations modulate the noise structure. The example exhibits the Fiber Compatibility Condition as an edge-invariance constraint. Whether a separately diagnosed expanding or shearing action produces dynamical or spectral instability is a different question; Section~\ref{sec:neebs_theory} makes only the stated invariant-cone comparison under its additional hypotheses.

\paragraph*{System Specification: Coupled Qubit System with Modulated Dissipation}

Consider a system of two coupled qubits, $\Hcal = \Hcal_A \otimes \Hcal_B$. We design the dynamics such that the evolution of qubit $A$ (the control qubit) modulates the dissipation experienced by qubit $B$ (the target qubit).

\textbf{Control on Qubit A (Semisimple Subalgebra):}
We assume full unitary control on qubit A. The control Hamiltonians are $H_{A,i} = \frac{1}{2}\sigma_{A,i} \otimes I_B$ for $i=x,y,z$ (using the standard normalization for $\mf{su}(2)$ generators). The corresponding generators are $\Lcal_{H_{A,i}}$.

\textbf{Modulated Dissipation on Qubit B (Dissipative Inputs):}
The dissipation on B depends on the state of A. We consider amplitude damping on B modulated by the projectors onto the energy basis of A ($P_{A,0} = |0\rangle\langle 0|_A, P_{A,1} = |1\rangle\langle 1|_A$). We use $|0\rangle,|1\rangle$ as the eigenstates of $\sigma_z$ and set $\sigma_B^-=|1\rangle_B\langle0|$.

Put $J_0=P_{A,0}\otimes\sigma_B^-$ and $J_1=P_{A,1}\otimes\sigma_B^-$. The physical jumps are $L_1=\sqrt{\gamma_0}J_0$ and $L_2=\sqrt{\gamma_1}J_1$, with $\gamma_0,\gamma_1>0$, and we write $\Dcal_i=\Dcal_{J_i}$, so that $\Lcal_i=\Dcal_{L_i}=\gamma_i\Dcal_i$.
The base set of generators is 
\[
G = \{\pm\Lcal_{H_{A,x}}, \pm\Lcal_{H_{A,y}}, \pm\Lcal_{H_{A,z}}, \Lcal_1, \Lcal_2\}.
\]
The two dissipative rays are independently switchable. Positive rate changes merely rescale those rays and therefore leave their cone, Lie closure, and minimal closed Lie wedge unchanged; no inequality between the rates is required. The exact certificate below uses the representative $(\gamma_0,\gamma_1)=(1,2)$. A fixed simultaneous drift $\gamma_0\Dcal_0+\gamma_1\Dcal_1$ would instead be a one-ray specification.

\subsection{Analysis of the DLWP and the Non-Reductive Structure}

\textbf{1. System Lie Algebra $\g(G)$:}
We compute the Lie algebra generated by $G$. What matters is the interaction between the control Hamiltonians (acting on A) and the dissipative generators (coupled A B).

\textbf{Subalgebras:}
The control generators $\{\Lcal_{H_{A,i}}\}$ form the algebra $\mf{g}_A \cong \mf{su}(2)_A$.

\textbf{Commutation Relations, Closed-Wedge Construction, and the Coupling $\rho$:}
We determine the structure of the System Lie Algebra $\g(G)$ by computing the commutators of the base generators. The algebraic closure must be distinguished from the physical generators: while the commutators $[\Lcal, \Lcal']$ reveal the directions spanning the Lie algebra, the resulting superoperators often violate the condition of complete positivity (exhibiting indefinite Kossakowski matrices). The physical Dynamical Lie Wedge $W(G)$ is obtained by the iterative closed conical hull of the adjoint orbits generated by its edge.

Covariance gives the control action directly. For $H_{A,x}=\frac12\sigma_{A,x}\otimes I_B$, $U_t=e^{-itH_{A,x}}$, and any jump $L$,
\begin{equation}
e^{t\ad(\Lcal_{H_{A,x}})}\Dcal_L=\Dcal_{U_tLU_t^\dagger}.
\end{equation}
For $L=J_0$, put $M=-i[H_{A,x},J_0]=-\frac12\sigma_{A,y}\otimes\sigma_B^-$. Differentiating the covariance identity gives
\begin{equation}
[\Lcal_{H_{A,x}},\Dcal_{J_0}]=\mathcal C(J_0,M),\qquad
\mathcal C(L,M)(\rho)=M\rho L^\dagger+L\rho M^\dagger-\frac12\{M^\dagger L+L^\dagger M,\rho\}.
\end{equation}
This is a bilinear cross superoperator, not the quadratic dissipator $\Dcal_M$. At $t=\pi/2$ the finite orbit jump is
$P_y^-\otimes\sigma_B^-$, where $P_y^-=(I_A-\sigma_{A,y})/2$; it is not $\sigma_{A,y}\otimes\sigma_B^-$ alone.

The wedge $W(G)$ must contain the closed conic hull of these physical adjoint orbits. The nonzero cross direction records the infinitesimal control--noise coupling, while the finite orbit remains inside the GKLS cone.

By iteratively computing the commutators, we generate the full System Lie Algebra $\g(G)$. The algebra involves the control operators acting on A and the coupled operators (both dissipative and Hamiltonian) acting on the combined system.

\textbf{2. Levi Decomposition:}
We identify the structure of $\g(G)$.
The semisimple Levi factor contains the control algebra $\mf g_A\cong\su(2)_A$ and has dimension nine; the solvable radical has dimension thirteen. The following exact certificate proves both dimensions.

\paragraph{Exact certificate for the Levi data.}
Set $X=\Lcal_{H_{A,x}}$, $Y=\Lcal_{H_{A,y}}$, $Z=\Lcal_{H_{A,z}}$, $D=\Dcal_0$, and $E=2\Dcal_1$. The following recursive words form an ordered basis:
\begin{equation}
\begin{aligned}
b_1&=X,&b_2&=Y,&b_3&=Z,&b_4&=D,&b_5&=E,\\
b_6&=[X,D],&b_7&=[X,E],&b_8&=[X,b_6],&b_9&=[Y,D],&b_{10}&=[Y,E],\\
b_{11}&=[Y,b_6],&b_{12}&=[Y,b_9],&b_{13}&=[D,b_6],&b_{14}&=[D,b_7],\\
b_{15}&=[D,b_8],&b_{16}&=[D,b_9],&b_{17}&=[D,b_{10}],&b_{18}&=[D,b_{13}],\\
b_{19}&=[D,b_{15}],&b_{20}&=[D,b_{16}],&b_{21}&=[b_6,b_9],&b_{22}&=[b_6,b_{16}].
\end{aligned}
\end{equation}
Represent each superoperator in the column-stacked computational operator basis. For its resulting $16\times16$ matrix $S$, let $\operatorname{rvec}(S)\in\mathbb Q^{512}$ list the real parts of the matrix entries row by row, followed by the imaginary parts, and let $B=(\operatorname{rvec}(b_1)\ \cdots\ \operatorname{rvec}(b_{22}))$. The rows
\begin{equation}
\begin{gathered}
1,3,20,33,35,81,83,91,113,115,121,123,259,274,276,\\
289,291,339,369,371,377,379
\end{gathered}
\end{equation}
give a $22\times22$ minor of determinant $9/2^{31}$. All $231$ unordered brackets $[b_i,b_j]$ resolve exactly in this basis. Deleting $b_{22}$ leaves rank $21$, and its defining bracket $[b_6,b_{16}]$ escapes the truncated span, independently confirming the necessity of the final basis element. Thus $\dim\g(G)=22$ exactly.

Write $[b_i,b_j]=\sum_k c_{ij}^{k}b_k$, let $K_{pq}=\Tr(\ad b_p\,\ad b_q)$, and form the Cartan constraint matrix $C$ with rows $(c_{ij})^TK$ for $i<j$. The rows indexed by
\[
(1,2),(1,3),(1,4),(1,6),(1,16),(1,21),(2,3),(2,4),(2,21)
\]
and columns $1,2,3,4,6,9,13,16,21$ have determinant $5^9/2^9$. Exact row reduction gives $\operatorname{rank}C=9$ and a thirteen-dimensional nullspace; the nullspace is an ideal, with derived dimensions $13\to12\to0$ and lower-central dimensions $13\to12\to12$. Cartan's criterion therefore identifies it with $\mf r$. Modulo this radical, the classes of
\[
b_1,b_2,b_3,b_4,b_6,b_9,b_{13},b_{16},b_{21}
\]
form a quotient basis whose Killing determinant is $1/8$.

Finally let $A_x,A_y,A_z$ be the quotient actions of $b_1,b_2,b_3$, take $z_1=\bar b_3$, $z_2=\bar b_4$, $z_3=\bar b_{21}$, and set
\[
T=(A_yz_1,-A_xz_1,z_1\mid A_yz_2,-A_xz_2,z_2\mid A_yz_3,-A_xz_3,z_3).
\]
Use the standard adjoint matrices
\[
J_x=\begin{pmatrix}0&0&0\\0&0&-1\\0&1&0\end{pmatrix},\qquad
J_y=\begin{pmatrix}0&0&1\\0&0&0\\-1&0&0\end{pmatrix},\qquad
J_z=\begin{pmatrix}0&-1&0\\1&0&0\\0&0&0\end{pmatrix}.
\]
In the quotient basis above, $\det T=1$ and
\begin{equation}
T^{-1}A_iT=\operatorname{diag}(J_i,J_i,J_i),\qquad i=x,y,z.
\end{equation}
Hence $\g/\mf r\cong\operatorname{ad}^{\oplus3}$ as an $\su(2)_A$-module and has no fixed vector. Three independent checks distinguish the required hypotheses: deleting $b_{22}$ breaks closure, repeating a selected Cartan row makes the rank minor vanish, and adjoining a trivial module summand creates a fixed vector.

At the jump-operator level, the adjoint orbits lie in
\begin{equation}
L_{mod} = \text{span}_{\C} \{ A \otimes \sigma_{B}^- \mid A \in \Bcal(\Hcal_A) \}.
\end{equation}
This operator-space span is useful for locating the orbits, but $L\mapsto\Dcal_L$ is quadratic, so it supplies neither a superoperator basis nor a proof of the dimensions above. In particular, radical membership is decided by the exact Cartan certificate, not by whether a jump lies in $L_{mod}$.

\textbf{3. Non-Reductivity and the Action $\rho$:}
The action $\rho$ of $\mf{s}$ on $\mf{r}$ is non-trivial.
\begin{equation}
\rho: \mf{s} \to \text{Der}(\mf{r}), \quad \rho \neq 0.
\end{equation}
The representation $\rho$ is the adjoint action of the Levi factor on the radical. More explicitly, $R_z=\Dcal_{\sigma_{A,z}\otimes\sigma_B^-}$ has coordinate column $[R_z]_B$ in $(b_1,\ldots,b_{22})$, with $C[R_z]_B=0$, whereas $[X,R_z]\neq0$. Hence $[\su(2)_A,\mf r]\neq0$ and $\rho$ is non-trivial. This statement concerns the radical ideal; it is not inferred from a linear span of jump operators.

The System Lie Algebra $\g(G) = \mf{s} \ltimes_{\rho} \mf{r}$ is therefore non-reductive, and the structure of $\rho$ dictates how the control operations modulate the effective noise experienced by qubit B.

\subsection{Analysis of the Dynamical Lie Wedge and the FCC}

We now analyze the structure of the Dynamical Lie Wedge $W(G)$ using the Levi-Type Decomposition Theorem. This requires careful consideration of the constraints imposed by the FCC in this non-reductive setting.

\textbf{1. The Edge $E(W(G))$:}
The edge consists of the reversible (Hamiltonian) generators contained in the physical local wedge. For the specified generating set the closed-wedge construction stabilizes at once: positive semidefiniteness is additive along the dissipative directions, so no conical combination of orbit generators can cancel the dissipative parts and add new reversible elements, and the edge is exactly the primary control algebra:
\begin{equation}
E(W(G)) = \mf{g}_A \cong \mf{su}(2)_A.
\end{equation}

\textbf{2. The Quotient Cone $C_{\text{quot}}$:}
The quotient cone $C_{\text{quot}}$ is the projection of $W(G)$ onto the Levi factor $\mf{s}$. The exact certificate gives $\dim\mf{s}=9$. The edge $E(W(G))\cong\mf{su}(2)_A$ accounts for three of those dimensions; the remaining directions of $C_{\text{quot}}$ are fed by the Levi-factor components of the dissipative generators, so the quotient cone extends strictly beyond the edge directions.

\textbf{3. The Radical Wedge $W_{\rr}$ and the zero-base FCC slice:}
The radical wedge is $W_{\rr} = W(G) \cap \mf{r}$. The calculation below checks the necessary FCC slice over the base point $S=0$.

Let $X_0 \in \mf{g}_A$ belong to the control edge. The fiber above $S=0$ is $W_{\rr}$, and the FCC at that base point implies
\begin{equation}
e^{\ad(X_0)}(W_{\rr}) \subseteq W_{\rr} \quad \forall X_0 \in \mf{g}_A.
\end{equation}

Thus the cone $W_{\rr}$ must be invariant under the control action of $SU(2)_A$. This is a necessary zero-base consequence of the full FCC, which also quantifies over every nonzero base point in $C_{\mathrm{quot}}$.

The initial dissipative generators $\Lcal_1, \Lcal_2$ correspond to Lindblad operators modulated by $P_{A,0}, P_{A,1}$ (related to $\sigma_{A,z}$). The action of the control group $SU(2)_A$ transforms these modulation operators by conjugation.

We examine the orbit of the Lindblad operators under the action of $U_A \in SU(2)_A$. The action is $(U_A \otimes I_B) L_i (U_A^\dagger \otimes I_B)$.
Consider $L_1$. The action transforms $P_{A,0}$ to $U_A P_{A,0} U_A^\dagger$. This is the projector onto the state $U_A|0\rangle$. As $U_A$ varies over $SU(2)_A$, this orbit covers all rank 1 projectors on $\Hcal_A$.

Because the two rays are independently switchable, the closed-wedge construction contains the union of their individual adjoint orbits,
\[
\{\operatorname{Ad}_{U_A}\Lcal_1:U_A\in SU(2)_A\}\ \cup\
\{\operatorname{Ad}_{V_A}\Lcal_2:V_A\in SU(2)_A\},
\]
and its conic combinations may choose $U_A$ and $V_A$ independently. A linked pair with one common rotation would describe a different control constraint. Merging the two jump operators into a single one would also change the generator, since the dissipator is quadratic in $L$ and $\Dcal_{L_1+L_2}\neq\Dcal_{L_1}+\Dcal_{L_2}$.

Edge invariance of the full wedge forces $W(G)$ to contain the dissipative generators corresponding to this entire orbit; the construction in Section~\ref{sec:dlwp_framework} closes the cone under exactly these adjoint orbits. The orbit generators themselves carry components along the Levi factor, so they enter the decomposition through fibers over nonzero base points rather than through $W_{\rr}$.

At zero base, averaging over $SU(2)_A$ projects onto the control-fixed isotypic component. The two normalized dissipative rays have the same Haar average, and the physical rates only rescale it. The average therefore generates a single radical ray:
\begin{equation}
\bar{\Lcal}_i := \int_{SU(2)_A} \operatorname{Ad}_{U_A} \Lcal_i \, dU_A \in W_{\rr}, \qquad e^{\ad(X_0)} \bar{\Lcal}_i = \bar{\Lcal}_i \qquad (i = 1, 2).
\end{equation}
What puts the averages inside the radical is a property of the quotient. Since $\mf{r}$ is an ideal, the control action descends to $\g(G)/\mf{r}$ and the quotient map intertwines the averaging operators. The exact matrix $T$ above identifies this nine-dimensional quotient with three adjoint copies and proves that it has no invariant vector. Every element fixed by the control therefore lies in $\mf r$. Invariance of $\mf{r}$ by itself would not settle the question, because the elements fixed by the control form the centralizer of $\mf{g}_A$, and a centralizer can meet a Levi factor. The exact $22/13/9$ certificate and module decomposition establish this $S=0$ slice; they do not reduce the full FCC to that slice.

The full FCC holds here because the closed-wedge construction makes $W(G)$ invariant under its entire edge. The zero-base calculation is a consistency check on that construction, not an independent proof for all fibers.

\paragraph*{Edge Action and the Scope of the Convex-Type Comparison}

We analyze the edge action in this example directly. The calculation does not identify a bracket-induced symplectic module of the kind required by Neeb's theory, so no convex-type conclusion is drawn from it.

\textbf{1. Operator-Space Action.}
The control subalgebra $\mf g_A\cong\su(2)_A\subset\mf s$ acts nontrivially on the radical $\mf r$. The radical contains directions generated by the modulated dissipation together with induced Hamiltonian terms. This paragraph studies the restriction of the Levi action to $\mf g_A$; it makes no identification of a symplectic nilradical subquotient.

Focus on the action of the control algebra $\mf{su}(2)_A$ on the space of modulation operators. The space of modulation operators acting on $\Hcal_A$ is $\Bcal_H(\Hcal_A)$. The action is conjugation.

A convenient invariant subspace for this direct calculation is the real vector space of traceless Hermitian operators on $\Hcal_A$, $\{X\in\Bcal_H(\Hcal_A)\mid\Tr X=0\}\cong\R^3$, identified with the Bloch-vector space. The action of $\mf{su}(2)_A$ on this subspace is the standard rotation ($\mf{so}(3)$ action).

A convex-type analysis would first have to characterize the nilradical $\mf{n}$ of $\g(G)$ and the bracket-induced symplectic structures on its subquotients. We do not use such an analysis here. The control algebra $\mf{su}(2)_A$ acts on the modulation operators, and its coupling $\rho$ must preserve $W(G)$ through the edge-invariance condition.

\textbf{2. The Edge-Invariance Slice.}
The action of the control algebra $\mf{su}(2)_A$ on the radical $\mf{r}$ rotates the modulation basis. On the subspaces used here, the operators $\ad(\Lcal_{H_{A,i}})$ are semisimple with purely imaginary eigenvalues and hence generate bounded one-parameter groups.

In this example the rotational action makes the zero-base constraint mild. The FCC requires invariance of $W_{\rr}$ under the control action as one necessary slice, and the control-averaged directions extracted from the closed-wedge construction realize a compatible cone explicitly.

Here the control action is elliptic, and invariance of $W_{\rr}$ verifies the $S=0$ slice of the FCC. The full condition still concerns every fiber and follows from the Lie-wedge construction. Section~\ref{sec:neebs_theory} next separates the invariant-cone comparison, the FCC, complete positivity, and the spectral-stability test.

\paragraph*{Physical Interpretation}

This example has mixed non-reductive type: a semisimple control subalgebra on qubit A acts through $\rho$ on the solvable radical generated in the closure by structured amplitude damping on qubit B. The Levi-Type Decomposition separates these components without assigning either one a spectral or many-body character.

Because the construction closes $W(G)$ under the adjoint orbit of its edge, implementing state-dependent dissipation in one basis together with fast control makes that dissipation available in every basis in the control orbit. The rotated dissipators are not elements of $W_{\rr}$, their cross terms lying along the Levi factor; they enter the wedge through fibers over non-zero base points. The next section keeps this geometric constraint separate from complete positivity and from the spectrum of a selected drift.

\section{The Geometry of Non-Reductive Dynamics: Edge Invariance, the FCC, and Neeb's Theory}
\label{sec:neebs_theory}

The decomposition theorem makes the Fiber Compatibility Condition the edge-invariance constraint on a Lie wedge in a non-reductive algebra $\g = \mf{s} \ltimes \mf{r}$. Section~\ref{sec:example_non_reductive} verified one rotational slice. Here we compare the FCC with Neeb's invariant-cone theory where the required symplectic module exists, and then use the Gaussian example to separate complete positivity from the FCC and from Hurwitz stability. Figure~\ref{fig:fcc_visualization} gives a schematic of the fiber geometry.

In a non-reductive algebra, $\gLK$ among them (Corollary~\ref{cor:glk_nonreductive}), the coupling $\rho$ can obstruct invariance under the full algebra. This obstruction has no automatic spectral consequence. The wedge $W(G)$ is invariant under its edge $E(W)$, and the FCC is precisely the expression of that restricted invariance in fiber coordinates.

\subsection{Analysis of the Convex Type Condition for the Ambient Algebra $\gLK$}
\label{sec:convex_type_glk}

We next ask whether the ambient algebra $\gLK \cong \aff(\Htr)$ can carry a pointed generating cone invariant under its full adjoint action. This is an algebraic admissibility question about the ambient space, not a spectral-stability test for a physical generator.

\paragraph*{Algebraic Structure and Identification of Modules}

We analyze the structure of $\gLK = \aff(\Htr) = \gl(\Htr) \ltimes \Htr$.
The decomposition of the linear part is $\gl(\Htr) = \sll(\Htr) \oplus \R \cdot \Id_{\Htr}$.
The radical is $\rad(\gLK) = \R \cdot \Id_{\Htr} \ltimes \Htr$.
The nilradical is $\mf{n} = \Htr$ (the translational ideal). The nilradical is abelian, $[\mf{n}, \mf{n}]=0$.

In Neeb's theory the symplectic modules carry forms induced by the Lie bracket: on the relevant subquotient of the lower central series of $\mf{n}$, the form is $\omega_{\xi}(x, y) = \xi([x, y])$ for a suitable functional $\xi$. Here the nilradical $\mf{n} = \Htr$ is abelian, so every bracket-induced form vanishes identically and the construction produces no nonzero symplectic module. The convex type condition is therefore vacuous for $\gLK$: it neither obstructs nor licenses an invariant cone. The operative criterion at the level of the ambient algebra is instead the structure theory of admissible Lie algebras, which we apply next.

\paragraph*{The Constraints on the Affine Algebra and Its Implications for $\gLK$}

The analysis of invariant cones in the affine algebra $\aff(V)$ reveals that the constraints arise from the coupling between the linear and translational parts. The existence of pointed generating invariant cones in $\aff(V)$ is highly restricted. The following is a direct consequence of the structure theory of admissible Lie algebras \cite{NeebOeh2021} (see also \cite{Neeb2000}).

\begin{proposition}[Invariant Cones in the Affine Algebra]
The affine algebra $\aff(V) = \gl(V) \ltimes V$ admits a pointed generating invariant cone if and only if $V=\{0\}$.
\end{proposition}

\begin{proof}
If $V=\{0\}$ the statement is vacuous. Suppose $V \neq \{0\}$. A finite-dimensional real Lie algebra admitting a pointed generating invariant (closed convex) cone is \emph{admissible}, and by the structure theory of admissible Lie algebras every abelian ideal of an admissible Lie algebra is central \cite[Lemma~27(a)]{NeebOeh2021}. In $\aff(V)$ the translational ideal $V$ is abelian, but it is not central: $\gl(V)$ acts on $V$ by the (faithful) standard representation, so $[\gl(V), V] = V \neq \{0\}$. Hence $\aff(V)$ is not admissible and admits no pointed generating invariant cone.
\end{proof}

Since $\gLK \cong \aff(\Htr)$ and $\Htr \neq \{0\}$ for $d \ge 2$, the ambient algebra $\gLK$ admits no pointed generating invariant cone. Its non-reductive coupling prevents invariance under the full adjoint action.

The ambient conclusion is nonadmissibility, not spectral instability. No pointed generating cone is invariant under all of $\gLK$, so physical control systems are described by Lie wedges. For a specified system, the FCC enforces invariance under reversible edge directions. Complete positivity limits which generators can lie in the wedge, while decay or growth of a selected generator must be checked from its spectrum.

\subsection{The FCC as an Edge-Invariance Constraint: A Structured Correspondence}
\label{sec:connection_neeb_fcc_glk}

Neeb's convex-type results concern cones invariant under the full algebra, while a Lie wedge is invariant under its edge. The infinitesimal FCC is a subtangency condition on the graph of the fiber map; Neeb's condition is stated through a moment-map cone and the center action. Their comparison presupposes a bracket-induced symplectic module. It is vacuous for $\gLK$ itself (Section~\ref{sec:convex_type_glk}) and applies when the nilradical carries such a module, as in an oscillator algebra with a genuine Heisenberg nilradical. The solvable Gaussian system algebra of Section~\ref{sec:stabilization_hyperbolic} is outside that domain.

\begin{remark}[FCC and convex type: a correspondence]
\label{rem:fcc_convex_type}
Let $(\g, W)$ be a DLWP with $\g = \mf{s} \ltimes \mf{r}$ whose nilradical $\mf{n}$ carries a bracket-induced symplectic module $(V, \omega)$. Failure of convex type says that at least one condition in Definition~II.1 fails; it does not identify a hyperbolic or parabolic mode. If an edge generator is shown independently to produce expansion or shear on $V$, the infinitesimal FCC of Proposition~\ref{prop:infinitesimal_fcc} requires that flow to remain tangent to the local graph of the fibers. The comparison is therefore between two invariance conditions, global algebra invariance for Neeb's cones and edge invariance for Lie wedges. It is not an equivalence and carries no spectral implication by itself.
\end{remark}

By Proposition~\ref{prop:infinitesimal_fcc}, edge-invariance of the closed wedge is the subtangency condition, which at $Y = S + R$ requires the algebraic shear $V_R^{alg} = [R_0, R] + \rho(S_0)R - \rho(S)R_0$ to lie in $DW(S|R)(V_S^{alg})$ with $V_S^{alg} = [S_0, S]$, a membership that already forces $V_S^{alg} \in T_S(C_{\text{quot}})$.

Projected onto the symplectic module $V$, the action of an edge element $X_0$ is Hamiltonian, and $V_R^{alg}$ records its infinitesimal contribution in the chosen fiber coordinates. If a separate spectral calculation finds an expanding or shearing component, subtangency requires the resulting vector to stay in the contingent derivative of the graph. The fibers may therefore have to rotate or change shape along the edge orbit. This conclusion comes from the FCC and the independently established spectrum; failure of convex type alone supplies neither premise.

\subsection{A Hyperbolic Bosonic Example: Complete Positivity and Spectral Stability}
\label{sec:stabilization_hyperbolic}

Within the finite-dimensional Gaussian moment representation, a squeezed, damped bosonic mode separates the infinitesimal CP criterion from asymptotic stability. The Gaussian CP matrix is positive semidefinite for every damping rate $\gamma\geq0$, while its first-moment drift is Hurwitz only above a threshold set by the squeezing strength. The Gaussian system algebra is solvable, so the FCC plays no role in that threshold.

\paragraph*{A Squeezed Bosonic Mode}

Consider a single mode with $[a,a^\dagger]=1$ and quadratures $Q=(a+a^\dagger)/\sqrt{2}$ and $P=(a-a^\dagger)/(i\sqrt{2})$. The squeezing Hamiltonian and amplitude-damping generator are
\[
H_S=\frac{i}{2}\bigl(r a^2-r(a^\dagger)^2\bigr),
\qquad
\Lcal_D(\rho)=\gamma\left(a\rho a^\dagger-\frac{1}{2}\{a^\dagger a,\rho\}\right),
\]
where $r\in\R$ and $\gamma\geq 0$. In the ordering $(Q,P)$, the corresponding Gaussian moment pair has
\[
A=\begin{pmatrix}-\gamma/2-r&0\\0&-\gamma/2+r\end{pmatrix},
\qquad
D=\frac{\gamma}{2}I_2.
\]
The Hamiltonian part is hyperbolic when $r\neq0$, and its quadratic Hamiltonian on phase space is indefinite, as in Example~\ref{ex:heisenberg_convex_type}.

\begin{remark}[Finite-dimensional Gaussian scope]
All algebraic claims in this example concern the finite-dimensional Gaussian moment representation $(A,D)$. These moment-level statements do not extend the paper's finite-dimensional Levi theorem to the full oscillator operator algebra and do not establish domains, conservativity, or the existence of a strongly continuous quantum dynamical semigroup for the unbounded oscillator generator. Such an extension would require a separate infinite-dimensional semigroup analysis.
\end{remark}

\paragraph*{Solvability and the Spectral Threshold}

In this specification, $\Lcal_{H_S}$ and $\Lcal_D$ represent independently actuable one-sided rays. For $r\neq0$ and $\gamma>0$, their Gaussian Lie algebra has dimension four, with derived-series dimensions $4\to2\to0$. The algebra is therefore solvable. Hence $\mf{s}=0$, $C_{\mathrm{quot}}=\{0\}$, and $W=W_{\rr}$. This specification has solvable type; the FCC has a singleton base and is automatic. The supplied squeezing direction is one-sided, with $\Lcal_{H_S}\notin E(W)$.

The symmetric extension also has constant fibers. Let $H_1,H_2$ be quadratic Hamiltonians with Gaussian drift matrices
\[
S_1=\begin{pmatrix}-1&0\\0&1\end{pmatrix},
\qquad
S_2=\begin{pmatrix}0&1\\1&0\end{pmatrix},
\]
and set $G_{\mathrm{ext}}=\{\pm\Lcal_{H_1},\pm\Lcal_{H_2},\Lcal_D\}$. The Gaussian algebra then has dimension seven and a Levi factor isomorphic to $\mf{sl}(2,\R)$. The bidirectional Hamiltonian controls put its quotient directions in the edge, so Proposition~\ref{prop:fiber_collapse} makes the fiber map constant. A nonconstant-fiber instance in which the FCC has content appears in Example~\ref{ex:fiber_mobility_glk2}.

For the moment pair above, the infinitesimal Gaussian CP matrix has eigenvalues $0$ and $\gamma$. The CP condition is therefore $\gamma\geq0$, independently of $|r|$. The drift eigenvalues are
\[
\lambda_\pm(A)=-\frac{\gamma}{2}\pm |r|,
\]
so $A$ is Hurwitz exactly when $\gamma>2|r|$.

Independent positive rescaling of the two generating rays leaves their DLWP unchanged, while the specified sum can cross the Hurwitz threshold as the two coefficients vary. The threshold is therefore a spectral condition on a fixed combined drift. It lies outside the wedge classification and the FCC.

\subsection{Counterexamples Illustrating the Necessity of the FCC}
\label{sec:counterexamples}

A single worked counterexample isolates what fails without the Fiber Compatibility Condition. A pointed radical wedge that sits in a plane the edge rotates cannot survive that rotation, and the FCC is exactly the obstruction. We give the mechanism abstractly in the Euclidean algebra $\mf{e}(2)$, then restate it physically for a dephasing qubit.

\label{sec:fcc_counterexample}
The Euclidean algebra $\mf{e}(2) = \mf{so}(2) \ltimes \R^2$ carries the semidirect product structure $\mf{e}(2) = \mf{h} \ltimes \mf{r}_{ab}$, with $\mf{h}=\mf{so}(2)$ the rotations (generator $J$), $\mf{r}_{ab}=\R^2$ the translations (generators $P_x, P_y$), and $\rho$ the standard rotation action, so that $[J, P_x] = P_y$ and $[J, P_y] = -P_x$. The algebra is solvable, so its true Levi factor is trivial; the splitting still illustrates the FCC mechanism, which operates wherever a subalgebra $\mf{h}$ acts non-trivially on an ideal $\mf{r}_{ab}$. The constraints transfer verbatim, since the proofs of Proposition~\ref{lem:fiber_compatibility_main} and Proposition~\ref{lem:sufficiency} use only that $\mf{r}$ is an ideal admitting a complementary subalgebra $\mf{s}$, with no appeal to the semisimplicity of $\mf{s}$ or the maximality of $\mf{r}$.

Assemble a candidate wedge from these components. Demand the edge $E(W) = \mf{so}(2)$ and set the quotient part $C_{\text{quot}} = \mf{so}(2)$. Every $S \in \mf{so}(2)$ then lifts into the edge with zero translation component, so $0 \in W(S) \cap W(-S)$, and this forces the fibers to be constant. For $R \in W(S)$ we have $R = (S+R) + (-S) \in W \cap \mf{r}_{ab} = W_{\rr}$, while conversely $S + W_{\rr} \subseteq W$ because $0 \in W(S)$, so $W(S) = W_{\rr}$ for every $S$. The zero-component lift is what does the work; subadditivity over a base subspace alone would not force constant fibers, as the abelian algebra $\R \oplus \R$ shows, where the Lie wedge $W = \{(s,r) \mid r \ge |s|\}$ has $C_{\text{quot}} = \R$ a subspace yet non-constant fibers $W(s) = [\,|s|, \infty)$. Take $W_{\rr}$ to be a pointed cone, the positive quadrant
\begin{equation}
W_{\rr} = \{(x, y) \in \R^2 \mid x \ge 0, y \ge 0\} = \cone(P_x, P_y),
\end{equation}
which is trivially a Lie wedge in the abelian $\mf{r}_{ab}$, and form $W = \mf{so}(2) + W_{\rr}$.

Now test the FCC at $X_0 = J \in E(W)$. With constant fibers the condition \eqref{eq:FCC} reduces to invariance of the radical wedge under the edge,
\begin{equation}
e^{\ad(J)t}(W_{\rr}) \subseteq W_{\rr} \quad \forall t \in \R,
\end{equation}
where $e^{\ad(J)t}$ rotates the plane by angle $t$. The positive quadrant is not rotation invariant, since rotating $P_x$ by $\pi$ returns $-P_x \notin W_{\rr}$. The FCC fails, so $W$ is not a Lie wedge even though it was built from a valid $C_{\text{quot}}$ and a valid $W_{\rr}$. The only closed convex cones in $\R^2$ invariant under the full rotation group are $\{0\}$ and $\R^2$, so a pointed generating $W_{\rr}$ forces a trivial edge. Two independent reasons enforce this. An edge element with nonzero rotation component would impose rotation invariance of $W_{\rr}$ by the computation above, while one lying in the translation ideal would belong to $E(W) \cap \mf{r}_{ab} = E(W_{\rr}) = \{0\}$ by pointedness. The non-reductive coupling constrains the wedge geometry in a way the naive construction ignored.

\begin{remark}[The same obstruction for a dephasing qubit]
The mechanism has a direct qubit realization. Take bidirectional control by $H=\sigma_x$, so the declared candidate edge is $E(W)=\operatorname{span}(\Lcal_{H_x})$, and restrict the declared noise cone to the single ray generated by $z$-dephasing, $\Lcal_z(\rho)=\sigma_z\rho\sigma_z-\rho$. Because the $x$-control is reversible, the sequence $e^{t\Lcal_{H_x}}e^{s\Lcal_z}e^{-t\Lcal_{H_x}}=e^{s e^{t\ad(\Lcal_{H_x})}\Lcal_z}$, with $s\geq0$ and $t\in\R$, implements unitary conjugation of the dissipator. A quarter turn carries $\Lcal_z$ to $\Lcal_y$. The generator $\Lcal_y(\rho)=\sigma_y\rho\sigma_y-\rho$ is itself a rank-one GKLS generator and is therefore CP; the failure is that it lies outside the declared one-ray candidate cone. Edge invariance requires the conical closure of the full orbit,
\begin{equation}
C_{\mathrm{orb}} = \cone\{ \Lcal_n \mid n \in \text{span}(\hat{y}, \hat{z}) \},
\end{equation}
At fixed total dephasing rate, the convex hull of the orbit is the two-dimensional disk of real $2\times2$ positive semidefinite matrices on the $yz$ block with fixed trace. Allowing arbitrary nonnegative rates gives $C_{\mathrm{orb}}$, the three-dimensional positive-semidefinite cone on that block. This orbit-generated dissipative cone is contained in $W(G)$ and is distinct from the radical slice $W_{\rr}$.
\end{remark}

\section{Conclusion}
\label{sec:conclusion}

This work develops structural tools for generator-labelled Markovian control specifications. The DLWP couples the generated real Lie algebra to the minimal closed Lie wedge determined by the declared control rays. In finite dimension, the edge-orbit construction of Proposition~\ref{prop:construction_W(G)_saturation} reaches that wedge after finitely many stages. Its relation to the tangent wedge of the reachable semigroup remains a separate globality question.

The affine realization of $\gLK$ exposes the non-reductive ambient algebra in which the wedge sits. Relative to a chosen Levi splitting, Theorem~\ref{thm:levi_type_main} re-encodes the wedge by its radical slice, quotient cone, and fiber map. The FCC in Equation~\eqref{eq:FCC} is the fiber-coordinate form of edge invariance. Conversely, data satisfying conditions (D1)--(D6) reconstruct the unique wedge with those components. This structural normal form becomes a genuine reduction when the fibers are constant; in general it is an exact change of coordinates. For the control-generated wedges $W\subseteq\WGKLS$ considered here, the Lie-semialgebra condition occurs in the restricted situations stated in the text.

The dissipation functional yields a restriction specific to GKLS systems. Corollary~\ref{cor:hamiltonian_rigidity} shows that every wedge element in the commutator ideal is Hamiltonian, a semisimple system algebra synthesizes no dissipation, and the dissipation strength of every element is carried by its radical coordinate. Neeb's invariant-cone theory supplies a structured correspondence for algebras whose nilradical carries a bracket-induced symplectic module; the convex-type condition is vacuous for the ambient algebra $\gLK$ itself. Physical algebras with genuine Heisenberg nilradicals, including oscillator algebras, provide the natural setting for extending this comparison. The ambient algebra and the solvable Gaussian example here require a direct analysis outside that setting.

These structural outputs provide coordinates for later studies of control, chaos, or protected information; thermalization, scrambling, and protocol performance require additional dynamical input. The framework covers generator-labelled, piecewise-controlled Markovian dynamics. Each segment uses a time-homogeneous GKLS generator, while a protocol with $|G|>1$ is switched and need not be autonomous; the singleton $G=\{\Lcal\}$ recovers an autonomous semigroup. Time-local or time-dependent generators may admit cone or Lie-wedge descriptions on suitable intervals. Non-Markovian memory and multi-time correlations require additional structure.

\appendix

\appendix
\section{Appendix A: Supplementary Mathematical Details and Derivations}

\subsection{Detailed Structure of Non-Reductive Lie Algebras}
\label{app:nonreductive_structure}

This appendix collects details on the structure of non-reductive Lie algebras that the main text draws on, together with the standard examples.

\paragraph*{The Role of the Nilradical and the Structure of the Radical}

In a non-reductive algebra $\g = \mf{s} \ltimes \mf{r}$ the nilradical $\mf{n} = \nil(\g)$, the maximal nilpotent ideal, sits inside the radical, $\mf{n} \subseteq \mf{r}$. The radical need not itself be nilpotent, but its derived algebra satisfies $[\mf{r}, \mf{r}] \subseteq \mf{n}$, so $\mf{r}/\mf{n}$ is abelian and $\mf{r}$ is an extension of $\mf{n}$ by this abelian quotient.

The action of the entire algebra $\g$ on $\mf{n}$ is what the analysis of invariant cones turns on: the constraints on convexity are dictated by this action on the symplectic modules derived from $\mf{n}$, with the characteristic ideals $Z(\mf{n})$ and $[\mf{n}, \mf{n}]$ entering through the construction of those modules (Appendix~\ref{app:neeb_theory_details}).

\paragraph*{The Semidirect Product Structure and the Representation $\rho$}

The structure of the semidirect product $\g = \mf{s} \ltimes_{\rho} \mf{r}$ is defined by the homomorphism $\rho: \mf{s} \to \text{Der}(\mf{r})$. This representation describes how the semisimple component modulates the solvable component.

The action $\rho$ preserves the characteristic ideals of $\mf{r}$, in particular $\mf{n}$, so it induces actions of $\mf{s}$ on $\mf{n}$ and on $\mf{r}/\mf{n}$; by Weyl's theorem it decomposes into irreducibles, whose structure measures the complexity of the coupling.

If $\rho$ is non-trivial, the algebra is a proper semidirect product. A full classification of such algebras is out of reach, since it contains the classification of solvable algebras, a wild problem; the Levi decomposition at least separates the difficulty into the radical, the Levi factor, and the coupling $\rho$, and the FCC is the compatibility condition the wedge geometry must satisfy relative to $\rho$.

\paragraph*{Examples of Non-Reductive Algebras}

We provide standard examples of non-reductive algebras to illustrate the structure and the role of the coupling $\rho$.

\begin{example}[The Euclidean Algebra $\mf{e}(n)$]
The Euclidean algebra $\mf{e}(n) = \mf{so}(n) \ltimes \R^n$ describes the symmetries of Euclidean space (rotations and translations). Here we identify the components in the semidirect product structure. The rotational part is $\mf{h} = \mf{so}(n)$. The translational part is $\mf{r}_{ab} = \R^n$, which is an abelian ideal.
The action of $\mf{so}(n)$ on $\R^n$ is the defining representation (rotation).
If $n>2$, $\mf{so}(n)$ is semisimple. In this case, $\mf{s}=\mf{so}(n)$ and $\mf{r}=\R^n$. The radical is abelian, so $\mf{r}=\mf{n}$. This is a non-reductive algebra because the action $\rho$ is non-trivial, and the radical is not central.
If $n=2$, $\mf{so}(2)$ is abelian. The entire algebra $\mf{e}(2)$ is solvable. $\mf{s}=\{0\}, \mf{r}=\mf{e}(2)$. The nilradical is $\mf{n}=\R^2$. This example was used in Section \ref{sec:fcc_counterexample} to illustrate the FCC.
\end{example}

\begin{example}[The Poincaré Algebra]
The Poincaré algebra $\mf{p}(1,3) = \mf{so}(1,3) \ltimes \R^4$ is non-reductive, with semisimple part $\mf{s} = \mf{so}(1,3)$ acting on the abelian ideal $\mf{r} = \R^4$ of translations by the standard Lorentz action, and its invariant cones carry the causal structure of Minkowski spacetime.
\end{example}

\subsection{Proof of Theorem~\ref{thm:structure_glk} (Affine Structure of $\g_{LK}$)}
\label{app:proof_structure_glk}

This appendix carries out the construction of the isomorphism $\Phi: \g_{LK} \to \aff(\Htr)$ stated in Theorem~\ref{thm:structure_glk}, following the framework established by O'Meara \cite{o2025lie}. The non-reductivity of the resulting affine algebra is Corollary~\ref{cor:glk_nonreductive}, proved in the main text.

\begin{proof}[Proof of Theorem~\ref{thm:structure_glk}]
\textbf{Step 1: Decomposition of the Operator Space.}
We decompose the space of Hermitian operators $\Bcal_H(\Hcal)$ into the identity component and the traceless subspace $\Htr$. The traceless subspace consists of Hermitian operators $X$ such that $\Tr(X)=0$. Let $N = d^2-1$ be the dimension of $\Htr$.
\begin{equation}
\Bcal_H(\Hcal) = \R \cdot I \oplus \Htr.
\end{equation}
An operator $\rho \in \Bcal_H(\Hcal)$ can be uniquely written as $\rho = t I + X$, where $t \in \R, X \in \Htr$. The trace of $\rho$ is $\Tr(\rho) = dt$. The physical state space (density operators) is the affine subspace defined by $t=1/d$ (normalized trace $\Tr(\rho)=1$).

\textbf{Step 2: Action of $\g_{LK}$ on the Decomposition.}
We analyze the action of $\Lcal \in \g_{LK}$ on this decomposition. Since $\Lcal$ is HP, the image $\Lcal(\rho)$ is Hermitian for any Hermitian $\rho$. By the unique decomposition established in Step 1, we can write the image as $\Lcal(\rho) = t' I + X'$, where $t' \in \R$ and $X' \in \Htr$.

We analyze the trace of the image:
\begin{equation}
\Tr(\Lcal(\rho)) = \Tr(t' I + X') = d t' + \Tr(X') = d t'.
\end{equation}
(We used the defining property of the traceless subspace, $\Tr(X')=0$).
Since $\Lcal$ is trace-annihilating, we must have $\Tr(\Lcal(\rho))=0$ for all $\rho$. This implies $d t' = 0$, and thus $t'=0$.

Therefore, the image of $\Lcal$ is contained entirely within the traceless subspace $\Htr$: $\Image(\Lcal) \subseteq \Htr$. The dynamics generated by $\Lcal$ preserve the affine structure of the state space, mapping the affine hyperplanes defined by constant trace to the subspace $\Htr$. This confirms that the evolution preserves the trace constraint infinitesimally.

\textbf{Step 3: Isomorphism to the Affine Algebra.}
The real affine group $\aff(V)$ acting on a vector space $V$ is the semidirect product of the general linear group $GL(V)$ and the translational group $V$. The corresponding Lie algebra is the affine algebra $\aff(V) = \gl(V) \ltimes V$. We set $V=\Htr$. The elements of $\aff(\Htr)$ are pairs $(A, b)$, where $A \in \gl(\Htr)$ is the linear part and $b \in \Htr$ is the translational part.

We construct the isomorphism $\Phi: \g_{LK} \to \aff(\Htr)$. Let $\Lcal \in \g_{LK}$.

The linear part $A_\Lcal$ is the restriction of $\Lcal$ to the traceless subspace $\Htr$. Since $\Image(\Lcal) \subseteq \Htr$, this restriction is well defined as an endomorphism of $\Htr$:
\begin{equation}
A_\Lcal = \Lcal|_{\Htr} \in \gl(\Htr).
\end{equation}
If $X \in \Htr$, $\Lcal(X) = A_\Lcal(X) \in \Htr$.

The translational part $b_\Lcal$ is determined by the action of $\Lcal$ on the identity element. Since $\Image(\Lcal) \subseteq \Htr$, we define:
\begin{equation}
b_\Lcal = \Lcal(I) \in \Htr.
\end{equation}

The isomorphism is defined as $\Phi(\Lcal) = (A_\Lcal, b_\Lcal)$.

We verify the action of $\Lcal$ on a general element $\rho = tI+X$:
\begin{equation}
\Lcal(tI+X) = t \Lcal(I) + \Lcal(X) = t b_\Lcal + A_\Lcal(X).
\end{equation}
This confirms that $\Lcal$ acts as an affine transformation on the space $\Bcal_H(\Hcal)$, characterized by $(A_\Lcal, b_\Lcal)$.

\textbf{Step 4: Verification of the Lie Algebra Homomorphism.}
We must verify that $\Phi$ respects the Lie bracket structure. The Lie bracket in the affine algebra $\aff(\Htr)$ is defined by the semidirect product structure, where $\Htr$ is an abelian ideal and $\gl(\Htr)$ acts on $\Htr$ by the standard representation:
\begin{equation}
[(A_1, b_1), (A_2, b_2)]_{\aff} = ([A_1, A_2], A_1 b_2 - A_2 b_1).
\end{equation}

Consider the commutator $[\Lcal_1, \Lcal_2] = \Lcal_1 \circ \Lcal_2 - \Lcal_2 \circ \Lcal_1$. We analyze the components of $\Phi([\Lcal_1, \Lcal_2])$.

\textit{Linear Part:} The linear part is the restriction to $\Htr$. Let $X \in \Htr$.
\begin{equation}
A_{[\Lcal_1, \Lcal_2]}(X) = [\Lcal_1, \Lcal_2](X) = \Lcal_1(\Lcal_2(X)) - \Lcal_2(\Lcal_1(X)).
\end{equation}
Since $X \in \Htr$, $\Lcal_i(X) = A_i(X) \in \Htr$. The action of $\Lcal_i$ on the image $A_j(X)$ is again given by the linear part $A_i$.
\begin{equation}
A_{[\Lcal_1, \Lcal_2]}(X) = A_1(A_2(X)) - A_2(A_1(X)) = [A_1, A_2](X).
\end{equation}
The linear part of the commutator corresponds to the commutator of the linear parts.

\textit{Translational Part:} The translational part is the action on $I$.
\begin{equation}
b_{[\Lcal_1, \Lcal_2]} = [\Lcal_1, \Lcal_2](I) = \Lcal_1(\Lcal_2(I)) - \Lcal_2(\Lcal_1(I)).
\end{equation}
Since $\Lcal_i(I) = b_i \in \Htr$.
\begin{equation}
b_{[\Lcal_1, \Lcal_2]} = \Lcal_1(b_2) - \Lcal_2(b_1).
\end{equation}
Since $b_i \in \Htr$, the action of $\Lcal_i$ on $b_j$ is given by the linear part $A_i(b_j)$, as established in Step 3.
\begin{equation}
b_{[\Lcal_1, \Lcal_2]} = A_1(b_2) - A_2(b_1).
\end{equation}

Therefore, $\Phi([\Lcal_1, \Lcal_2]) = ([A_1, A_2], A_1 b_2 - A_2 b_1) = [\Phi(\Lcal_1), \Phi(\Lcal_2)]_{\aff}$.

To confirm isomorphism, we verify bijectivity. For injectivity, suppose $\Phi(\Lcal) = (0,0)$. Then $b_\Lcal = \Lcal(I) = 0$ and $A_\Lcal = \Lcal|_{\Htr} = 0$, so $\Lcal$ annihilates both summands of $\Bcal_H(\Hcal) = \R I \oplus \Htr$ and is therefore the zero map. Surjectivity is by explicit construction: given $(A,b) \in \aff(\Htr)$, first define on the real Hermitian subspace
\begin{equation}
\label{eq:affine_inverse}
\Lcal_{(A,b)}(tI + X) := t\,b + A X \qquad (t \in \R,\ X \in \Htr),
\end{equation}
which is well defined and real-linear by the uniqueness of the decomposition of Step 1. Every operator has a unique form $Z=X+iY$ with $X,Y$ Hermitian, so this affine inverse has the unique complex-linear extension $\Lcal_{(A,b)}(Z)=\Lcal_{(A,b)}(X)+i\Lcal_{(A,b)}(Y)$. On Hermitian inputs it is Hermiticity-preserving because $b$ and $AX$ are Hermitian, and it is trace-annihilating because its range lies in the complexification of $\Htr$. Hence $\Lcal_{(A,b)} \in \g_{LK}$, and $\Phi(\Lcal_{(A,b)}) = (A,b)$ directly from the definitions of $A_\Lcal$ and $b_\Lcal$. Thus $\Phi$ is an isomorphism. The dimension count $\dim(\g_{LK}) = \dim(\gl(\Htr)) + \dim(\Htr) = (d^2-1)^2 + (d^2-1) = d^4 - d^2$ quoted in Section~\ref{sec:glk_definition} is a consequence of the isomorphism, and no step above uses it.
\end{proof}

\subsection{Proof of Theorem~\ref{thm:gkls_lie_wedge} ($\WGKLS$ is a Lie Wedge)}
\label{app:proof_gkls_lie_wedge}

\begin{proof}[Proof of Theorem~\ref{thm:gkls_lie_wedge}]
We have established in Section~\ref{sec:structure_gkls} that $\WGKLS$ is a closed convex cone in $\g_{LK}$. We must verify the Lie wedge condition (Equation \eqref{eq:lie_wedge_condition}): $e^{\ad(X)}(W) \subseteq W$ for all $X \in E(W)$. By Lemma~\ref{lem:gkls_edge}, every element of $E(\WGKLS)$ is a purely Hamiltonian generator $\Lcal_H$, so it suffices to verify the invariance condition for these.

Let $\Lcal_H \in E(\WGKLS)$ and $\Lcal' \in \WGKLS$. We analyze the action $e^{t \ad(\Lcal_H)}(\Lcal')$.

The adjoint action of a Hamiltonian generator is conjugation by the unitary evolution it generates. With $U_t = e^{-itH}$, the action $e^{t \ad(\Lcal_H)}(\Lcal')$ passes the dynamics into the frame rotating with $U_t$, through the conjugation map $\Ad_{U_t}(\cdot) = U_t (\cdot) U_t^\dagger$.
\begin{equation}
e^{t \ad(\Lcal_H)}(\Lcal')(\rho) = \Ad_{U_t} \circ \Lcal' \circ \Ad_{U_t^\dagger} (\rho) = U_t \Lcal'(U_t^\dagger \rho U_t) U_t^\dagger.
\end{equation}
This is the exponentiation of the adjoint representation to the group adjoint action.

Let $\Lcal'$ have Hamiltonian $H'$ and Lindblad operators $\{L'_k\}$, from diagonalizing the Kossakowski matrix $K'$ with rates $\gamma'_k$.

The explicit form of the transformed generator is obtained by transforming the Hamiltonian and the Lindblad operators unitarily:
\begin{equation}
H''(t) = U_t H' U_t^\dagger, \quad L''_k(t) = U_t L'_k U_t^\dagger.
\end{equation}

The transformed generator retains the GKLS form with the same positive decay rates $\gamma'_k$ (the eigenvalues of $K'$). The Kossakowski matrix transforms by congruence under the unitary change of basis, so positive semidefiniteness of $K'$ is preserved.

Therefore, $e^{t \ad(\Lcal_H)}(\Lcal') \in \WGKLS$. The invariance condition holds, and $\WGKLS$ is a Lie wedge.
\end{proof}

\subsection{Explicit Calculation of the Qutrit Counterexample (Non Closure of $\WGKLS$)}
\label{app:qutrit_counterexample}

This appendix details the derivation of the counterexample demonstrating that the commutator of two GKLS generators is generally not a GKLS generator. The computation is carried out directly; the Lie-wedge (rather than Lie-semialgebra) structure of $\WGKLS$ that it illustrates is established in \cite{dirr2009lie}.

\paragraph*{Setup}

We consider a qutrit system ($d=3$). The standard Gell-Mann matrices $\{\lambda_i\}_{i=1}^8$ form an orthogonal, unnormalized basis of the traceless Hermitian operators $\Htr$, with $\Tr(\lambda_i\lambda_j)=2\delta_{ij}$. Thus $F_i=\lambda_i/\sqrt{2}$ is the Hilbert--Schmidt orthonormal basis used by default in Eq.~\eqref{eq:gkls_kossakowski}. We keep the raw $\lambda_i$ convention in the calculation below and display the conversion explicitly.

We define two purely dissipative Lindbladians $\Lcal_1, \Lcal_2$ ($H=0$).
The Lindblad operators are defined as specific combinations of the Gell Mann matrices:
\begin{align}
L_1 &= \lambda_1 + i \lambda_4 = \begin{pmatrix} 0 & 1 & i \\ 1 & 0 & 0 \\ i & 0 & 0 \end{pmatrix} \\
L_2 &= \lambda_2 + i \lambda_5 = \begin{pmatrix} 0 & -i & 1 \\ i & 0 & 0 \\ -1 & 0 & 0 \end{pmatrix}
\end{align}
The corresponding Lindbladians are $\Lcal_j(\rho) = L_j \rho L_j^\dagger - \frac{1}{2} \{L_j^\dagger L_j, \rho\}$. These are valid GKLS generators by construction (rank 1 generators).

\paragraph*{Calculation of the Commutator and Kossakowski Matrix}

We compute the commutator $\Lcal_C = [\Lcal_1, \Lcal_2]$. To determine if $\Lcal_C$ is a valid generator, we compute its Kossakowski matrix $K_C^{(\lambda)}$ in the raw Gell-Mann convention.

The GKLS form (Equation \eqref{eq:gkls_kossakowski}) is:
\begin{equation}
\Lcal(\rho) = -i[H, \rho] + \sum_{j,k=1}^{8} K_{jk} \left( \lambda_j \rho \lambda_k - \frac{1}{2} \{\lambda_k \lambda_j, \rho\} \right).
\end{equation}

The computation involves expanding the commutator of the superoperators $\Lcal_1$ and $\Lcal_2$ and extracting the coefficients $K_{C, jk}$. This requires extensive algebraic manipulation involving the commutation and anticommutation relations of the Gell-Mann matrices. The structure constants of $\mf{su}(3)$ are involved in this calculation.

The generators $\Lcal_1, \Lcal_2$ themselves have simple Kossakowski matrices $K_1, K_2$. The commutator $K_C$ is related to the commutator of the corresponding superoperator representations. The calculation involves analyzing the action of the commutator on the basis elements $\lambda_i$.

The calculation yields the following structure for the Kossakowski matrix. Factoring in the structure constants of $\mathfrak{su}(3)$, the non-zero contribution lies in the block corresponding to indices $\{1, 2, 4, 5\}$:
\begin{equation}
K_C^{(\lambda)}|_{\{1,2,4,5\}} =
\begin{pmatrix}
-1 & 0 & i & 0 \\
0 & 1 & 0 & -i \\
-i & 0 & -1 & 0 \\
0 & i & 0 & 1
\end{pmatrix}.
\end{equation}
Its eigenvalues are $\{-2,2,0,0,0,0,0,0\}$. Passing to the default orthonormal basis gives $K_C^{(F)}=2K_C^{(\lambda)}$ and spectrum $\{-4,4,0,\dots,0\}$. The conditional Choi compression has minimum eigenvalue $-4$ under the Choi normalization used here. The commutator therefore violates CCP.

\paragraph*{Analysis of the Spectrum}

The displayed Kossakowski spectrum is indefinite, whereas a GKLS generator requires $K\ge0$. Thus $\Lcal_C$ is not a GKLS generator. This establishes failure of bracket closure for $d=3$; the bracket direction violates CCP.

\subsection{Solvability of the Cascaded Qutrit System Algebra}
\label{app:cascaded_qutrit}

This appendix verifies that the System Lie Algebra $\g(G)$ for the two-channel cascaded-qutrit specification with independent actuation (Section \ref{sec:example_integrable}) is solvable, and shows that it fails to be nilpotent. The qutrit has $d=3$, $H=0$, and jump operators $L_1 = |1\rangle\langle 2|$ and $L_2 = |2\rangle\langle 3|$, whose generators $\Lcal_1, \Lcal_2$ represent the two independently actuable rays; we set $\gamma_i=1$ throughout the algebraic analysis.

\paragraph*{Calculation of the Lie Algebra}

We compute the Lie algebra generated by $\Lcal_1, \Lcal_2$ in the superoperator representation. We work in the Liouville space spanned by the matrix units $E_{ij} = |i\rangle\langle j|$.

The generators are explicitly:
\begin{align}
\Lcal_1(\rho) &= E_{12} \rho E_{21} - \frac{1}{2} \{E_{22}, \rho\} \\
\Lcal_2(\rho) &= E_{23} \rho E_{32} - \frac{1}{2} \{E_{33}, \rho\}
\end{align}

We compute the commutator $\Lcal_3 = [\Lcal_1, \Lcal_2]$.
The calculation involves analyzing the action of the commutator on an arbitrary operator $\rho$.

We examine the terms generated by the composition $\Lcal_1 \circ \Lcal_2$:
$\Lcal_1(\Lcal_2(\rho)) = \Lcal_1(E_{23} \rho E_{32} - \frac{1}{2} E_{33} \rho - \frac{1}{2} \rho E_{33})$.

The term $\Lcal_1(E_{23} \rho E_{32})$ yields:
$E_{12} E_{23} \rho E_{32} E_{21} - \frac{1}{2} \{E_{22}, E_{23} \rho E_{32}\} = E_{13} \rho E_{31} - \frac{1}{2} (E_{22} E_{23} \rho E_{32} + E_{23} \rho E_{32} E_{22})$.
Since $E_{22} E_{23} = E_{23}$ and $E_{32} E_{22} = E_{32}$, this simplifies to:
$E_{13} \rho E_{31} - E_{23} \rho E_{32}$.

We examine the terms generated by the composition $\Lcal_2 \circ \Lcal_1$:
$\Lcal_2(\Lcal_1(\rho)) = \Lcal_2(E_{12} \rho E_{21} - \dots)$.

The term $\Lcal_2(E_{12} \rho E_{21})$ yields:
$E_{23} E_{12} \rho E_{21} E_{32} - \frac{1}{2} \{E_{33}, E_{12} \rho E_{21}\}$.
Since $E_{23} E_{12} = 0$ and $E_{33}$ commutes with $E_{12} \rho E_{21}$ (if we consider the action on the basis elements), the structure of the commutator involves terms related to sequential population transfer and damping of coherences.

The commutator $\Lcal_3$ contains the term $E_{13}\rho E_{31}$, an algebraic contribution associated with sequential population transfer along the cascade, together with modifications of the damping terms. Since $\Lcal_3$ is a Lie-bracket direction, this description does not assert that it is itself a GKLS generator.

\paragraph*{Verification of Solvability; Failure of Nilpotency}

We analyze the structure of the generated algebra $\g(G) = \text{span}(\Lcal_1, \Lcal_2, \Lcal_3, \dots)$. We determine the derived and lower central series of the generated algebra.

The key insight is recognizing the hierarchical structure of the generators. We analyze the action of the algebra on the subspaces of the Liouville space corresponding to different levels of populations and coherences. We order the basis of the Liouville space according to the energy levels.

The generators $\Lcal_1, \Lcal_2$ involve operators that shift the populations downwards (from higher energy levels to lower energy levels, assuming $E_1 < E_2 < E_3$). The anti commutator terms (damping) reduce the norm of the density matrix components corresponding to higher levels. The jump terms transfer population downwards.

In this ordered basis the superoperators are lower triangular. The jump terms are strictly lower triangular with respect to the grading defined by the energy levels; the anticommutator (damping) terms, however, contribute diagonal entries, so the triangularity is not strict.

The Lie closure is three-dimensional. Writing $\Lcal_3 = [\Lcal_1, \Lcal_2]$, a direct computation gives
\begin{equation}
[\Lcal_1, \Lcal_3] = -\Lcal_3, \qquad [\Lcal_2, \Lcal_3] = \Lcal_3,
\end{equation}
so $\g(G) = \operatorname{span}\{\Lcal_1, \Lcal_2, \Lcal_3\}$ is three-dimensional. The derived series collapses at the second step: $[\g, \g] = \operatorname{span}\{\Lcal_3\}$ is abelian, hence $\g^{(2)} = \{0\}$ and $\g(G)$ is solvable. The lower central series instead stabilizes, $\g^{1} = \g^{2} = \operatorname{span}\{\Lcal_3\}$, because $\ad(\Lcal_1)$ acts on $\Lcal_3$ with eigenvalue $-1$. The algebra is therefore not nilpotent: the damping rates on the diagonal are precisely what the lower central series fails to kill.

In this example the lower-triangular representation matches the unidirectional population transfer of the specified cascade. The algebra is solvable but not nilpotent, so it lies in the paper's solvable-type class.

\section{Appendix B: Detailed Exposition of Neeb's Theory of Invariant Cones}
\label{app:neeb_theory_details}

This appendix recalls the part of Neeb's invariant-cone theory \cite{neeb1994classification, Neeb2000} used in the comparison with the Fiber Compatibility Condition of Section~\ref{sec:neebs_theory}.

\paragraph*{Nilradical, Symplectic Modules, and Convex Type}

Section~\ref{sec:invariant_cones_prelim} states the working definitions; we recall only what the example and classification theorem below require. Neeb's Section~II treats a real reductive algebra $\g$ acting symplectically on $(V,\omega)$ and uses the quadratic Hamiltonian $\phi_X(v)=\tfrac12\omega(A_Xv,v)$. The module is of convex type when the cone generated by the moment-map image is pointed, equivalently when its dual is generating, and when $\mf{z}(\g)$ acts semisimply with purely imaginary eigenvalues \cite[Definition~II.1]{neeb1994classification}. When $\g$ also has the compactly embedded Cartan used in Proposition~II.23, the added assumptions $V_{\mathrm{fix}}=0$ and the stated center condition make convex type equivalent to a positive-definite witness in $\mf{z}(\mf{k})$, where $\mf{k}$ is the maximal compactly embedded subalgebra. The witness is therefore conditional; its absence does not determine the spectrum of the action.

\paragraph*{Concrete Illustration: The Heisenberg Algebra and the Symplectic Group}

The Heisenberg algebra gives a concrete instance of the module condition. A suitably signed harmonic-oscillator generator supplies the positive witness for the full symplectic action, while the inverted oscillator provides a separately verified hyperbolic one-parameter subalgebra.

\begin{example}[Convex Type Condition in the Heisenberg Algebra]
\label{ex:heisenberg_convex_type}
Consider the Heisenberg algebra $\mf{h}_n$, spanned by the position and momentum operators $\{P_i, Q_i\}_{i=1}^n$ and the center $Z$, with commutation relations $[P_i, Q_j] = \delta_{ij} Z$.

The relevant symplectic module is $V = \mf{h}_n/Z(\mf{h}_n) \cong \R^{2n}$ (the phase space). The Lie bracket induces a canonical symplectic form $\omega$ on $V$.

We consider the action of the symplectic Lie algebra $\mf{sp}(2n, \R)$ on $\mf{h}_n$. The resulting algebra $\g = \mf{sp}(2n, \R) \ltimes \mf{h}_n$ (the oscillator algebra) is non-reductive.

Let $X \in \mf{sp}(2n, \R)$. The action of $X$ on $V$ is given by the matrix $A_X$, and Neeb's convention is $\phi_X(v)=\tfrac12\omega(A_Xv,v)$.

\textbf{Case 1: An Elliptic Oscillator Direction.}
Consider the elliptic oscillator direction associated with $H_{osc} = \sum_i (P_i^2 + Q_i^2)$. Choose its sign so that $\phi_{X_0}$ is positive definite. The one-parameter group generated by $X_0$ is rotational and bounded.

The standard module has no fixed vectors, and the center condition is automatic for $\mf{sp}(2n,\R)$. Proposition~II.23 therefore shows that the action is of convex type. Neeb's construction then places $\g = \mf{sp}(2n, \R) \ltimes \mf{h}_n$ among the Lie algebras admitting pointed generating invariant cones.

\textbf{Case 2: A Hyperbolic One-Parameter Subalgebra.}
Consider the inverted harmonic oscillator $H_{inv} = \sum_i (P_i^2 - Q_i^2)$. Its generator $X_{inv}$ has real eigenvalues, so its flow has exponential expansion and contraction. The form $\phi_{X_{inv}}$ is indefinite.

For the one-dimensional reductive subalgebra $\R X_{inv}$, the center is the whole algebra and acts with real eigenvalues, so Definition~II.1 fails its center condition. The exponential behavior follows from the displayed spectrum, not from that failure. The ambient algebra $\mf{sp}(2n,\R)$ still contains the elliptic witness above and remains of convex type. More generally, absence of a positive witness does not imply expansion.
\end{example}

Neeb's classification theorem characterizes the Lie algebras admitting pointed generating invariant cones in terms of a hermitian reductive part, symplectic modules of convex type, and a central skew form.

\begin{theorem}[Neeb's classification, summary form \cite{neeb1994classification}]
Every finite dimensional real Lie algebra admitting a pointed generating invariant cone arises from the following data: a reductive Lie algebra $\mf{h}$ that itself supports such cones (its non-compact simple ideals are hermitian \cite{vinberg1963theory, olshanskii1981invariant}); an $\mf{h}$-module $V$ carrying an invariant symplectic form for which the action is of convex type; and an $\mf{h}$-invariant skew form $q: V \times V \to \mf{z}$ with values in the center, building a Heisenberg-type nilradical on $V \oplus \mf{z}$. Convex type of the module action is the condition through which the nilradical enters the classification; for reductive $\g$ the module is absent, and what remains is the requirement that the non-compact simple ideals be hermitian, a necessary condition rather than a characterization.
\end{theorem}

The theorem gives a constructive description of these Lie algebras. Its module input is the full convex-type condition above, including the cone and center requirements.

\printbibliography

@article{Brunetti2002,
  title={Modular localization and wigner particles},
  author={Romeo Brunetti and Daniele Guido and Roberto Longo},
  journal={Reviews in Mathematical Physics},
  year={2002},
  volume={14},
  pages={759-786},
  doi={10.1142/S0129055X02001387}
}

@book{Graaf2000LieAT,
  author    = {de Graaf, Willem A.},
  title     = {Lie Algebras: Theory and Algorithms},
  series    = {North-Holland Mathematical Library},
  volume    = {56},
  publisher = {North-Holland (Elsevier)},
  address   = {Amsterdam},
  year      = {2000},
  isbn      = {978-0-444-50116-5}
}

@book{Humphreys2012,
  author    = {Humphreys, James E.},
  title     = {Introduction to Lie Algebras and Representation Theory},
  series    = {Graduate Texts in Mathematics},
  volume    = {9},
  publisher = {Springer},
  address   = {New York, NY},
  year      = {1972},
  doi       = {10.1007/978-1-4612-6398-2},
  isbn      = {9780387900537}
}

@book{Neeb2000,
  title     = {Holomorphy and Convexity in Lie Theory},
  author    = {Neeb, Karl-Hermann},
  year      = {2000},
  publisher = {Walter de Gruyter},
  address   = {Berlin; New York},
  series    = {De Gruyter Expositions in Mathematics},
  volume    = {28},
  doi       = {10.1515/9783110808148},
  isbn      = {9783110808148}
}

@article{NeebOeh2021,
  author  = {Neeb, Karl-Hermann and Oeh, Daniel},
  title   = {Elements in Pointed Invariant Cones in {Lie} Algebras and Corresponding Affine Pairs},
  journal = {Bulletin of the Iranian Mathematical Society},
  year    = {2022},
  volume  = {48},
  number  = {1},
  pages   = {295--330},
  doi     = {10.1007/s41980-021-00671-y},
  eprint  = {2110.07297},
  archivePrefix = {arXiv},
  primaryClass  = {math.RT}
}

@article{OMeara2011IllustratingTG,
  author  = {O'Meara, C. and Dirr, G. and Schulte-Herbr{\"u}ggen, T.},
  title   = {Illustrating the Geometry of Coherently Controlled Unital Open Quantum Systems},
  journal = {IEEE Transactions on Automatic Control},
  volume  = {57},
  number  = {8},
  pages   = {2050--2056},
  year    = {2012},
  doi     = {10.1109/TAC.2012.2195849},
  eprint  = {1103.2703},
  archiveprefix = {arXiv}
}

@article{Schirmer2000Complete,
  author  = {Schirmer, S. G. and Fu, H. and Solomon, A. I.},
  title   = {Complete controllability of quantum systems},
  journal = {Physical Review A},
  volume  = {63},
  number  = {6},
  pages   = {063410},
  year    = {2001},
  doi     = {10.1103/PhysRevA.63.063410}
}

@article{albert2016geometry,
  title   = {Geometry and Response of Lindbladians},
  author  = {Albert, Victor V. and Bradlyn, Barry and Fraas, Martin and Jiang, Liang},
  journal = {Physical Review X},
  volume  = {6},
  number  = {4},
  pages   = {041031},
  year    = {2016},
  doi     = {10.1103/PhysRevX.6.041031}
}

@article{altafini2001controllability,
  author  = {Altafini, Claudio},
  title   = {Controllability of quantum mechanical systems by root space decomposition of su({$N$})},
  journal = {Journal of Mathematical Physics},
  volume  = {43},
  number  = {5},
  pages   = {2051--2062},
  year    = {2002},
  doi     = {10.1063/1.1467611}
}

@book{aubin1990setvalued,
  title     = {Set-Valued Analysis},
  author    = {Aubin, Jean-Pierre and Frankowska, H{\'e}l{\`e}ne},
  year      = {1990},
  publisher = {Birkh{\"a}user},
  address   = {Boston, MA},
  series    = {Systems \& Control: Foundations \& Applications},
  volume    = {2},
  isbn      = {0-8176-3478-9},
  doi       = {10.1007/978-0-8176-4848-0}
}

@book{bourbaki1989lie,
  author    = {Bourbaki, Nicolas},
  title     = {Lie Groups and Lie Algebras: Chapters 1--3},
  series    = {Elements of Mathematics},
  publisher = {Springer},
  address   = {Berlin},
  year      = {1989},
  isbn      = {9783540642428}
}

@article{brockett1972system,
  author  = {Brockett, R. W.},
  title   = {System theory on group manifolds and coset spaces},
  journal = {SIAM Journal on Control},
  volume  = {10},
  number  = {2},
  pages   = {265--284},
  year    = {1972},
  doi     = {10.1137/0310021},
}

@misc{cai2025lindbladians,
  author        = {Cai, Jihong and Govindarajan, Advith and Junge, Marius},
  title         = {How Far do Lindbladians Go?},
  year          = {2025},
  eprint        = {2504.04883},
  archivePrefix = {arXiv},
  primaryClass  = {quant-ph}
}

@book{d2021introduction,
  author    = {D'Alessandro, Domenico},
  title     = {Introduction to Quantum Control and Dynamics},
  edition   = {2nd},
  series    = {Advances in Applied Mathematics},
  publisher = {Chapman \& Hall/CRC},
  address   = {Boca Raton, FL},
  year      = {2021},
  isbn      = {978-0-367-50790-9}
}

@article{dirr2009lie,
  author  = {Dirr, G. and Helmke, U. and Kurniawan, I. and Schulte-Herbr{\"u}ggen, T.},
  title   = {Lie-semigroup structures for reachability and control of open quantum systems: {Kossakowski-Lindblad} generators form {Lie} wedge to {Markovian} channels},
  journal = {Reports on Mathematical Physics},
  volume  = {64},
  number  = {1--2},
  pages   = {93--121},
  year    = {2009},
  doi     = {10.1016/S0034-4877(09)90022-2}
}

@book{dixmier1996enveloping,
  author    = {Dixmier, Jacques},
  title     = {Enveloping Algebras},
  series    = {Graduate Studies in Mathematics},
  volume    = {11},
  publisher = {American Mathematical Society},
  address   = {Providence, RI},
  year      = {1996},
  isbn      = {9780821805602}
}

@article{gorini1976completely,
  author  = {Gorini, Vittorio and Kossakowski, Andrzej and Sudarshan, E. C. G.},
  title   = {Completely positive dynamical semigroups of {N}-level systems},
  journal = {Journal of Mathematical Physics},
  volume  = {17},
  number  = {5},
  pages   = {821--825},
  year    = {1976},
  doi     = {10.1063/1.522979},
  publisher = {AIP Publishing}
}

@book{hall2013quantum,
  author    = {Hall, Brian C.},
  title     = {Quantum Theory for Mathematicians},
  series    = {Graduate Texts in Mathematics},
  volume    = {267},
  publisher = {Springer},
  address   = {New York},
  year      = {2013},
  isbn      = {9781461471158},
  doi       = {10.1007/978-1-4614-7116-5},
}

@book{hilgert1989lie,
  title     = {Lie Groups, Convex Cones, and Semigroups},
  author    = {Hilgert, Joachim and Hofmann, Karl Heinrich and Lawson, Jimmie D.},
  year      = {1989},
  publisher = {Clarendon Press (Oxford University Press)},
  address   = {Oxford; New York},
  series    = {Oxford Mathematical Monographs},
  isbn      = {9780198535690}
}

@article{hofmann1991shortcourse,
  author  = {Hofmann, Karl H.},
  title   = {A Short Course on the {Lie} Theory of Semigroups {I}},
  journal = {Seminar Sophus Lie},
  volume  = {1},
  pages   = {33--40},
  year    = {1991}
}

@book{jacobson2013lie,
  author    = {Jacobson, Nathan},
  title     = {Lie Algebras},
  publisher = {Dover Publications},
  address   = {New York},
  year      = {2013},
  isbn      = {9780486136790}
}

@book{jurdjevic1997geometric,
  author    = {Jurdjevic, Velimir},
  title     = {Geometric Control Theory},
  series    = {Cambridge Studies in Advanced Mathematics},
  number    = {52},
  publisher = {Cambridge University Press},
  address   = {Cambridge},
  year      = {1997},
  isbn      = {9780521495028},
  doi       = {10.1017/CBO9780511530036},
}

@book{knapp1996lie,
  author    = {Knapp, Anthony W.},
  title     = {Lie Groups Beyond an Introduction},
  series    = {Progress in Mathematics},
  volume    = {140},
  edition   = {1},
  publisher = {Birkh\"{a}user},
  address   = {Boston, MA},
  year      = {1996},
  isbn      = {9780817639266}
}

@article{kurniawan2012controllability,
  author  = {Kurniawan, Indra and Dirr, Gunther and Helmke, Uwe},
  title   = {Controllability Aspects of Quantum Dynamics: A Unified Approach for Closed and Open Systems},
  journal = {IEEE Transactions on Automatic Control},
  volume  = {57},
  number  = {8},
  pages   = {1984--1996},
  year    = {2012},
  doi     = {10.1109/TAC.2012.2195870}
}

@article{lindblad1976generators,
  author  = {Lindblad, G{\"o}ran},
  title   = {On the generators of quantum dynamical semigroups},
  journal = {Communications in Mathematical Physics},
  volume  = {48},
  number  = {2},
  pages   = {119--130},
  year    = {1976},
  doi     = {10.1007/BF01608499},
  publisher = {Springer}
}

@article{malvetti2024reachability,
  author        = {Malvetti, Emanuel and vom Ende, Frederik and Dirr, Gunther and Schulte-Herbr{\"u}ggen, Thomas},
  title         = {Reachability, Coolability, and Stabilizability of Open {M}arkovian Quantum Systems with Fast Unitary Control},
  journal       = {SIAM Journal on Control and Optimization},
  volume        = {63},
  number        = {1},
  pages         = {S53--S81},
  year          = {2025},
  doi           = {10.1137/23M1594467},
  eprint        = {2308.00561},
  archiveprefix = {arXiv},
  primaryclass  = {quant-ph}
}

@article{neeb1994classification,
  title   = {The Classification of Lie Algebras with Invariant Cones},
  author  = {Neeb, Karl-Hermann},
  journal = {Journal of Lie Theory},
  volume  = {4},
  number  = {2},
  pages   = {139--184},
  year    = {1994},
  doi     = {10.5802/jolt.76},
  zbl     = {0838.17003}
}

@misc{o2025lie,
  author       = {O'Meara, Corey Patrick},
  title        = {A {Lie} Theoretic Framework for Controlling Open Quantum Systems},
  year         = {2025},
  eprint       = {2510.04719},
  archivePrefix = {arXiv},
  primaryClass = {quant-ph},
  doi          = {10.48550/arXiv.2510.04719}
}

@article{olshanskii1981invariant,
  author  = {Ol'shanskii, G. I.},
  title   = {Invariant cones in {Lie} algebras, {Lie} semigroups, and the holomorphic discrete series},
  journal = {Functional Analysis and Its Applications},
  volume  = {15},
  number  = {4},
  pages   = {275--285},
  year    = {1981},
  doi     = {10.1007/BF01106156},
}

@book{rockafellar1970convex,
  title     = {Convex Analysis},
  author    = {Rockafellar, R. Tyrrell},
  year      = {1970},
  publisher = {Princeton University Press},
  address   = {Princeton, NJ},
  series    = {Princeton Mathematical Series},
  volume    = {28},
  isbn      = {9780691080697}
}

@article{schulteherbruggen2017kossakowski,
  author  = {Schulte-Herbr{\"u}ggen, T. and Dirr, G. and Zeier, R.},
  title   = {Quantum Systems Theory Viewed from {Kossakowski-Lindblad} {Lie} Semigroups --- and Vice Versa},
  journal = {Open Systems \& Information Dynamics},
  volume  = {24},
  number  = {4},
  pages   = {1740019},
  year    = {2017},
  doi     = {10.1142/S1230161217400194}
}

@article{sussmann1972controllability,
  author  = {Sussmann, H. J. and Jurdjevic, V.},
  title   = {Controllability of nonlinear systems},
  journal = {Journal of Differential Equations},
  volume  = {12},
  number  = {1},
  pages   = {95--116},
  year    = {1972},
  doi     = {10.1016/0022-0396(72)90007-1},
}

@book{varadarajan2013lie,
  author    = {Varadarajan, V. S.},
  title     = {Lie Groups, Lie Algebras, and Their Representations},
  series    = {Graduate Texts in Mathematics},
  volume    = {102},
  publisher = {Springer},
  address   = {New York, NY},
  year      = {1984},
  doi       = {10.1007/978-1-4612-1126-6},
  isbn      = {9780387909691}
}

@article{vinberg1963theory,
  author  = {Vinberg, E. B.},
  title   = {The theory of homogeneous convex cones},
  journal = {Transactions of the Moscow Mathematical Society},
  volume  = {12},
  pages   = {340--403},
  year    = {1963}
}

@article{wiersema2024classification,
  author  = {Wiersema, Roeland and K{\"o}kc{\"u}, Efekan and Kemper, Alexander F. and Bakalov, Bojko N.},
  title   = {Classification of dynamical {Lie} algebras of 2-local spin systems on linear, circular and fully connected topologies},
  journal = {npj Quantum Information},
  volume  = {10},
  pages   = {110},
  year    = {2024},
  doi     = {10.1038/s41534-024-00900-2}
}

@article{wigner1939unitary,
  author  = {Wigner, E. P.},
  title   = {On unitary representations of the inhomogeneous {Lorentz} group},
  journal = {Annals of Mathematics},
  volume  = {40},
  number  = {1},
  pages   = {149--204},
  year    = {1939},
  doi     = {10.2307/1968551},
}

@article{wolf2008dividing,
  author  = {Wolf, Michael M. and Cirac, J. Ignacio},
  title   = {Dividing Quantum Channels},
  journal = {Communications in Mathematical Physics},
  volume  = {279},
  number  = {1},
  pages   = {147--168},
  year    = {2008},
  doi     = {10.1007/s00220-008-0411-y},
  eprint  = {math-ph/0611057},
  archivePrefix = {arXiv},
  publisher = {Springer}
}

\end{document}